\documentclass[letter, 10pt]{article}
\usepackage[T1]{fontenc}
\usepackage{tgpagella}
\usepackage{graphicx,amsmath,amsfonts,amscd,amssymb,bm,url,color,latexsym,amsthm}
\usepackage{fullpage}
\usepackage[small,bf]{caption}
\usepackage{subcaption}
\usepackage{bbm}
\usepackage{microtype}
\usepackage{algorithm,algorithmicx,algpseudocode} 
\usepackage{hyperref}
\usepackage{verbatim}
\usepackage{framed}
\usepackage{graphicx,sidecap,wrapfig}
\usepackage[nameinlink]{cleveref}
\usepackage{tikz}
\usepackage{fge}
\usepackage{pgfplots}
\usepackage{mathdots}

\allowdisplaybreaks  

\hypersetup{
    colorlinks=true,%
    citecolor=blue,%
    filecolor=blue,%
    linkcolor=blue,%
    urlcolor=blue
}
\usepackage[toc, page]{appendix}

\newtheorem{theorem}{Theorem}[section]
\newtheorem{lemma}[theorem]{Lemma}
\newtheorem{corollary}[theorem]{Corollary}
\newtheorem{proposition}[theorem]{Proposition}
\newtheorem{definition}[theorem]{Definition}

\renewcommand{\mathbf}{\boldsymbol} 
\newcommand{\mb}{\mathbf}
\newcommand{\mc}{\mathcal}

\newcommand{\bb}{\mathbb}
\newcommand{\set}[1]{\left\{ #1 \right\}}

\newcommand{\eps}{\varepsilon}
\newcommand{\R}{\bb R}

\newcommand{\indicator}[1]{\mathbbm 1_{#1}}

\newcommand{ \brac }[1]{\left[ #1 \right]}
\newcommand{ \paren }[1]{ \left( #1 \right) }

\DeclareMathOperator{\trace}{tr}
\DeclareMathOperator{\supp}{supp}

\DeclareMathOperator{\diag}{diag}
\DeclareMathOperator{\offdiag}{offdiag}
\DeclareMathOperator{\poly}{poly}

\DeclareMathOperator{\grad}{grad}
\DeclareMathOperator{\Hess}{Hess}

\DeclareMathOperator{\st}{s.t.}

\newcommand{\event}{\mc E}

\newcommand{\deltagrad}{\mb \delta_{\mathrm{grad}}}
\newcommand{\dbgrad}{\bar{\mb \delta}_{\mathrm{grad}}}

\newcommand{\norm}[2]{\left\| #1 \right\|_{#2}}
\newcommand{\abs}[1]{\left| #1 \right|}
\newcommand{\innerprod}[2]{\left\langle #1,  #2 \right\rangle}
\newcommand{\prob}[1]{\bb P\left[ #1 \right]}
\newcommand{\expect}[1]{\bb E\left[ #1 \right]}

\newcommand{\cconv}{\circledast}
\newcommand{\extend}[1]{\widetilde{#1}}
\newcommand{\injector}{\mb \iota}
\newcommand{\shift}[2]{s_{#2}[#1]}
\newcommand{\proj}[2]{\mathcal{P}_{#2}\left[ #1 \right]}
\newcommand{\simiid}{\sim_{\mathrm{i.i.d.}}}

\newcommand{\qinit}{\mb q_{\mathrm{init}}}
\newcommand{\zetainit}{\mb \zeta_{\mathrm{init}}}
\newcommand{\zetahatinit}{\hat{\mb \zeta}_{\mathrm{init}}}

\numberwithin{equation}{section}

\def \endprf{\hfill {\vrule height6pt width6pt depth0pt}\medskip}
\renewenvironment{proof}{\noindent {\bf Proof} }{\endprf\par}

\begin{document}

\title{Structured Local Optima in Sparse Blind Deconvolution}
\author{Yuqian Zhang, Han-Wen Kuo, John Wright \\ 
Department of Electrical Engineering and Data Science Institute\\
Columbia University}
\maketitle

\begin{abstract}
Blind deconvolution is a ubiquitous problem of recovering two unknown signals from their convolution. Unfortunately, this is an ill-posed problem in general. This paper focuses on the {\em short and sparse} blind deconvolution problem, where the one unknown signal is short and the other one is sparsely and randomly supported. This variant captures the structure of the unknown signals in several important applications. We assume the short signal to have unit $\ell^2$ norm and cast the blind deconvolution problem as a nonconvex optimization problem over the sphere. We demonstrate that (i) in a certain region of the sphere, every local optimum is close to some shift truncation of the ground truth, and (ii) for a generic short signal of length $k$, when the sparsity of activation signal $\theta\lesssim k^{-2/3}$ and number of measurements $m\gtrsim\poly\paren{k}$, a simple initialization method together with a descent algorithm which escapes strict saddle points recovers a near shift truncation of the ground truth kernel. 
\end{abstract}

\section{Introduction}
\label{sec:intro}

Blind deconvolution is the problem of recovering two unknown signals $\mb a_0$ and $\mb x_0$ from their convolution $\mb y = \mb a_0 \ast \mb x_0$. This fundamental problem recurs across several fields, including astronomy, microscopy data processing \cite{Cheung17-Nature}, neural spike sorting \cite{Lewicki1998}, computer vision \cite{Kundur1996-SPM}, etc.  However, this problem is ill-posed without further priors on the unknown signals, as there are infinitely many pairs of signals $(\mb a, \mb x)$ whose convolution equals a given observation $\mb y$. Fortunately, in practice, the target signals $(\mb a, \mb x )$ are often structured. In particular, a number of practical applications exhibit a common {\em short-and-sparse} structure: 

In {\em Neural spike sorting}: Neurons in the brain fire brief voltage spikes when stimulated. The signatures of the spikes encode critical features of the neuron and the occurrence of such spikes are usually sparse and random in time \cite{Lewicki1998, Ekanadham2011-NIPS}. 

In {\em Microscopy data analysis}: The nanoscale materials of interests are contaminated by randomly and sparsely distributed ``defects'', which can dramatically change the electronic structure of the material \cite{Cheung17-Nature}.  

In {\em Image deblurring}: Blurred images due to camera shake can be modeled as a convolution of the latent sharp image and a kernel capturing the motion of the camera. Although natural images are not sparse, they typically have (approximately) sparse gradients \cite{Chan1988-TIP, Levin2011-PAMI}.

In the above applications, the observation signal $\mb y\in\R^m$ is generated via the convolution of a {\em short} kernel $\mb a_0\in\R^k$ with $k\ll m$ and a {\em sparse} activation coefficient $\mb x_0\in\R^m$ with $\norm{\mb x_0}0\ll m$. Without loss of generality, we let $\mb y$ denote the circular convolution of $\mb a_0$ and $\mb x_0$
\begin{equation}
\mb y=\mb a_0\cconv \mb x_0=\extend{\mb a_0}\cconv \mb x_0,
\label{eqn:circulant_conv}
\end{equation}
with $\extend{\mb a_0}\in\R^m$ denoting the zero padded $m$-length version of $\mb a_0$, which can be expressed as $\extend{\mb a_0}=\injector_k\mb a_0$. Here, $\injector_k:\R^k\to\R^m$ is a zero padding operator.
Its adjoint $\injector_k^*:\R^m\to\R^k$ acts as a projection onto the lower dimensional space by keeping the first $k$ components.  

The short-and-sparse blind deconvolution problem exhibits a {\em scaled-shift ambiguity}, which derives from the basic properties of a convolution operator. Namely, for any observation signal $\mb y$, and any nonzero scalar $\alpha$ and integer shift $\tau$, the following equality always holds
\begin{equation}
\mb y = \paren{ \pm \alpha \shift{\extend{\mb a_0}}{\tau} } \cconv \paren{ \pm \alpha^{-1} \shift{\mb x_0}{-\tau} }.
\end{equation} 
Here, $\shift{\mb v}{-\tau}$ denotes the cyclic shift of the vector $\mb v$ by $\tau$ entries: 
\begin{equation}
\shift{\mb v}{\tau}(i)=\mb v\paren{[i-\tau-1]_m+1},\quad\forall\;i\in\set{1,\cdots,m}.
\end{equation}

Clearly, both scaling and cyclic shifts preserve the short-and-sparse structure of $(\mb a_0, \mb x_0)$. This {\em scaled-shift symmetry} raises nontrivial challenges for computation, making straightforward convexification approaches ineffective, and leading to complicated nonconvex optimization landscape. \cite{Zhang2017-CVPR} considers a natural nonconvex formulation of sparse blind deconvolution, in which the kernel $\mb a\in\R^k$ is constrained to have unit Frobenius norm. \cite{Zhang2017-CVPR} argues that under certain idealized conditions, this problem has well-structured local optima, in the sense that {\em every local optimum is close to some shift truncation of the ground truth}. The presence of these local optima can be viewed as a result of the shift symmetry associated to the convolution operator: the shifted and truncated kernel $\injector_k^* \shift{\extend{\mb a_0}}{\tau}$ can be convolved with the sparse signal $s_{-\tau}[ \mb x_0 ]$ (shifted in the opposite direction) to produce a near approximation to $\mb y$ that 
\begin{equation}
\paren{ \injector_k^* \shift{\extend{\mb a_0}}{\tau} } \cconv \shift{\mb x_0}{-\tau} \approx\mb y.
\end{equation}

\begin{SCfigure}
\centering
\includegraphics[width=0.3\textwidth]{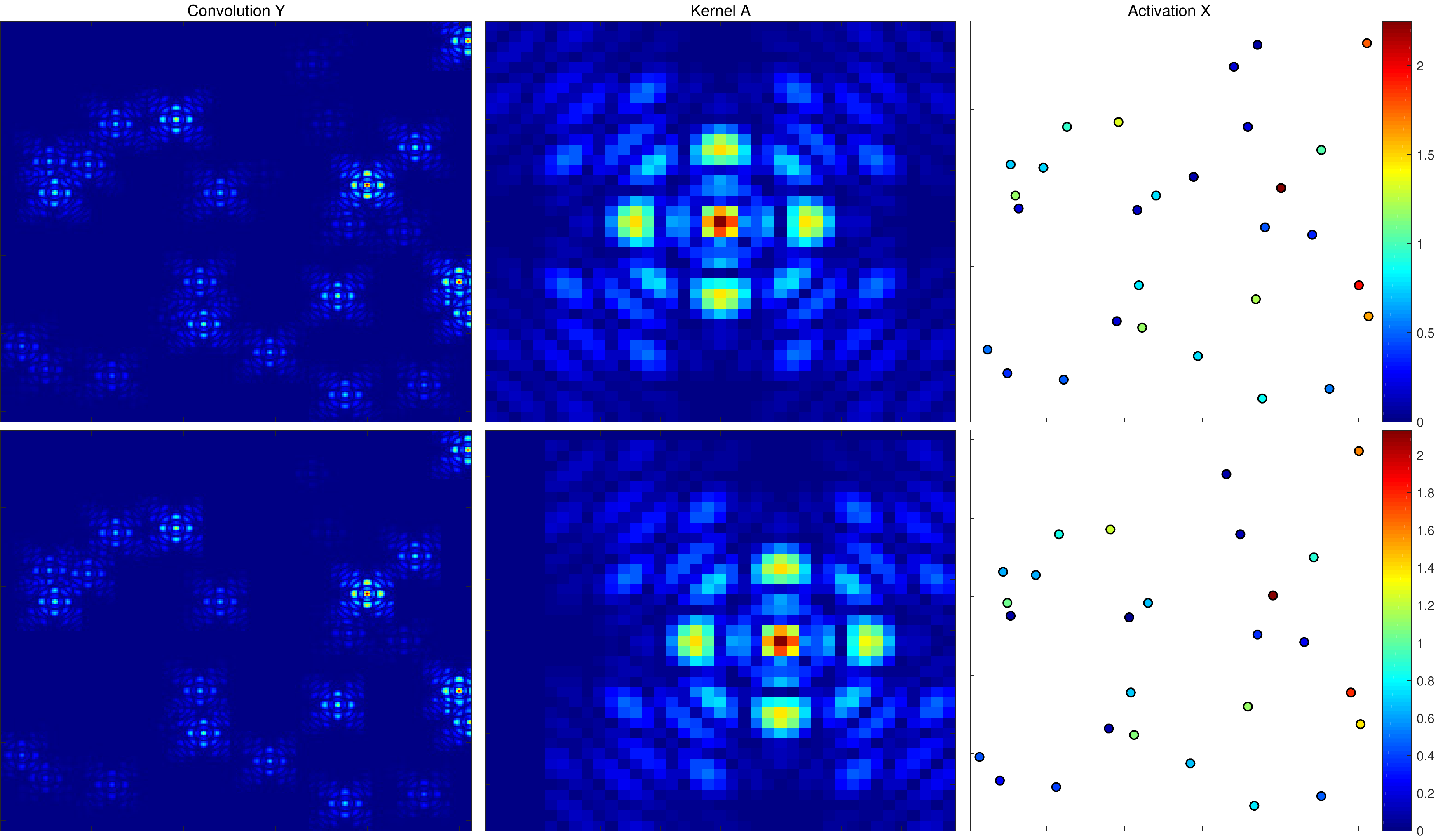}
\caption{{\bf Local Minimum.} \\
Top: observation $\mb y = \mb a_0 \cconv \mb x_0$, ground truth $\mb a_0$ and $\mb x_0$; \\
Bottom: recovered $\mb a \cconv \mb x$, $\mb a$, and $\mb x$ at one local minimum of a natural formulation in \cite{Zhang2017-CVPR}. }
\label{fig:local_min}
\end{SCfigure} 

In \cite{Zhang2017-CVPR}, the geometric insight about local minima is corroborated with a lot of experiments, but rigorous proof is only available under rather restrictive conditions. In this paper, we adopt the unit Frobenius norm constraint as in \cite{Zhang2017-CVPR} but consider a different objective function over the kernel sphere $\bb S^{k-1}$.  We formulate the sparse blind deconvolution problem as the following optimization problem over the sphere:
\begin{equation}
\label{eqn:obj_loose}
\min\;-\norm{\check{\mb y}\cconv \mb r_{\mb y}\paren{\mb q}}4^4 \quad \st \quad \norm{\mb q}F=1
\end{equation}
Here, $\check{\mb y}$ denotes the reversal\footnote{Denote $\mb y=\brac{y_1,y_2,\cdots,y_{m-1},y_m}^T$, then its reversal $\check{\mb y}=\brac{y_1,y_m,y_{m-1},\cdots,y_2}^T$.} of $\mb y$ and $\mb r_{\mb y}\paren{\mb q}$ is a preconditioner which we will discuss in detail later. Convolution $\check{\mb y} \cconv \mb r_{\mb y}\paren{\mb q}$ approximates the reversed underlying activation signal $\mb x_0$, and $-\norm{\cdot}4^4$ serves as the sparsity penalty.

This paper studies the function landscape of the short-and-sparse blind deconvolution problem assuming the short $k$-length convolutional kernel lives on a unit Frobenius norm sphere, denoted as $\bb S^{k-1}$. We demonstrate that even when $\mb x_0$ is relatively dense, a shift truncation $\injector_k^* s_{\tau}\brac{ \extend{\mb a_0} }$ of the ground truth still can be obtained as one local minimum in certain region of the kernel sphere. Such benign region contains the sub-level set of small objective value, and an initial point with small objective value can be easily found. Specifically, for a generic kernel on the sphere \footnote{Here, we refer a kernel sampled following a uniform distribution over the sphere as a generic kernel on the sphere.} $\mb a_0\in\bb S^{k-1}$, if the sparsity rate $\theta\lesssim k^{-2/3}$ and the number of measurement $m \gtrsim \mathrm{poly}(k)$, initializing with some $k$ consecutive entries of $\mb y$ and applying any optimization method which (i) is a descent method, and (ii) converges to a local minimizer under a strict saddle hypothesis \cite{Jin2017escape, XRM2017}, produces a near shift-truncation of the ground truth.

\subsection{Related Works}

Even after accounting for the scale ambiguity, the general blind deconvolution problem remains ill-posed. Different types of prior knowledge about the unknown signals have been introduced and to make the blind deconvolution problem well posed. For example, if the signals $\mb a_0$ and $\mb x_0$ live on known linear subspaces, the blind deconvolution problem can be cast as a low-rank recovery problem, and solved via semidefinite programming. \cite{Ahmed2012} proves that if one of the subspaces is random and the other satisfies a spectral flatness condition, this approach recovers the pair $(\mb a_0,\mb x_0)$ up to scale.
\cite{Li16-pp} provides a more efficient nonconvex algorithm for blind deconvolution under this subspace model. \cite{Ling2015-IP} consider a more complicated model in which one of the signals is sparse in some known dictionary. \cite{Lee17-IT} considers the case where both convolutional signals are sparse in some known dictionaries. These known dictionaries are assumed to be random (e.g., Gaussian or partial Fourier). Identifiability of these blind deconvolution problems is investigated in \cite{Li16-IT, Li17-IT}. \cite{Ling2017-IT} further addresses a simultaneous demixing and deconvolution problem, where the observation is the superposition of multiple convolutions.

The above results offer efficient and guaranteed algorithms for blind deconvolution problems in which the signals of interest are sparse in a random dictionary.
However, in the short-and-sparse blind deconvolution problem in microscopy image analysis or neural spike sorting, the sparse signal is {\em sparse with respect to the standard basis} rather than a random dictionary. Any cyclic shift of a standard basis is another standard basis, therefore the short-and-sparse blind deconvolution problem is only identifiable up to shifts. This is in contrast to the aforementioned random models, which only exhibit a scale ambiguity. When casting the short-and-sparse blind deconvolution problem as an optimization problem, this shift ambiguity creates a large group of equivalent global solutions (convolutional pairs of opposite shifts $\shift{\extend{\mb a_0}}{\tau}$ and $\shift{\mb x_0}{-\tau}$) and therefore much more complicated optimization landscape.

For sparsity in the standard basis, \cite{CM2014-pp1,CM2014-pp2} show that sparsity alone is not sufficient for unique recovery, by demonstrating the existence of manifolds $(\mb a, \mb x)$ of signals that are not identifiable from the convolution $\mb y = \mb a \ast \mb x$. This construction requires both the support and magnitudes of the two signals to be regular: the support of $\mb x$ needs to have the form $J \cup s_1(J)$ for some set $J$, and the nonzero entries of $\mb x$ to take on specific values. When $\mb x$ is either Bernoulli or Bernoulli-Gaussian, with probability one, the pair $(\mb a, \mb x)$ does not fall in this non-identifiable set. 
\cite{Chi2016-TIP} proposes a convex relaxation for a variant of the sparse blind deconvolution problem in which $\mb a$ lies in a random subspace and $\mb x$ is a superposition of spikes with continuous-valued locations. A strong point of this method is that it avoids discretization. Because of the random subspace model on $\mb a$, the results of \cite{Chi2016-TIP} are not directly comparable to ours. However, if the rates from this work were adapted to the short-and-sparse setting, they would require $\mb x$ to be sparse enough that the observation $\mb y$ contains many isolated (non-overlapping) copies of $\mb a$. This seems to reflect a fundamental limitation of convexification approaches in handling signals with multiple structures \cite{OJFEH2016}. 
\cite{WC16} studies another variant where multiple independent observations of circulant convolutions are available, motivated by multi-channel blind deconvolution. Although the convolution kernel is short compared to the total measurements, each independent "short" measurement is self contained. While in the short-and-sparse blind deconvolution problem, only one measurement is available and any "short" measurement heavily depends on adjacent measurements. This nuance leads to much more complicated optimization geometry.

Although the theory of short-and-sparse blind deconvolution remains completely open, many nonconvex algorithms have been developed and practiced in computer vision, where the convolution kernel captures the image blurring process due to camera shake \cite{Levin2011-PAMI}. Motivated by this physical model, people assume the convolutional kernel to be entry-wise nonnegative and sums up to $1$, and then minimize the objective function of following form
\begin{equation}
\min_{\mb a\ge0,\norm{\mb a}1=1}\;\min_{\mb x}\;\tfrac12\norm{\mb y-\mb a\cconv \mb x}2^2+\lambda\norm{\mb x}{\star}.
\end{equation}
In the image deblurring application, $\mb x$ represents the gradient of a natural image and $\norm{\cdot}{\star}$ penalizes the sparsity of $\mb x$. However, such formulation always admits one local minimum obtained at the convolutional pair $\paren{\mb a,\mb x}=\paren{\mb\delta, \mb y}$ \cite{Benichoux2013-ICASSP, Perrone2014-CVPR}. In contrast, \cite{Wipf2013-BayesianBD, Zhang2013-CVPR} carefully compare the difference in MAP and VB approaches, and propose to instead constrain $\mb a$ to have unit Frobenius norm -- i.e., to reside on a high-dimensional sphere. \cite{Zhang2017-CVPR} studies the optimization landscape of the sphere constrained sparse blind deconvolution and firstly identifies the structure of the local solutions. In particular,  \cite{Zhang2017-CVPR} casts the short-and-sparse blind deconvolution problem as an optimization problem over the sphere: 
\begin{equation}
\label{eqn:bd_lasso}
\min_{\mb a\in\bb S^{k-1}}\min_{\mb x}\tfrac12\norm{\mb y-\mb a\cconv\mb x}2^2+\lambda\norm{\mb x}1,
\end{equation}
and presents empirical evidence that {\em local minima $\bar{\mb a}$ are close to certain shift truncations of $\mb a_0$.}  \cite{Zhang2017-CVPR} further proves that a ``linearized'' version of \eqref{eqn:bd_lasso}, which neglects quadratic interactions in $\mb a$, satisfies this property, in the ``dilute limit'' in which the sparse signal $\mb x_0$ is a single spike. In this paper, we demonstrate that for a different objective function, this claim holds under much broader conditions than what is proved in \cite{Zhang2017-CVPR}. In particular, our results allow the sparse signal $\mb x_0$ to be much denser. 

\subsection{Assumptions and Notations}
We assume that $\mb x_0\in\R^m$ follows the Bernoulli-Gaussian (BG) model with sparsity level $\theta$: $\mb x_0\paren{i}=\omega_i g_i$ with $\omega_i\sim \mathrm{Ber}\paren{\theta}$ and $g_i\sim\mc N\paren{0,1}$, where all the different random variables are jointly independent. For simplicity, we write $\mb x_0\sim_{i.i.d.} \mathrm{BG}\paren{\theta}$.

Throughout this paper, vectors $\mb v\in\mathbb{R}^k$ are indexed as $\mb v =[v_1,v_2,\cdots,v_k]$, and $[\cdot]_{m}$ denotes the modulo operator of $m$. We use $\norm{\cdot}2$ to denote the operator norm,  $\norm{\cdot}{F}$ to denote the Frobenius norm, and $\norm{\cdot}{p}$ to denote the entry wise $\ell^p$ norm. $\paren{\cdot}_{I}$ denotes the projection onto subset with index $I$ and $\proj{\cdot}{\bb S}=\frac{\cdot}{\norm{\cdot}{F}}$ denotes the projection onto the Frobenius sphere. $\paren{\cdot}^{\circ p}$ is the entry wise $p$-th order exponent operator. We use $C$, $c$ to denote positive constants, and their value change across the paper.

\section{Problem Formulation and Main Results}
\label{sec:formulation}
In the short-and-sparse blind deconvolution problem, any $k$ consecutive entries in $\mb y$ only depend on  $2k-1$ consecutive entries in $\mb x_0$:
\vspace{-.1in}
\begin{align}
\mb y_i&=\brac{y_{i} , \cdots , y_{1+\brac{i+k-1}_m}}^T
=\sum_{\tau=-\paren{k-1}}^{k-1} x_{1+\brac{i+\tau-1}_m}\cdot\injector_k^*\shift{\extend{\mb a_0}}{\tau}\\
&= \underbrace{\begin{bmatrix}
a_{k} & a_{k-1} & \cdots  & a_1  & \cdots & 0 & 0 \\
0 & a_{k} & \cdots & a_2  & \cdots & 0 & 0\\
\vdots & \vdots & \ddots & \vdots  & \ddots & \vdots & \vdots\\
0 & 0 & \cdots & a_{k-1} & \cdots & a_1 & 0\\
0 & 0 & \cdots & a_{k} & \cdots & a_2 & a_1
\end{bmatrix}}_{\mb A_0\in\R^{k\times\paren{2k-1}}}
\underbrace{\begin{bmatrix}
x_{1+\brac{i-k}_m} \\
\vdots\\
x_{i}\\
\vdots\\
x_{1+\brac{i+k-2}_m}
\end{bmatrix}}_{\mb x_i\in\R^{\paren{2k-1}\times1}}.
\vspace{-.05in}
\end{align}
Write $\mb Y = [ \mb y_1, \mb y_2, \dots, \mb y_m ] \in \mathbb{R}^{k \times m}$ and $\mb X_0 = [ \mb x_1, \dots, \mb x_m ] \in \mathbb{R}^{2k -1 \times m}$. Using the above expression, we have that 
\begin{equation}
\mb Y = \mb A_0 \mb X_0.
\end{equation}
Each column $\mb x_i$ of $\mb X_0$ only contains some $2k-1$ entries of $\mb x_0$. The {\em rows} of $\mb X_0$ are cyclic shifts of the reversal of $\mb x_0$: 
\vspace{-.1in}\begin{equation}
\mb X_0 = \left[ \begin{smallmatrix} \shift{\check{\mb x}_0}{0} \\  \vdots \\ \shift{\check{\mb x}_0}{2k-2} \end{smallmatrix} \right]. 
\end{equation}
The shifts of $\check{\mb x}_0$ are {\em sparse vectors} in the linear subspace $\mathrm{row}(\mb X_0)$. Note that if we could recover some shift $s_\tau[\mb x_0]$, we could subsequently determine $\shift{\mb a_0}{-\tau}$ by solving a linear system of equations, and hence solve the deconvolution problem, up to the shift ambiguity.

\subsection{Finding a Shifted Sparse Signal}

In light of the above observations, a natural computational approach to sparse blind deconvolution is to attempt to find $\mb x_0$ by searching for a sparse vector in the linear subspace $\mathrm{row}(\mb X_0)$, e.g., by solving an optimization problem
\begin{equation}
\min\quad\norm{\mb v}{\star}\quad\st\quad\mb v\in \mathrm{row}\paren{\mb X_0},\;\norm{\mb v}2=1,
\end{equation}
where $\norm{\cdot}{\star}$ is chosen to encourage sparsity of the target signal \cite{Spielman12-pp, SQW15-pp, QSW16-IT, Hopkins2016-STOC}. 

In sparse blind deconvolution, we do not have access to the row space of $\mb X_0$. Instead, we only observe the subspace $\mathrm{row}(\mb Y ) \subset \mathrm{row}(\mb X_0)$. The subspace $\mathrm{row}(\mb Y)$ does not necessarily contain the desired sparse vector $\mb e_i^T\mb X_0$, but it {\em does} contain some approximately sparse vectors. In particular, consider following vector in $\mathrm{row}(\mb Y)$, 
\begin{align}
\mb v = \mb Y^T\mb a_0=\underset{ \text{\bf sparse}}{\check{\mb x}_0}
+\underbrace{\sum_{i\neq 0}\innerprod{\mb a_0}{\shift{\mb a_0}{i}}\shift{\check{\mb x}_0}{i}}_{\text{\bf ``noise'' $\mb z$}}. \label{eqn:spiky-vector}
\end{align}
The vector $\mb v$ is a superposition of a sparse signal $\check{\mb x}_0$ and its scaled shifts $\innerprod{\mb a_0}{\shift{\mb a_0}{i}}\shift{\check{\mb x}_0}{i}$.  If the shift-coherence $|\innerprod{\mb a_0}{\shift{\mb a_0}{\tau}}|$ is small\footnote{For a generic kernel $\mb a_0$, the shift-coherence is bounded as $\sup_{\tau\neq0}\abs{\innerprod{\mb a_0}{\shift{\mb a_0}{\tau}}}\lesssim\sqrt{\log{k}/k}$.} and $\mb x_0$ is sparse enough, $\mb z$ can be viewed as small noise.\footnote{In particular, under a Bernoulli-Gaussian model, for each $j$, $\bb E[ \mb z_j^2 ] = \theta \sum_{i\ne0} \innerprod{ \mb a_0 }{ s_i[\mb a_0] }^2$.} The vector $\mb v$ is not sparse, but it is {\em spiky}: a few of its entries are much larger than the rest. We deploy a milder sparsity penalty $-\norm{\cdot}4^4$ to recover such a spiky vector, as $\norm{\cdot}4^4$ is very flat around $0$ and insensitive to small noise in the signal.\footnote{In comparison, the classical choice $\norm{\cdot}{\star} = \norm{\cdot}{1}$ is a strict sparsity penalty that essentially encourages all small entries to be $0$.} This gives
\vspace{-.05in}\begin{equation}
\min\quad-\tfrac14\norm{\mb v}4^4\quad\st\quad\mb v\in \mathrm{row}\paren{\mb Y},\;\norm{\mb v}2=1.
\end{equation}
We can express a generic unit vector $\mb v \in \mathrm{row}(\mb Y)$ as $\mb v = \mb Y^T\paren{\mb Y\mb Y^T}^{-1/2}\mb q$, with $\norm{\mb v}2=\norm{\mb q}2$. This leads to the following equivalent optimization problem over the sphere

\begin{align}
\label{eqn:psi}
\min&\quad\psi\paren{\mb q}\doteq-\frac{1}{4m}\norm{\mb Y^T\paren{\mb Y\mb Y^T}^{-1/2}\mb q}4^4\nonumber\\
\st&\quad\norm{\mb q}2=1.
\end{align}

\paragraph{Interpretation: preconditioned shifts.} This objective $\psi\paren{\mb q}$ can be rewritten as
\begin{align}
\psi\paren{\mb q} &= -\frac{1}{4m}\norm{\check{\mb y}\cconv\paren{\mb Y\mb Y^T}^{-1/2}\!\mb q}4^4\\
&= -\frac{1}{4m}\norm{\check{\mb x}_0\cconv\mb A_0^T\paren{\mb Y\mb Y^T}^{-1/2}\!\mb q}4^4\\
& \sim \norm{ \check{\mb x}_0 \circledast \mb \zeta }{4}^4, \label{eqn:psi2}
\end{align}
where $\mb \zeta = \mb A_0^T (\mb A_0 \mb A_0^T)^{-1/2} \mb q$. This approximation becomes accurate as $m$ grows.\footnote{As $\bb E_{\mb x_0\simiid\mathrm{BG}\paren{\theta}}[\mb Y\mb Y^T]=\bb E_{\mb x_0\simiid\mathrm{BG}\paren{\theta}}[\mb A_0\mb X_0\mb X_0^T\mb A_0^T]=\theta m\mb A_0\mb A_0^T.$} This objective encourages the convolution of $\check{\mb x}_0$ and $\mb\zeta$ to be as spiky as possible.
 Reasoning analogous to \eqref{eqn:spiky-vector} suggests that $\check{\mb x}_0 \circledast \mb\zeta$ will be spiky if

\begin{equation}
\mb \zeta=\mb A_0^T\paren{\mb A_0\mb A_0^T}^{-1/2}\mb q\approx \mb e_l,\quad l\in\set{1,\cdots,2k-1}.
\end{equation}
For simplicity, we define the preconditioned convolution matrix
\begin{equation}
\mb A \doteq \paren{\mb A_0\mb A_0^T}^{-1/2}\mb A_0=\begin{bmatrix}\mb a_1&\mb a_2&\cdots &\mb a_{2k-1}
\end{bmatrix},
\end{equation}
with column coherence (preconditioned shift coherence) $\mu\doteq\max_{i\neq j}\abs{\innerprod{\mb a_i}{\mb a_j}}$. As $\mb A$ is preconditioned, we have $\norm{\mb\zeta}2=\norm{\mb q}2=1$ and
\begin{align}
\norm{\mb a_i}2^2\le \norm{\mb A^T\mb a_i}2\le \norm{\mb a_i}2\implies\norm{\mb a_i}2\le1.
\end{align}
 Here, the unit vector $\mb\zeta$ can also be interpreted as measuring the inner products of $\mb q$ with columns of $\mb A$. We will show that minimizing this objective over a certain region of the sphere yields a preconditioned shift truncate $\mb a_l$, from which we can recover a shift truncate of the original signal $\mb a_0$. 

\subsection{Structured Local Minima}
\begin{wrapfigure}{r}{0.25\textwidth}
\vspace{-.2in}
\centering\includegraphics[trim = {1in 1.5in 1in .8in}, clip, width=0.25\textwidth]{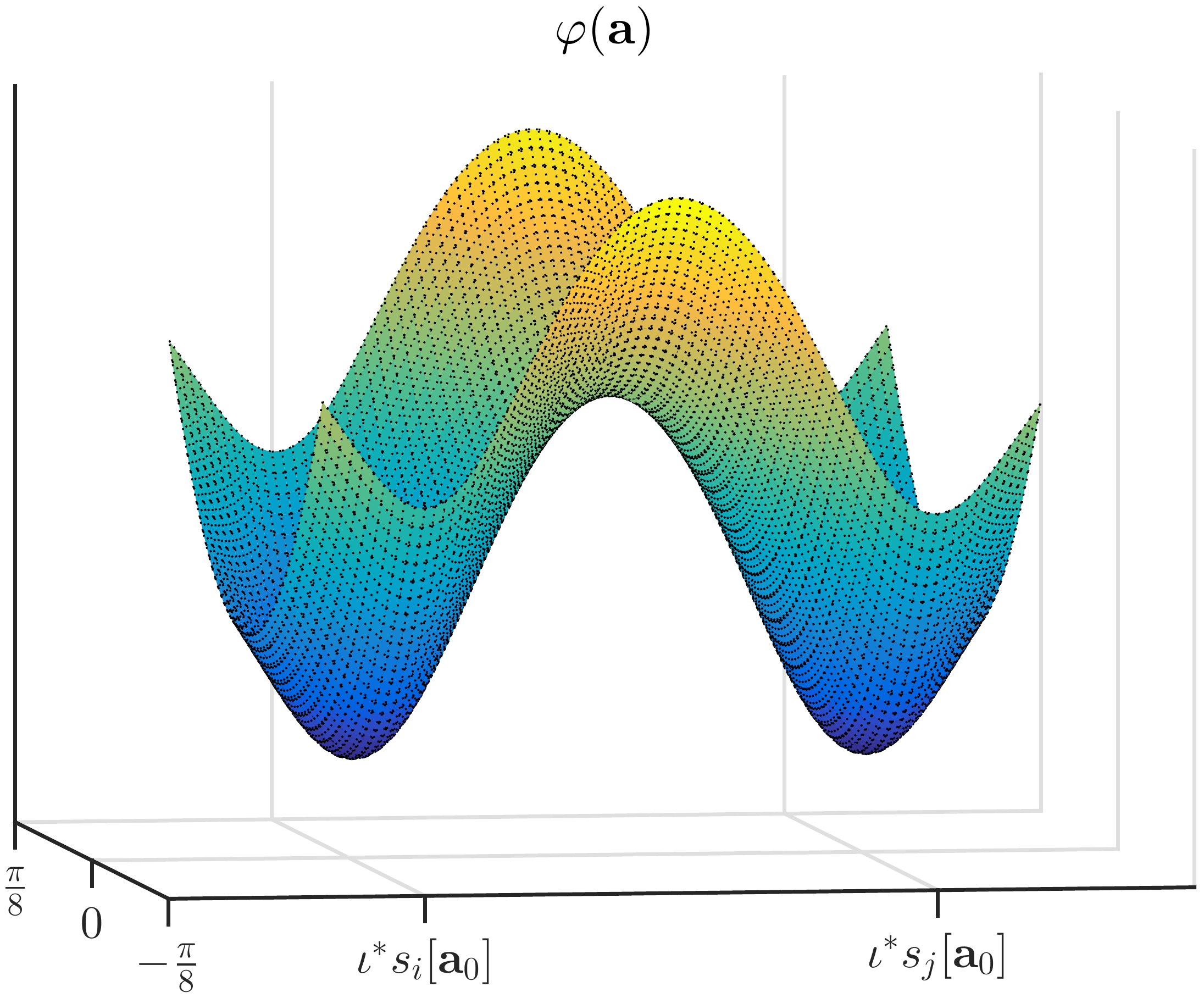}
\vspace{-.2in}
\caption{{\bf Saddles points} are approximately superpositions of local minima.}
\label{fig:dq_geo}
\end{wrapfigure} 

We will show that in a certain region $\mc R_{C_\star} \subset \bb S^{k-1}$, the preconditioned shift truncations $\mb a_l$ are the {\em only} local minimizers. Moreover, the other critical points in $\mc R_{C_\star}$ can be interpreted as resulting from competition between several of these local minima (\Cref{fig:dq_geo}). At any saddle point, there exists strict negative curvature in the direction of a nearby local minimizer which breaks the balance in favor of some particular $\mb a_l$. The region $\mc R_{C_\star}$ is defined as follows: 
\begin{definition}
\label{def:R}
For fixed ${C_\star} > 0$, letting $\kappa$ denote the condition number of $\mb A_0$, and $\mu\doteq\max_{i\neq j}\abs{\innerprod{\mb a_i}{\mb a_j}}$ the column coherence of $\mb A$, we define two regions $\mc R_{C_\star}$, $\hat{\mc R}_{C_\star} \subset \bb S^{k-1}$, as 
\begin{align}
&\mc R_{C_\star}\!\doteq\set{\mb q\in\bb S^{k-1}\!\mid\!\norm{\mb A^T\mb q}4^6\ge {C_\star}\mu \kappa^2\norm{\mb A^T\mb q}3^3}.\label{eqn:R}\\
&\hat{\mc R}_{C_\star}\!\doteq\set{\mb q\in\bb S^{k-1}\!\mid\!\norm{\mb A^T\mb q}4^6\ge C_\star\mu \kappa^2}\subseteq\mc R_{C_\star}.\label{eqn:Rhat}
\end{align}
\end{definition}
A simpler and smaller region $\hat{\mc R}_{C_\star}$ is also introduced in Definition (\ref{def:R}). This region $\hat{\mc R}_{C_\star}$ can be viewed as a sub-level set for $-\norm{\mb A^T\mb q}4^4$, which is proportional to the objective value $\psi\paren{\mb q}$ assuming $m$ is sufficiently large\footnote{Please refer to Section \ref{sec:geo} for more arguments.}. Therefore, once initialized within $\hat{\mc R}_{C_\star}$, the iterates produced by a descent algorithm will stay in $\hat{\mc R}_{C_\star}$. 

In particular, at any stationary point $\mb q\in\mc R_{10}$, the local optimization landscape can be characterized in terms of the number of spikes (entries with nontrivial magnitude\footnote{We call any $\zeta_l$ with magnitude no smaller than $2\mu\norm{\mb\zeta}3^3/\norm{\mb\zeta}4^4$ to be nontrivial and defer technical reasonings to later sections.}) in $\mb\zeta$. If there is only one spike in $\mb \zeta$, then such stationary point $\mb q$ is a local minimum that is close to one local minimizer; if there are more than two spikes in $\mb \zeta$, then such stationary point $\mb q$ is saddle point. Based on the above characterizations of stationary points in $\mc R_{C_\star}$ with $C_\star\ge 10$, we can deduce that any local minimum is close to some $\mb a_l$, a preconditioned shift truncation of the ground truth $\mb a_0$.

\begin{theorem}[Main Result]
\label{thm:main}
Assuming observation $\mb y\in\R^m$ is the circulant convolution of $\mb a_0\in\R^k$ and $\mb x_0\simiid \mathrm{BG}\paren{\theta}\in\R^m$, where the convolutional matrix $\mb A_0$ has minimum singular value $\sigma_{\min}>0$ and condition number $\kappa\ge 1$, and $\mb A$ has column incoherence $0\le\mu<1$. There exists a positive constant $C$ such that whenever the number of measurements
\begin{equation}
m\ge C\frac{\min\set{\mu^{-4/3}\!,\kappa^2k^2}}{\paren{1-\theta}^2\sigma^2_{\min}}\kappa^8k^4\log^3\!\paren{\frac{\kappa k}{\paren{1-\theta}\sigma_{\min}}}
\end{equation}
and $\theta \ge \log{k}/k$, then with high probability, any local optima $\mb{\bar q}\in\hat{\mc R}_{2C_\star}$ satisfies 
\begin{equation}
\abs{\innerprod{\mb{\bar q}}{\mc P_{\bb S}\brac{\mb a_{l}}}}\ge 1-c_{\star}\kappa^{-2}
\end{equation}
for some integer $1\le l\le 2k-1$. Here, $C_\star\ge 10$ and $c_\star = 1/C_\star$.
\end{theorem}
This theorem says that any local minimum in $\hat{\mc R}_{2C_\star}$ is close to some normalized column of $\mb A$ given polynomially many observation. The parameters $\sigma_{\min}$, $\kappa$ and $\mu$ effectively measure the spectrum flatness of the ground truth kernel $\mb a_0$ and characterize how broad the results hold. A generic kernel usually has larger $\sigma_{\min}$, smaller $\kappa$ and $\mu$, which equivalently implies the result holds in a large sub-level set $\hat{\mc R}_{2C_\star}$ even with fewer observations.\footnote{In comparison, a low pass or high pass signal always has smaller $\sigma_{\min}$, bigger $\kappa$ and $\mu$, with simulations presented in the Appendix (\Cref{fig:para_bandpass_A0}).}

Hence, once assuring the algorithm finds a local minimum in $\hat{\mc R}_{2C_\star}$, then some shifted truncation of the ground truth kernel $\mb a_0$ can be recovered. In other words, if we can find an initialization point with small objective value, then a descent algorithm minimizing the objective function guarantees that $\mb q$ always stays in $\hat{\mc R}_{2C_{\star}}$ in proceeding iterations. Therefore, any descent algorithm that escapes a strict saddle point can be applied to find some $\mb a_l$, or some shift truncation of $\mb a_0$.

\subsection{Initialization with a Random Sample}
Recall that $\mb y_i = \mb A_0\mb x_i$, which is a sparse superposition of about $2\theta k$ columns of $\mb A_0$. Intuitively speaking, such $\qinit$ already encodes certain preferences towards a few preconditioned shift truncations of the ground truth. Therefore, we randomly choose an index $i$ and set the initialization point as 
\begin{equation}
\qinit=\mc P_{\bb S}\brac{\paren{\mb Y\mb Y^T}^{-1/2}\mb y_{i}}.\end{equation}
Using $\bb E_{\mb x_0\simiid\mathrm{BG}\paren{\theta}}[\mb Y\mb Y^T]=\theta m\mb A_0\mb A_0^T$ again, we have
\begin{equation}
\zetainit= \mb A^T\qinit\approx \mc P_{\bb S}\brac{\mb A^T\mb A\mb x_i}.
\end{equation}
For a generic kernel $\mb a_0\in \bb S^{k-1}$, $\mb A^T\mb A$ is close to a diagonal matrix, as the magnitudes of off-diagonal entries are bounded by column incoherence $\mu$.  
Hence, the sparse property of $\mb x_i$ can be approximately preserved, that $\mc P_{\bb S}\brac{\mb A^T\mb A\mb x_i}$ is spiky vector with small $-\norm{\cdot}4^4$. 
By leveraging the sparsity level $\theta$, one can make sure such initialization point $\qinit$ falls in $\hat{\mc R}_{2C_\star}$. Therefore, we propose Algorithm \ref{alg:ssbd} for solving sparse blind deconvolution with its working conditions stated in Corollary \ref{coro:main}. For the choice of descent algorithms which escape strict saddle points, there are several such algorithms specially tailored for sphere constrained optimization problems \cite{Absil2007-FCM, Goldfarb2010-SIAM}.
\begin{algorithm}
\caption{Short and Sparse Blind Deconvolution}
\label{alg:ssbd}
\begin{algorithmic}[1]
\renewcommand{\algorithmicrequire}{\textbf{Input:}}
\renewcommand{\algorithmicensure}{\textbf{Output:}}
\Require~\
Observations $\mb y\in\R^m$ and kernel size $k$.
\Ensure~~\
Recovered Kernel $\mb{\bar a}$.
\State Generate random index $i\in\brac{1,m}$ and set 
\begin{equation*}
\mb q_{init}=\mc P_{\bb S}\brac{\paren{\mb Y\mb Y^T}^{-1/2}\mb y_{i}}.
\end{equation*}

\State Solve following nonconvex optimization problem with a descent algorithm that escapes saddle point and find a local minimizer
\begin{equation*}
\mb{\bar q}=\arg\min_{\mb q\in\bb S^{k-1}}\varphi\paren{\mb q}.
\end{equation*}

\State Set $\mb{\bar a}=\mc P_{\bb S}\brac{\paren{\mb Y\mb Y^T}^{1/2}\mb{\bar q}}$.
\end{algorithmic}
\end{algorithm}

\begin{corollary}
\label{coro:main}
Suppose the ground truth $\mb a_0$ kernel has preconditioned shift coherence $0\le\mu\le\tfrac1{8\times48}\log^{-3/2}\paren{k}$ and sparse coefficient $\mb x_0\simiid \mathrm{BG}\paren{\theta}\in\R^m$. There exist positive constants $C\ge 2560^4$ and $C'$ such that whenever the sparsity level
\begin{align}
&64k^{-1}\log{k}\le\theta\le\min\big\{\tfrac1{48^2}\mu^{-2}k^{-1}\log^{-2}{k},\nonumber\\ 
&\quad \paren{\tfrac1{4}-\tfrac{640}{C^{1/4}}}\paren{3C_{\star}\mu\kappa^2}^{-2/3}k^{-1}\paren{1+36\mu^2k\log k}^{-2}\big\},\nonumber
\end{align}
and signal length 
\begin{align}
m\;\ge\;&\max\big\{C\theta^2\sigma^{-2}_{\min}\kappa^6k^3\paren{1+36\mu^2k\log{k}}^4\log\paren{\kappa k},\nonumber\\
C'&\paren{1-\theta}^{-2}\sigma^{-2}_{\min}\min\set{\mu^{-1},\kappa^2k^2}\kappa^8 k^4\log^3\paren{\kappa k}\big\},\nonumber
\end{align}
then with high probability, Algorithm \ref{alg:ssbd} recovers $\mb{\bar a}$ such that 
\begin{equation}
\norm{\mb{\bar a}\pm\mc P_{\bb S}\brac{\injector_k\shift{\extend{\mb a_0}}{\tau}}}2\le  4\sqrt{c_\star}+ck^{-1} 
\end{equation}
for some integer shift $-\paren{k-1}\le\tau\le k-1$. 
\end{corollary}
For a generic $\mb a_0\in\bb S^{k-1}$, plugging in the numerical estimation\footnote{Exact and rigorous calculation of these parameters involves property of the banded Toeplitz matrix, which has been under intense study while remains open.} of the parameters $\sigma_{\min}$, $\kappa$ and $\mu$ (Figure \ref{fig:para_A0}), accurate recovery can be obtained with $m\gtrsim\theta^2k^6\poly\log\paren{k}$ measurements and sparsity level $\theta\lesssim k^{-2/3}\poly\log\paren{k}$. For bandpass kernels $\mb a_0$, $\sigma_{\min}$ is smaller and $\kappa$, $\mu$ are larger, and so our results require $\mb x_0$ to be longer and sparser.

\section{Asymptotic Function Landscape}
\label{sec:geo}
In the next two sections, we discuss some key elements of our analysis. In this section,  we first investigate the stationary points of the ``population'' objective $\bb E_{\mb x_0}[\psi (\mb q)]$. We demonstrate that any local minimizer in $\mc R_{C_\star}$ is close to a signed column of $\mb A$, a preconditioned shift truncation of $\mb a_0$. In the next section, we then demonstrate that when $m$ is sufficiently large, the ``finite sample'' objective $\psi(\mb q)$ satisfies the same property.

In Section \ref{sec:stat_appr}, we show how to accurately estimate the vector $\mb\zeta = \mb A^T \mb q$ at any stationary point $\mb q\in\mc R_{C_\star}$. In Section \label{sec:stat_spike}\ref{sec:stat_geo}, we show how the number of spikes in $\mb\zeta $ determines the geometry around a stationary point.
\begin{itemize}
\item For any stationary point $\mb q\in\mc R_{C_\star}$, its preconditioned cross-correlation $\mb\zeta$ has {\em at least} one large entry (Section \ref{sec:stat_spike}). This implies that any stationary point $\mb q$ must be close some local minimizer.
\item If $\mb\zeta$ has {\em only} one large entry, then $\mb q$ is a local minimizer. (Section \ref{sec:stat_localmin})  
\item If $\mb\zeta$ has {\em more than} one large entry, then $\mb q$ is a strict saddle point. (Section \ref{sec:stat_saddle})
\end{itemize}
With above three characterizations, we can deduce that any local minimizer in $\mc R_{C_\star}$ is close to some column of $\mb A$, a preconditioned shift truncation of $\mb a_0$.

\subsection{Stationary Points}
\label{sec:stat_appr}
Using $\bb E_{\mb x_0\simiid\mathrm{BG}\paren{\theta}}[\mb Y\mb Y^T]
=\theta m\mb A_0\mb A_0^T$ again, the expectation of the objective function $\psi\paren{\mb q}$ can be approximated (\Cref{lem:obj_exp}) as
\begin{align}
\bb E_{\mb x_0 \sim_{\mathrm{i.i.d.}} \mathrm{BG}(\theta)} [ \psi(\mb q)]  &\approx \bb E_{\mb x_0\simiid\mathrm{BG}\paren{\theta}}\brac{-\frac1m\norm{\mb Y^T\paren{\theta m\mb A_0\mb A_0^T}^{-1/2}\mb q}4^4}\nonumber\\
&=-\frac{1}{\theta^2m^2}\brac{3\theta\paren{1-\theta}\norm{\mb A^T\mb q}4^4+3\theta^2\norm{\mb A^T\mb q}2^4}\nonumber\\
&=-\frac{3\paren{1-\theta}}{\theta m^2}\norm{\mb A^T\mb q}4^4-\frac{3}{m^2}.
\end{align}

In the next section, we will argue that the critical points of the finite sample objective $\psi(\mb q)$ are close to those of the asymptotic approximation $\phi$. We can therefore study the critical points of $\psi$ by studying the simpler problem  
\begin{equation}
\min_{\mb q\in\R^{k-1}}\varphi\paren{\mb q}\doteq-\frac1{4}\norm{\mb A^T\mb q}4^4 = -\frac1{4}\norm{\mb\zeta}4^4.
\end{equation}
The Euclidean gradient and Hessian for $\varphi(\mb q)$ can be calculated as 
\begin{align}
\nabla\varphi(\mb q)&=-\mb A\mb\zeta^{\circ3},\\
\nabla^2\varphi(\mb q)&=-3\mb A\diag\paren{\mb\zeta^{\circ2}}\mb A^T.
\end{align}
We can study the critical points of $\varphi$ over the sphere using the {\em Riemannian} gradient and hessian \cite{Absil2007}
\begin{align}
\grad\varphi(\mb q)&=\mb P_{\mb q^{\perp}}\brac{\nabla\varphi(\mb q)}\\
&=-\mb A\mb\zeta^{\circ3}+\mb q\norm{\mb\zeta}4^4, \label{eqn:riem-grad} \\
\Hess\varphi(\mb q)&=\mb P_{\mb q^{\perp}}\!\brac{\nabla^2\varphi(\mb q)-\innerprod{\nabla\varphi(\mb q)}{\mb q}\mb I}\mb P_{\mb q^{\perp}}\\
&=-\mb P_{\mb q^{\perp}}\!\!\brac{3\mb A\diag(\mb\zeta^{\circ2})\mb A^T\!\!-\!\norm{\mb\zeta}4^4\mb I}\!\mb P_{\mb q^{\perp}}.
\end{align}
Here, $\mb P_{\mb q^{\perp}}=\mb I-\mb q\mb q^T$ denotes the projection onto the tangent space of the Frobenius sphere at point $\mb q\in\bb S^{k-1}$.

As in the Euclidean space, a stationary point on the sphere satisfies $\grad\brac{\varphi}\paren{\mb q}=\mb 0$. Using \eqref{eqn:riem-grad}, at any stationary point of $\varphi$,
\begin{equation}
\mb A\mb\zeta^{\circ3}-\mb q\norm{\mb\zeta}4^4=\mb 0.
\end{equation}
Left-multiplying both sides of the equation by $\mb A^T$, we have
\begin{equation}
\mb A^T\mb A\mb\zeta^{\circ3}-\mb A^T\mb q\norm{\mb\zeta}4^4=\mb 0.
\end{equation}
For the $i$-th entry, following equality always holds
\begin{align}
&0=\norm{\mb a_i}2^2\zeta^3_i+\sum_{j\neq i}\innerprod{\mb a_i}{\mb a_j}\zeta^3_j-\zeta_i\norm{\mb\zeta}4^4\\
\Rightarrow\quad & 0=\zeta^3_i-\zeta_i\underbrace{ \frac{ \norm{\mb\zeta}4^4 }{ \norm{\mb a_i}2^2 } }_{\alpha_i}+\underbrace{ \frac{ \sum_{j\neq i}\innerprod{\mb a_i}{\mb a_j}\zeta^3_j }{ \norm{\mb a_i}2^2 } }_{\beta_i}.
\end{align}
For simplicity, we deploy the following notations
\begin{equation}
\alpha_i = \frac{ \norm{\mb\zeta}4^4 }{ \norm{\mb a_i}2^2 },\quad
\beta_i = \frac{ \sum_{j\neq i}\innerprod{\mb a_i}{\mb a_j}\zeta^3_j }{ \norm{\mb a_i}2^2 }.
\end{equation}
If $\alpha_i \gg \beta_i$, \Cref{prop:root} shows that $\zeta_i$ is very close to one of three values: $0$, or $\pm \sqrt{\alpha_i}$. 
\begin{proposition}
\label{prop:root}
Let $\mb q\in\bb S^{k-1}$ be a stationary point satisfying $\norm{\mb A^T\mb q}4^6\ge 4\mu\norm{\mb A^T\mb q}3^3$, then the $i$-th entry of $\mb \zeta=\mb A^T\mb q$ falls in the range 
\begin{equation}
\set{0,\pm\sqrt{\alpha_i}}\pm\frac{2\beta_i}{\alpha_i},
\end{equation}
with 
\begin{equation}
\alpha_i = \frac{ \norm{\mb\zeta}4^4 }{ \norm{\mb a_i}2^2 },\quad
\beta_i = \frac{ \sum_{j\neq i}\innerprod{\mb a_i}{\mb a_j}\zeta^3_j }{ \norm{\mb a_i}2^2 }.
\end{equation}
\end{proposition}

\begin{proof}
Since $\norm{\mb\zeta}4^6\ge 4\mu\norm{\mb\zeta}3^3$ and $\norm{\mb a_i}2\le 1$, for any index $i$ we have
\begin{align}
\norm{\mb\zeta}4^6\ge 4\mu\norm{\mb\zeta}3^3 \ge 4\norm{\mb a_i}2 \sum_{j\neq i}\innerprod{\mb a_i}{\mb a_j}\zeta^3_j.
\end{align}
This implies $\beta_i \le \frac14\alpha_i^{3/2}$ for any index $i$. Therefore, the roots can be estimated by applying \Cref{lem:cubic} with
\begin{align}
&\sqrt{\alpha_i}=\frac{\norm{\mb\zeta}4^2}{\norm{\mb a_i}2},\\
&\frac{2\beta_i}{\alpha_i}=\frac{ 2\sum_{j\neq i}\innerprod{\mb a_i}{\mb a_j}\zeta^3_j }{\norm{\mb\zeta}4^4}\le\frac{2\mu\norm{\mb\zeta}3^3}{\norm{\mb\zeta}4^4}.
\end{align}
\end{proof}

This implies that either $\abs{\innerprod{\mb a_i}{\mb q}}$ is large ($\approx \sqrt{\alpha_i}$) or it is very close to zero. 

\subsection{Function Landscape on $\mc R_{C_\star}$ }
\label{sec:stat_geo}
In this section, we study the optimization landscape around a stationary point $\mb q$ by bounding the eigenvalues of the Riemannian Hessian $\Hess\brac{\varphi}\paren{\mb q}$: if $\Hess\brac{\varphi}\paren{\mb q}$ is positive semidefinite, then the $\varphi$ is convex in a neighborhood of $\mb q$ and hence $\mb q$ is a local minimum; if $\Hess\brac{\varphi}\paren{\mb q}$ has a negative eigenvalue, then there exists a direction along which the objective value decreases and hence $\mb q$ is a saddle point.

Note that the Riemannian Hessian $\Hess\brac{\varphi}\paren{\mb q}$ at stationary point $\mb q$ is a function of $\mb\zeta$ which can be accurately estimated when constrained in $\mc R_{C_\star}$ with $C_\star\ge 10$. By plugging the estimation of $\mb\zeta$ in the Riemannian Hessian, we can bound the eigenvalues of $\Hess\brac{\varphi}\paren{\mb q}$, and hence we can characterize the optimization landscape around a stationary point $\mb q$.

\subsubsection{Nontrivial Preference of a Stationary Point}
\label{sec:stat_spike}
First, we demonstrate that for any stationary point $\mb q \in\mc R_{C_\star}$ with $C_\star\ge 10$, $\mb \zeta$ must have at least one large entry. 

\begin{lemma}
\label{lem:preference}
For any stationary point $\mb q\in\mc R_{C_\star}$ with $C_{\star}\ge 10$,  \begin{equation}
\norm{\mb\zeta}{\infty}\ge\frac{2\mu\norm{\mb\zeta}3^3}{\norm{\mb\zeta}4^4}.
\end{equation}
\end{lemma}

\begin{proof}
We give a proof by contradiction. Suppose that $\mb q\in\mc R_{C_\star}$ with $C_{\star}\ge 10$, and every entry of $\mb\zeta$ has small magnitude such that $\norm{\mb\zeta}{\infty}< 2\mu\norm{\mb\zeta}3^3/\norm{\mb\zeta}4^4$, then 
\begin{equation}
\norm{\mb\zeta}4^4\le\norm{\mb\zeta}{\infty}^2\le\paren{\frac{2\beta_i}{\alpha_i}}^2\le\frac{4\mu^2\norm{\mb\zeta}3^6}{\norm{\mb\zeta}4^8},
\end{equation}
which indicates $\norm{\mb\zeta}4^6\le2\mu\norm{\mb\zeta}3^3$ and contradicts the assumption $\norm{\mb\zeta}4^6>C_{\star}\mu \kappa^2\norm{\mb\zeta}3^3$. Therefore, at least one entry of $\mb\zeta$ has large enough magnitude.
\end{proof}

Geometrically, the nontrivial entry $\zeta_i$ indicates the preference to corresponding column $\mb a_i$, as $\zeta_i=\innerprod{\mb a_i}{\mb q}$. Therefore, \Cref{lem:preference} implies that any stationary point $\mb q$ in $\mc R_{C_\star}$ should be close to at least one column of $\mb A$.

\subsubsection{Local Minima}
\label{sec:stat_localmin}

Suppose $\mb q\in\mc R_{C_\star}$ ($C_\star\ge 10$) is a stationary point and vector $\mb\zeta$ only has one nontrivial entry $\zeta_l$, then we can demonstrate that  the Riemannian Hessian $\Hess\varphi\paren{\mb q}$ is positive definite, and hence $\mb q$ is a local minimizer near $\mb a_l$.
\begin{lemma}
\label{lem:local-min}
Suppose $\mb q$ is a stationary point in $\mc R_{C_\star}$ with $C_\star\ge 10$, and $\mb \zeta=\mb A^T\mb q$ has only one entry $\zeta_l$ of magnitude no smaller than $2\mu\norm{\mb\zeta}3^3/\norm{\mb\zeta}4^4$. Then $\mb q$ is a local minimum near $\mb a_l$ and $\abs{\innerprod{\mb q}{\proj{\mb a_l}{\bb S}}}>1-2c_{\star}\kappa^{-2}$ with $c_{\star}=1/C_{\star}$.
\end{lemma}
\begin{proof}
Suppose $\mb\zeta$ has only one big entry $\zeta_l$, and other entries are bounded by $2\beta_l/\alpha_l$ 
\begin{align}
\norm{ \mb\zeta }4^4&=\zeta_l^4+\sum_{j\neq l}\zeta_j^4\\
&\le\zeta_l^4+\max_{j\neq l}\zeta_j^2\cdot\sum_{j\neq l}\zeta_j^2\\
&\le\zeta_l^4+\frac{4\mu^2\norm{\mb\zeta}3^6}{ \norm{\mb\zeta}4^8 },
\end{align}
with $\norm{\mb\zeta}4^6\ge C_{\star}\mu\kappa^2\norm{\mb\zeta}3^3$, and for simplicity let $c_{\star} = 1/C_{\star}$,  we have
\begin{equation}
\zeta_l^4\ge \norm{ \mb\zeta }4^4 - \frac{ 4\mu^2\norm{\mb\zeta}3^6 }{ \norm{\mb\zeta}4^8 }\ge\paren{1-4c_{\star}^2\kappa^{-4}}\norm{ \mb\zeta }4^4.
\end{equation}
On the other hand, we also have 
\begin{align}
\zeta_l^2&\le \paren{\sqrt{\alpha_l} + \frac{2\beta_l}{\alpha_l}}^2\\
&\le \frac{\norm{\mb\zeta}4^4}{\norm{\mb a_i}2^2}+\frac{4\mu\norm{\mb\zeta}3^3}{\norm{\mb a_i}2\norm{\mb\zeta}4^2}+\frac{4\mu^2\norm{\mb\zeta}3^6}{ \norm{\mb\zeta}4^8 }\\
&\le\frac{\norm{\mb\zeta}4^4}{\norm{\mb a_i}2^2}\paren{1+4c_{\star}\kappa^{-2}+4c_{\star}^2\kappa^{-4}}.
\end{align} 
Combining above two inequalities, we have
\begin{equation}
\zeta_l^2\le\frac{1+4c_{\star}\kappa^{-2}+4c_{\star}^2\kappa^{-4}}{1-4c_{\star}^2\kappa^{-4}}\frac{\zeta_l^4}{\norm{\mb a_i}2^2},
\end{equation}
thus the local minimum $\mb q$ is close to $\mb a_l$:
\begin{equation}\label{eqn:q_al_lb}
\frac{ \abs{\innerprod{\mb q}{\mb a_l}} }{ \norm{\mb a_l}2 }\ge\frac{\sqrt{1-4c_{\star}^2\kappa^{-4}}}{1+2c_{\star}\kappa^{-2}}\ge1-2c_{\star}\kappa^{-2}.
\end{equation}

Next, we need to verify that the Riemannian Hessian at $\bar{\mb q}$ is definite positive, recall that
\begin{equation}
\Hess\varphi\paren{\mb q}=-\mb P_{\mb q^{\perp}}\brac{3\mb A\diag(\mb\zeta^{\circ2})\mb A^T-\norm{\mb\zeta}4^4\mb I}\mb P_{\mb q^{\perp}}.
\end{equation}
Let $\mb v$ be a unit vector such that $\mb v\perp{\mb q}$, then 
\begin{align}
\lefteqn{\mb v^T\Hess\varphi\paren{\mb q}\mb v}\\
&=-\mb v^T\paren{3\mb A\diag(\mb\zeta^{\circ2})\mb A^T-\norm{\mb\zeta}4^4\mb I}\mb v\\
&=\norm{\mb\zeta}4^4-3\mb v^T\mb A\diag(\mb\zeta^{\circ2})\mb A^T\mb v\\
&=\norm{\mb\zeta}4^4-3\innerprod{\mb a_l}{\mb v}^2\zeta_l^2-3\sum_{i\neq l}\innerprod{\mb a_i}{\mb v}^2\zeta_i^2\\
&\ge\norm{\mb\zeta}4^4-3\innerprod{\mb a_l}{\mb v}^2\zeta_l^2-3\max_{i\neq l}\zeta_i^2.
\label{eqn:Hess_asympt_localmin}
\end{align}
The last inequality is due to $\sum_{i\neq l}\innerprod{\mb a_i}{\mb v}^2\le\norm{\mb A^T\mb v}2^2=1$. Since $\mb v\perp\bar{\mb q}$ and $\zeta_l$ is the only entry with nontrivial magnitude, then derive from \eqref{eqn:q_al_lb}: 
\begin{align}
\innerprod{\mb a_l}{\mb v}^2\zeta_l^2&\le 2c_{\star}\norm{\mb a_l}2^2\paren{\sqrt{\alpha_l}+\frac{2\beta_l}{\alpha_l}}^2\\
&\le 2c_{\star}\norm{\mb a_l}2^2 \cdot \paren{1+2c_{\star}}^2\alpha_l\\
&\le 2c_{\star}\paren{1+2c_{\star}^2}^2\norm{\mb\zeta}4^4,
\end{align}
and
\begin{align} 
\max_{i\neq l}\zeta_i^2&\le\frac{4\beta^2}{\alpha^2}\le\frac{4\mu^2\norm{\mb\zeta}3^6}{\norm{\mb\zeta}4^8}\le\frac{4c_{\star}^2\norm{\mb\zeta}4^{12}}{\norm{\mb\zeta}4^8}\le 4c_{\star}^2\norm{\mb\zeta}4^4.
\end{align}
Hence, the inequality $\mb v^T\Hess\varphi\paren{\mb q}\mb v\ge\paren{1-6c_{\star}-36c_{\star}^2-24c_{\star}^3}\norm{\mb\zeta}4^4$ holds for any $\mb v$ satisfying $\mb v\perp\mb q$, thus implies positive curvature along any tangent direction at such stationary point $\mb q$ when $C_{\star} \ge 10$.

\end{proof}

The lemma says if $\mb q$ is a stationary point in $\mb R_{C_\star}$ and $\mb q$ is only close to one column $\mb a_l$, then $\mb q$ is a local minimizer and satisfies $\abs{\innerprod{\mb q}{\proj{\mb a_l}{\bb S}}}>1-2c_{\star}\kappa^{-2}$ with $c_\star = 1/C_\star$.

\subsubsection{Saddle Points}
\label{sec:stat_saddle}
At last, if $\mb q\in\mc R_{C_\star}$ ($C_\star\ge 10$) is a stationary point and vector $\mb\zeta$ has more than one nontrivial entry. Denote any two nontrivial entries of $\mb\zeta$ with $\zeta_l$ and $\zeta_{l'}$, then we can prove that the Riemannian Hessian $\Hess\varphi\paren{\mb q}$ has negative curvature in the span of $\mb a_{l}$ and $\mb a_{l'}$, hence $\mb q$ is a saddle point. 

\begin{lemma}
\label{lem:saddles}
Suppose $\mb q$ is a stationary point in $\mc R_{C_\star}$ with $C_\star\ge 10$, and $\mb \zeta=\mb A^T\mb q$ has at least two entries $\zeta_l$ and $\zeta_{l'}$ with magnitude magnitude $\ge 2\mu\norm{\mb\zeta}3^3/\norm{\mb\zeta}4^4$, then the Riemannian Hessian at $\mb q$ has at least one negative eigenvalue and $\mb q$ is a saddle point.
\end{lemma}
\begin{proof} Suppose $\mb\zeta$ has at least two big entries $\zeta_l$ and $\zeta_{l'}$ satisfying
\begin{align}
\zeta_l^2&\ge\paren{ \sqrt{\alpha_l}-\frac{2\beta_l}{\alpha_l} }^2\\
&\ge\frac{ \norm{\mb\zeta}4^4 }{ \norm{ \mb a_l}2^2 } -\frac{4\mu\norm{\mb\zeta}3^3}{ \norm{\mb\zeta}4^2\norm{ \mb a_l}2 }+\frac{4\mu^2\norm{\mb\zeta}3^6}{\norm{ \mb\zeta }4^8}\\
&>\frac{ \norm{\mb\zeta}4^4 }{ \norm{ \mb a_l}2^2 } -\frac{4\mu\norm{\mb\zeta}3^3}{ \norm{\mb\zeta}4^2\norm{ \mb a_l}2 },
\end{align}
and $\zeta_{\ell'}$ likewise.
Since the nontrivial entry $\zeta_l=\innerprod{\mb a_l}{\mb q}$, and again let $c_{\star} = 1/C_{\star}$, it is easy to show that the norm of $\mb a_l$ is sufficiently large:
\begin{align}
\norm{\mb a_l}2^2\ge\zeta_l^2&\ge\paren{\sqrt{\alpha_l}-\frac{2\beta_l}{\alpha_l}}^2\\
&\ge\paren{1-2c_{\star}}^2\frac{ \norm{\mb\zeta}4^4 }{ \norm{\mb a_l}2^2 }\\
&\ge\paren{1-c_{\star}}^2C_{\star}^{2/3}\frac{ \mu^{2/3}\norm{\mb\zeta}3^2 }{ \norm{\mb a_l}2^2 },
\end{align}
or 
\begin{equation}
\norm{\mb a_l}2\ge \paren{1-c_{\star}}^{1/2}C_{\star}^{1/6}\mu^{1/6}\norm{\mb\zeta}3^{1/2}.
\end{equation}
Similar result holds for $\norm{\mb a_{l'}}2$, therefore
\begin{equation}
\frac{\mu}{ \norm{\mb a_l}2\norm{\mb a_{l'}}2}\le\frac{\mu^{2/3}}{C_{\star}^{1/3}\norm{\mb\zeta}3}\le\frac{C_{\star}^{-2/3}\norm{\mb\zeta}4^4}{C_{\star}^{1/3}\norm{\mb\zeta}3^3}\le c_{\star}.
\end{equation}

Now we are ready to show there exists a unit vector $\mb v$ such that $\mb v\in {\mathrm{span}}(\mb a_l,\mb a_{l'})$ and $\mb v\perp\mb q$, and the Hessian has negative curvature along such $\mb v$: 
\begin{align}
\lefteqn{\mb v^T\Hess \varphi(\mb q)\mb v}\nonumber\\
&=-3\mb v^T\mb A\diag(\mb\zeta^2)\mb A^T\mb v+\norm{\mb\zeta}4^4\\
&\le-3\mb v^T\paren{ \mb a_l\zeta^2_l\mb a^T_l+\mb a_{l'}\zeta^2_{l'}\mb a^T_{l'} }\mb v+\norm{\mb\zeta}4^4\\
&<-3\paren{\abs{\innerprod{\frac{\mb a_l}{\norm{\mb a_l}2}}{\mb v}}^2+\abs{\innerprod{\frac{\mb a_{l'}}{\norm{\mb a_{l'}}2}}{\mb v}}^2}\norm{\mb\zeta}4^4 + \frac{4\mu\norm{\mb\zeta}3^3}{ \norm{\mb\zeta}4^2 }\paren{ \norm{\mb a_l}2+\norm{\mb a_{l'}}2 } + \norm{\mb\zeta}4^4\\
&< -3\paren{1-\frac{\mu}{ \norm{\mb a_l}2\norm{\mb a_{l'}}2 }}\norm{\mb\zeta}4^4+ \frac{4\mu\norm{\mb\zeta}3^3}{ \norm{\mb\zeta}4^2 }\paren{ \norm{\mb a_l}2+\norm{\mb a_{l'}}2 } + \norm{\mb\zeta}4^4 \\
&\le\paren{-2+11c_{\star}}\norm{\mb\zeta}4^4. 
\label{eqn:Hess_asympt_saddle}
\end{align}
The third inequality is implied by \Cref{lem:cross} and is negative when $C_{\star} \ge 10$.

\end{proof}

This lemma says if the stationary point $\mb q$ has large inner product with any two columns $\mb a_l$ and $\mb a_{l'}$, then this $\mb q$ is a saddle point and the objective value decreases along the direction that breaks symmetry between $\mb a_l$ and $\mb a_{l'}$. The saddle point $\mb q$ can be seen as resulting from the competition between the two target solutions $\mb a_l$ and $\mb a_{l'}$.

\section{Large Sample Concentration}
\label{sec:finite}
In this section, we argue that the geometric characteristics of $\psi\paren{\mb q}$ are similar to those of $\varphi\paren{\mb q}$, by demonstrating that the critical points of the finite sample objective function $\psi(\mb q)$ are similar to those of the asymptotic objective function $\varphi( \mb q )$:
\begin{itemize}
\item {\bf Critical points are close.} The Riemannian gradient (\Cref{lem:grad_scale}) and Hessian (\Cref{lem:hess_scale}) concentrate, such that there is a bijection between critical points $\mb q_{\varphi}$ of $\varphi$ and critical points $\mb q_{\psi}$ of $\psi$, with $\norm{\mb q_{\varphi} - \mb q_{\psi} }{2}$ small. 

\item {\bf Curvature is preserved.} The Riemannian Hessian (\Cref{lem:hess_scale}) concentrates, such that $\mathrm{Hess}[\psi](\mb q_{\mathrm{fs}})$ has a negative eigenvalue if and only if $\mathrm{Hess}[\varphi]( \mb q_{\mathrm{pop}} )$ has a negative eigenvalue, and $\mathrm{Hess}[\psi](\mb q_{\mathrm{fs}})$ is positive definite if and only if $\mathrm{Hess}[\varphi]( \mb q_{\mathrm{pop}} )$ is positive definite. 

\end{itemize}

This implies that every local minimizer of the finite sample objective function is close to a preconditioned shift-truncation (\Cref{lem:geo_scale}). 

\begin{lemma}
If the following inequalities hold
\begin{align}
\norm{\grad[\psi]\paren{\mb q}-\frac{3\paren{1-\theta}}{\theta m^2}\grad[\varphi]\paren{\mb q}}2
&\le\frac{3c_{\star}}{2\kappa^2}\frac{1-\theta}{\theta m^2}\norm{\mb A^T\mb q}4^6,\\
\norm{\Hess[\psi]\paren{\mb q}-\frac{3\paren{1-\theta}}{\theta m^2}\Hess[\varphi]\paren{\mb q}}2 
&\le 3\paren{1-6c_{\star}-36c_{\star}^2-24c_{\star}^3}\frac{1-\theta}{\theta m^2}\norm{\mb A^T\mb q}4^4.
\end{align}
for all $\mb q\in\mc R_{2C_{\star}}$ with $C_{\star}\ge 10$ and $c_{\star} = 1/C_{\star}$, then any local minimum $\mb{\bar q}$ of $\psi\paren{\mb q}$ in $\mc R_{2C_{\star}}$ satisfies $\abs{\innerprod{\mb{\bar q}}{\mc P_{\bb S}\brac{\mb a_l}}}\ge 1-2c_{\star}\kappa^{-2}$ for some index $l$.
\label{lem:geo_scale}
\end{lemma}
\begin{proof}
Please refer to Appendix \ref{sec:main}.
\end{proof}

The Riemannian gradient and Hessian of the finite sample objective function $\psi\paren{\mb q}$ have similar expressions as those of the asymptotic objective function $\varphi( \mb q )$. Let $\mb\eta = \mb Y^T\paren{\mb Y\mb Y^T}^{-1/2}\mb q\in\bb S^{m-1}$. Then  
\begin{equation}
\psi\paren{\mb q}=-\frac1{4m}\norm{\mb Y^T\paren{\mb Y\mb Y^T}^{-1/2}\mb q}4^4=-\frac1{4m}\norm{\mb\eta}4^4,
\end{equation}
we calculate the Euclidean gradient and Hessian of the objective function 
\begin{align}
\nabla\psi\paren{\mb q}&=-\frac1m\paren{\mb Y\mb Y^T}^{-1/2}\!\mb Y\mb\eta^{\circ3},\\
\nabla^2\psi\paren{\mb q}&=-\frac3m\paren{\mb Y\mb Y^T}^{-1/2}\!\mb Y\diag(\mb\eta^{\circ2})\mb Y^T\paren{\mb Y\mb Y^T}^{-1/2}.
\end{align}
Similarly, the Riemannian gradient and Hessian have the form 
\begin{align}
\grad[\psi]\paren{\mb q}&=\mb P_{\mb q^{\perp}}\brac{\nabla\psi\paren{\mb q}}\\
&=-\frac1m\paren{\mb Y\mb Y^T}^{-1/2}\mb Y\mb\eta^{\circ3}+\frac1m\mb q\norm{\mb\eta}4^4, \label{eqn:grad1} \\
\Hess[\psi]\paren{\mb q}&=\mb P_{\mb q^{\perp}}\brac{\nabla^2\psi\paren{\mb q}-\innerprod{\nabla\psi\paren{\mb q}}{\mb q}\mb I}\mb P_{\mb q^{\perp}}\\
&=\mb P_{\mb q^{\perp}}\Big[\frac3m\paren{\mb Y\mb Y^T}^{-1/2}\!\mb Y\diag(\mb\eta^{\circ2})\mb Y^T\paren{\mb Y\mb Y^T}^{-1/2}+\frac1m\norm{\mb\eta}4^4\mb I\Big]\mb P_{\mb q^{\perp}}. \label{eqn:hess1}
\end{align}

Since $\mb Y = \mb A_0 \mb X_0$, we can see that the Riemannian gradient and Hessian are (complicated) functions of the random circulant matrix $\mb X_0$. Although the entries of the vector $\mb x_0$ are probabilistically independent, the entries of $\mb X_0$ are dependent random variables. To remove the dependence within the random circulant matrix $\mb X_0$, we break $\mb X_0$ into submatrices $\mb X_1, \dots, \mb X_{2k-1}$ that
\begin{equation}
\mb X_i=\brac{\mb x_{i},\mb x_{i+\paren{2k-1}},\cdots,\mb x_{i+\paren{m-2k-1}}}.
\end{equation}
Each of which is (marginally) distributed as a $\paren{2k-1} \times \frac{m}{2k-1}$ i.i.d.\ $\mathrm{BG}(\theta)$ random matrix. Indeed, there exists a permutation $\mb \Pi$ such that 
\begin{equation}
\mb X_0 \mb\Pi=\brac{\mb X_1,\mb X_2,\cdots, \mb X_{2k-1}}.
\end{equation} 
A detailed analysis of \eqref{eqn:grad1}-\eqref{eqn:hess1} (see \Cref{sec:grad_scale} and \Cref{sec:hess_scale} in the Appendix) allows us to control the finite sample fluctuations of the gradient and Hessian in terms of analogous quantities for each $\mb X_i$. Because the $\mb X_i$ are i.i.d., they are amenable to standard tools from measure concentration. Taking a union bound over $i$, we show that the gradient (Lemma \Cref{lem:grad_scale}) and hessian (Lemma \Cref{lem:hess_scale}) concentrate as desired:

\begin{lemma}
\label{lem:grad_scale}
Suppose $\mb x_0\simiid\mathrm{BG}\paren{\theta}\in\R^m$. There exists positive constant $C$ that whenever
\begin{equation}
m\ge C\frac{\min\set{\paren{2C_{\star}\mu}^{-1},\kappa^2k^2}}{\paren{1-\theta}^2\sigma^2_{\min}}\kappa^8k^4\log^3\paren{\kappa k},
\end{equation}
and $\theta\ge1/k$, then with probability no smaller than $1-\exp\paren{-k}-\theta^2\paren{1-\theta}^2k^{-4}-2\exp\paren{-\theta k} -48k^{-7}-48m^{-5}- 24k\exp\paren{-\tfrac1{144}\min\set{k,3\sqrt{\theta m}}}$,
\begin{equation}
\norm{\grad[\psi]\paren{\mb q}-\frac{3(1-\theta)}{\theta m^2}\grad[\varphi]\paren{\mb q}}2\!\le c\frac{1-\theta}{\theta m^2} \frac{\norm{\mb A^T\mb q}4^6}{\kappa^2 },
\end{equation}
holds for all $\mb q\in\hat{\mc R}_{2C_{\star}}$ with $c\le 3/\paren{2C_\star}\le\frac3{20}$.

\end{lemma}
\begin{proof}
Please refer to section \ref{sec:grad_scale}.
\end{proof}

\begin{lemma}
\label{lem:hess_scale}
Suppose $\mb x_0\simiid\mathrm{BG}\paren{\theta}$. There exists positive constant $C$ that whenever
\begin{equation}
m\ge C\frac{\min\set{\paren{2C_\star\mu\kappa^2}^{-4/3}\!,k^2}}{\paren{1-\theta}^2\sigma^2_{\min}}\kappa^6 k^4\log^3\paren{\kappa k},
\end{equation}
and $\theta\ge1/k$, then with probability no smaller than $1-\exp\paren{-k}-\theta^2\paren{1-\theta}^2k^{-4}-2\exp\paren{-\theta k} -48k^{-7}-48m^{-5}- 24k\exp\paren{-\tfrac1{144}\min\set{k,3\sqrt{\theta m}}}$,
\begin{equation}
\norm{\Hess[\psi]\paren{\mb q}-\frac{3(1-\theta)}{\theta m^2}\Hess[\varphi]\paren{\mb q}}2\!\le c\frac{1-\theta}{\theta m^2}\norm{\mb A^T\mb q}4^4,
\end{equation}
holds for all $\mb q\in\hat{\mc R}_{2C_{\star}}$ with positive constant $c\le0.048\le3\paren{1-6c_{\star}-36c_{\star}^2-24c_{\star}^3}$.
\end{lemma}

\begin{proof}
Please refer to section \ref{sec:hess_scale}.
\end{proof}

\section{Experiments}
\label{sec:exp}
\subsection{Properties of a Random Kernel} Our results are stated in terms of several parameters, including the condition number $\kappa$ of $\mb A_0$ and the column coherence of $\mb A$. In \Cref{fig:para_A0}, we demonstrate the typical values of $\sigma_0$, $\kappa$, and $\mu$ for generic unit-norm kernels of varying dimension $k=10,20,\cdots,1000$.

From this figure, for a generic unit-norm kernel, we have following estimates:
\begin{align}
&\sigma_0\approx\log^{-1}\paren{k},\\
&\kappa\approx\log^{4/3}\paren{k},\\
&\mu\approx\sqrt{\log\paren{k}/k}.
\end{align}
\begin{figure*}[h!]
\centerline{\input{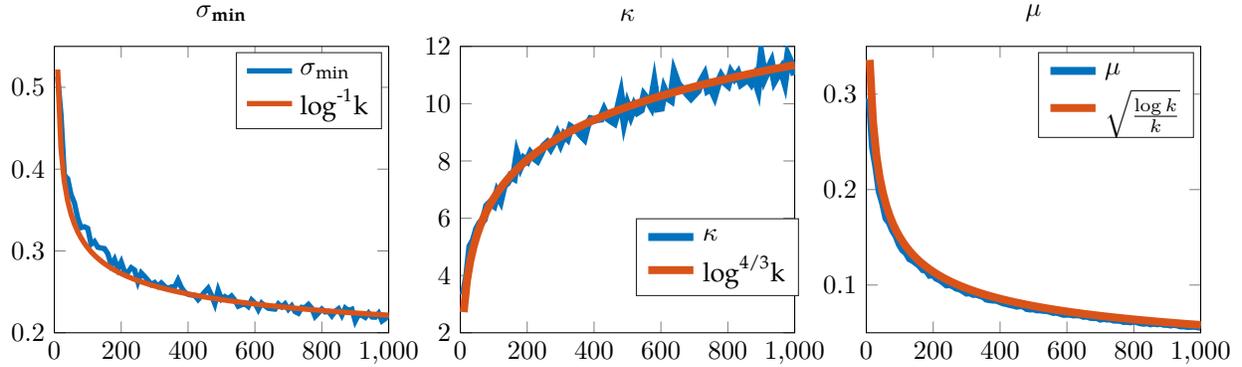}}
\caption{{\bf Average of Parameters $\sigma_{\min}$, $\kappa$, and $\mu$} of a random unit norm kernel $\mb a_0$ over $50$ independent trials, as a function of dimension $k$.}
\label{fig:para_A0}
\end{figure*} 
On the other hand, if the kernel $\mb a_0$ is bandpass, then both $\kappa$ and $\mu$ are larger. In this situation, our results require more observations $m$ and smaller sparsity rate $\theta$.

\subsection{Recovery Accuracy of Local Minima}
We next investigate the performance of \Cref{alg:ssbd} under varying settings. We define the recover error as $\mathrm{err}=1-\max_{\tau}\abs{\innerprod{\bar{\mb a}}{\proj{\injector_k^*\shift{\extend{\mb a_0}}{\tau}}{\bb S}}}$, and calculate the average error from 50 independent experiments.  
In \Cref{fig:exp_local}, the left figure plots the average error when we fix the kernel size $k=50$, and vary the dimension $m$ and the sparsity $\theta$ of $\mb x_0$.\footnote{Note that the $x$-axis is indexed with overlapping ratio $k\cdot\theta$, which indicates how many times the kernel $\mb a_0$ present in a $k$-length window of $\mb y$ on average.} The right figure plots the average error when we vary the dimensions $k,m$ of both convolution signals, and set the sparsity as $\theta = k^{-2/3}$. 
\begin{figure}[h]
\centering
\begin{tabular}{p{0.4\textwidth} p{0.4\textwidth}}
\includegraphics[trim = {.5cm 6cm 1.5cm 6cm}, clip, width=0.4\textwidth]{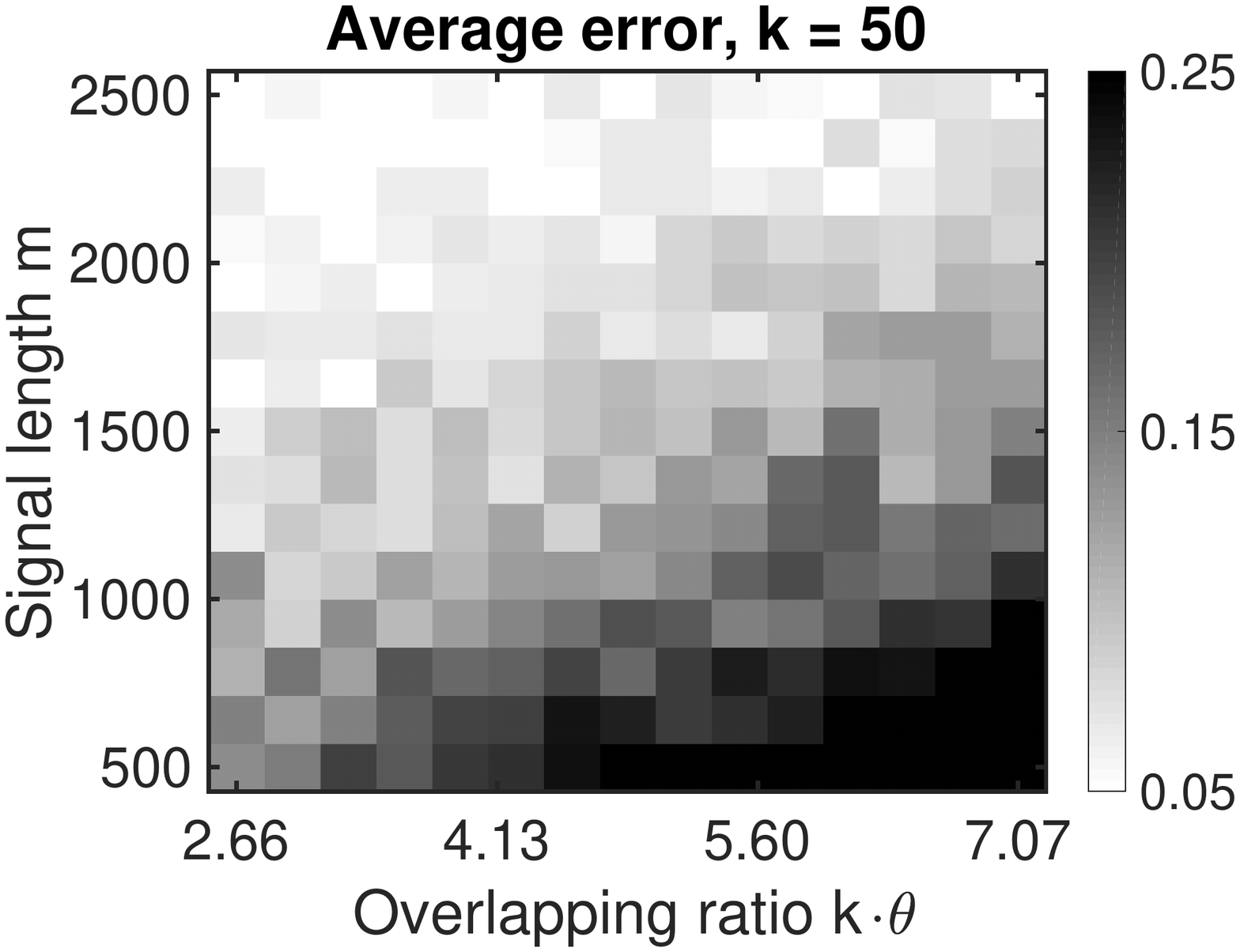} &
\includegraphics[trim = {.5cm 6cm 1.5cm 6cm}, clip, width=0.4\textwidth]{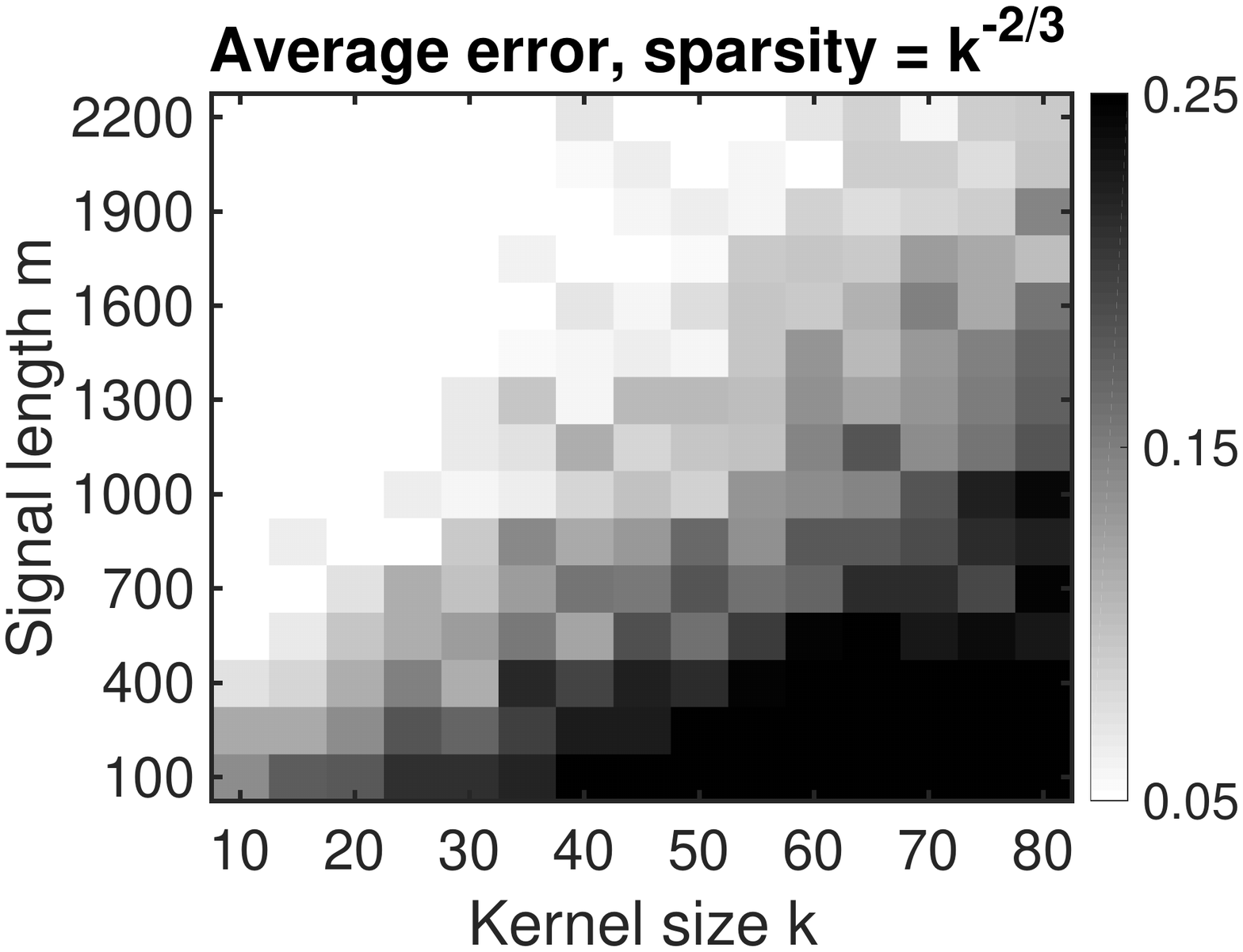}
\end{tabular}
\caption{{\bf Recovery Error} of the Shift Truncated Kernel of \Cref{alg:ssbd}.}
\label{fig:exp_local}
\end{figure}

\noindent This figure agrees with the theory developed in this paper: when the activation coefficient $\mb x_0$ is long and sparse (large $m$ and small $\theta$), the algorithm obtains a closer estimate of a shift-truncation of the ground truth.

\subsection{Recovery Accuracy of the Ground Truth Kernel}
In this section, we provide experiment results for the recovery of the ground truth kernel obtained by the annealing algorithm proposed in \cite{Zhang2017-CVPR}. The annealing algorithm recovers the ground truth kernel by minimizing the Lasso cost in \eqref{eqn:bd_lasso}, initialized at the zero-padded shift truncated kernel rendered from \Cref{alg:ssbd}. The recovery accuracy presented in \Cref{fig:exp_global} is measured as $\mathrm{err}=\min_{\tau}\norm{\bar{\mb a}^{(+)}\pm\shift{\extend{\mb a_0}}{\tau}}2$. Here, $\bar{\mb a}^{(+)}$ denote the local minimum in the lifted optimization space.

\begin{figure}[h]
\centering
\begin{tabular}{p{0.4\textwidth} p{0.4\textwidth}}
\includegraphics[width=0.4\textwidth]{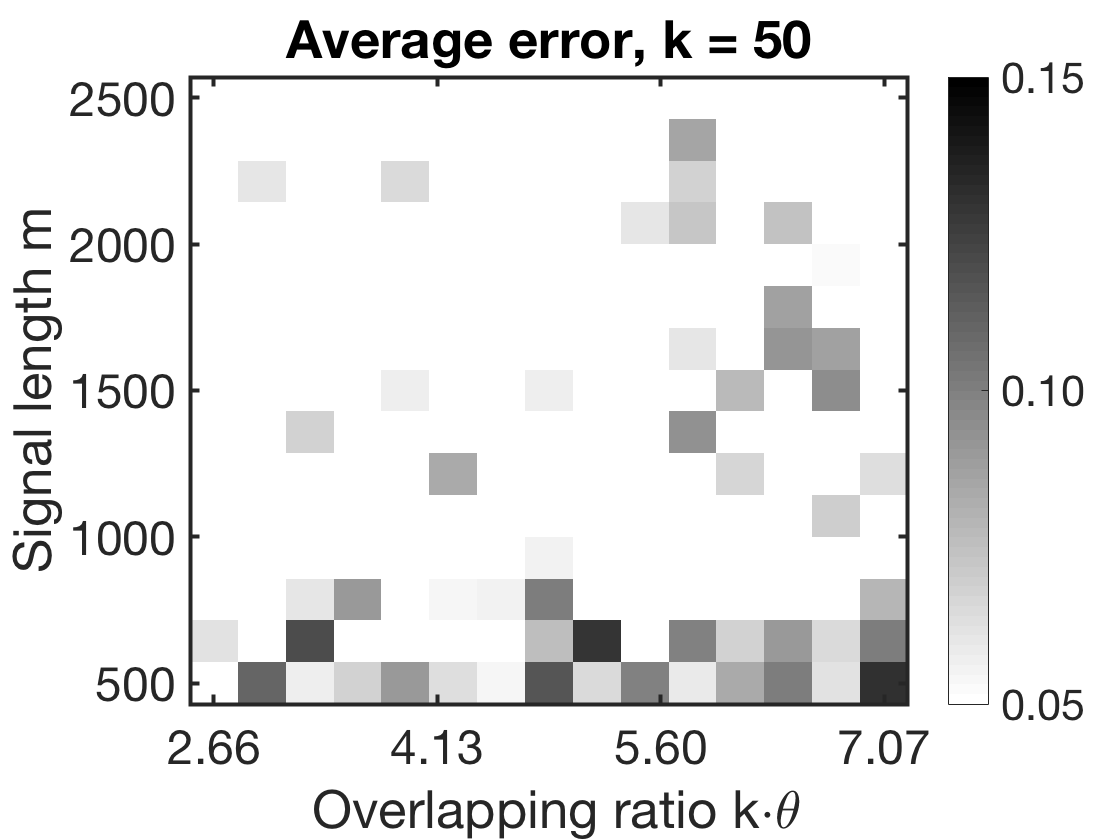} &
\includegraphics[width=0.4\textwidth]{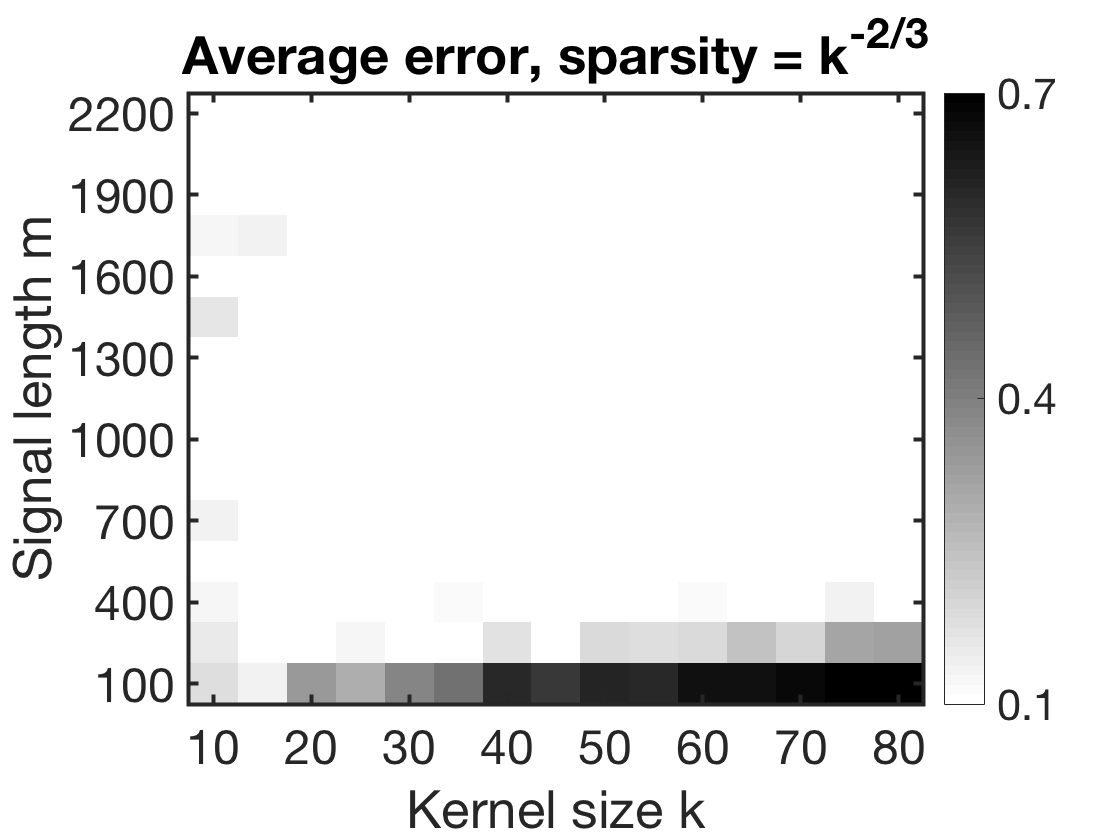}
\end{tabular}
\caption{{\bf Recovery Error of the Ground Truth Kernel} with  \Cref{alg:ssbd} finding a shift truncated kernel and the annealing Lasso problem recovering the ground truth kernel.} 
\label{fig:exp_global}
\end{figure}

For comparison, we also present experiment results of the algorithm proposed by \cite{Zhang2017-CVPR}, which is composed of solving two Lasso minimization problems over the original kernel sphere and lifted kernel sphere respectively. 

\begin{figure}[h]
\centering
\begin{tabular}{p{0.4\textwidth} p{0.4\textwidth}}
\includegraphics[width=0.4\textwidth]{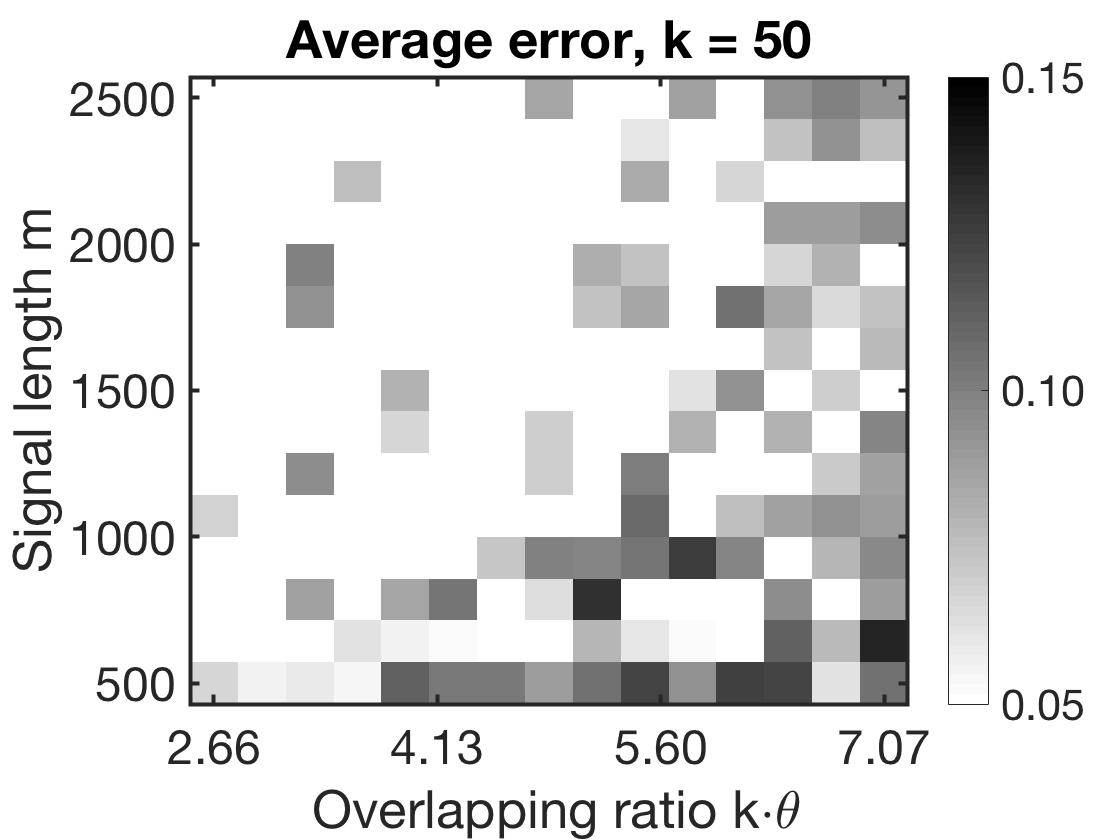} &
\includegraphics[width=0.4\textwidth]{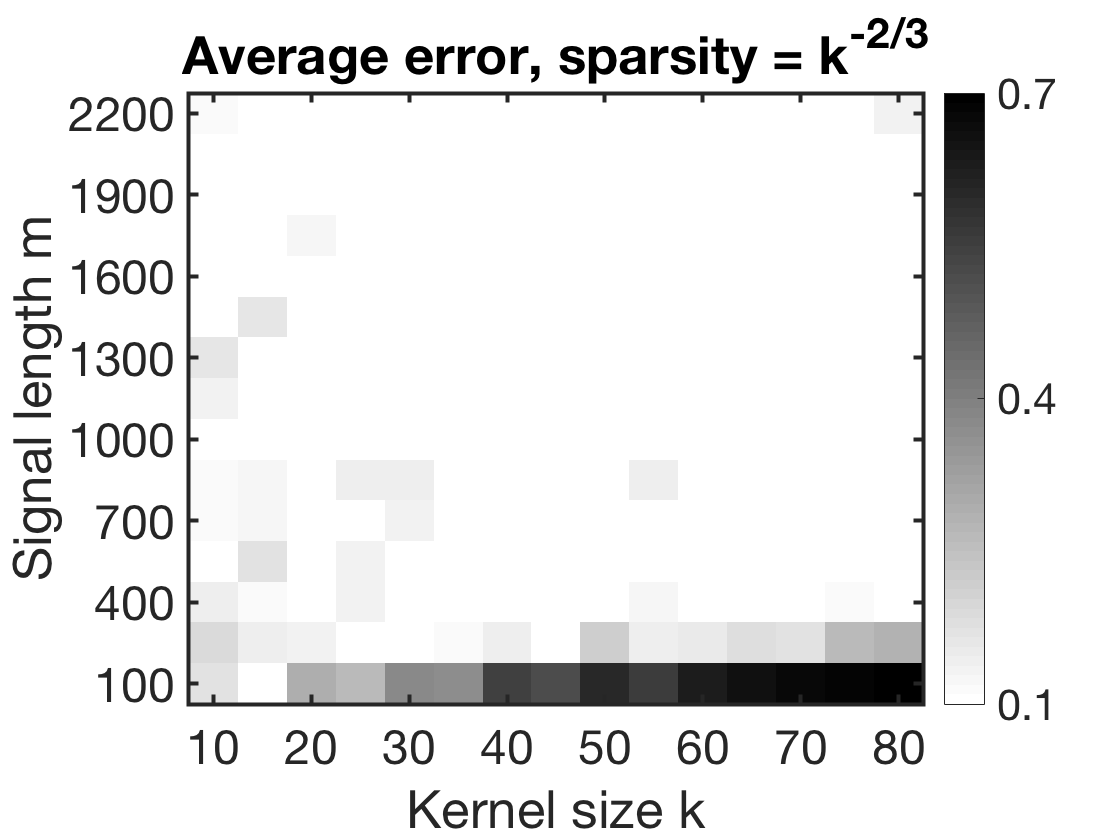}
\end{tabular}
\caption{{\bf Recovery Error of the Ground Truth Kernel} by minimizing the Lasso objective function recovering both the shift truncated kernel as well as the ground truth kernel.}
\label{fig:exp_global}
\end{figure}

In terms of the recovery accuracy of the ground truth kernel, \Cref{alg:ssbd} proposed in this paper achieves better recovery for sparser and longer observations, while the \cite{Zhang2017-CVPR} manifests slight advantages when the observations is limited. As the optimization landscape studied in \cite{Zhang2017-CVPR} varies with different choice of sparsity parameter $\lambda$, it is possible that experiment results for \cite{Zhang2017-CVPR} could be improved. On the other hand, only empirical knowledge about the choice of $\lambda$ is available while there is little disciplined understanding. In contrast,  \Cref{alg:ssbd} does not depend on any parameter tuning and guarantees recovery once the working conditions are met.

\section{Discussions}
Finally, we provide some comments about the results and proof strategy presented in this paper, and discuss directions for future research.

This paper casts the sparse blind deconvolution problem as finding a {\em spiky} vector in a subspace and studies its optimization landscape. We prove that the geometric property that {\em any local solution is close to a shift-truncation of the ground truth kernel} holds on a sub-level set of the sphere. This holds even when the observation contains densely overlapping copies of the true kernel. In addition, we propose a simple initialization scheme such that any descent algorithm that escapes strict saddles  can recover the local minimum, which is a near shift-truncation of the ground truth kernel.

\paragraph{Sample Complexity.} The sample complexity shown in this paper $m\sim k^{6}$ is suboptimal. Our proofs relies heavily on ``worst case'' tools such as the triangle inequality, multiplication of operator norm, and union bound. In particular, we believe that the sample complexity can be improved by replacing the sample splitting argument in Section \Cref{sec:grad_scale} and \Cref{sec:hess_scale} in the Appendix with more sophisticated arguments based on decoupling (see also \cite{QZEW17-pp}).

\paragraph{Global Geometry.} The theoretical results presented in this paper demonstrate that ``all local optima are benign" in the sub-level set $\mc R_{C_\star}$. Our empirical results suggest that this is a property holds over the whole sphere. Proving this could be challenging, as our characterization of the saddle points only applies when $\norm{\mb \zeta}{4}^4$ is large. 
It would be exciting to see if further research investigating other techniques for nonconvex optimization problems could be motivated by our current work.

\paragraph{Convolutional Dictionary Learning.} This is a natural and practical extension of blind deconvolution, where the observation is the superposition of several convolutions. The empirical observations and algorithm proposed in \cite{Zhang2017-CVPR} hold in this more challenging situation. It would be interesting to develop efficient and provable algorithms for convolutional dictionary learning based on the $\ell^4$ formulation.

\section*{Acknowledgement}
The authors gratefully acknowledge support from NSF 1343282, NSF CCF 1527809, and NSF IIS 1546411. It is a great pleasure to acknowledge conversations with Yenson Lau, Sky Cheung, and Abhay Pasupathy.

{\small
\bibliographystyle{alpha}
\bibliography{deconv}
}

\newpage
\appendix
\section*{Appendix}
 \Cref{sec:basics} contains some basic lemmas for quantities used repeatedly;  \Cref{sec:main} presents the proofs of the main theorem and corollary of this paper.  \Cref{sec:init} and \Cref{sec:preconditioning} provide proofs supporting the initialization point $\qinit$ and the preconditioning term $\mb Y^T\mb Y$ (or $\mb A_0^T\mb A_0$) respectively. Finite sample concentration for the Riemannian gradient and Hessian are presented in \Cref{sec:grad_scale} and \Cref{sec:hess_scale} respectively.

\begin{figure*}
\centering
\includegraphics[trim = {6cm 0 6cm 0}, clip, width=1\textwidth]{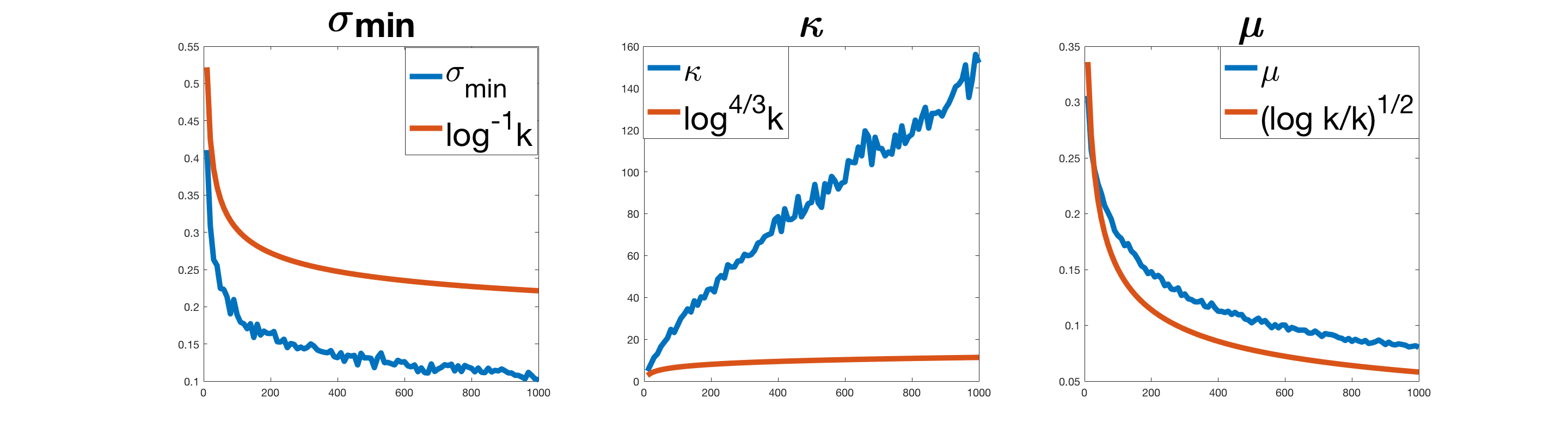}
\caption{{\bf Average of Parameters $\sigma_{\min}$, $\kappa$, and $\mu$} of a {\em band-pass} unit norm kernel $\mb a_0$ over $50$ independent trials, as a function of dimension $k$.}
\label{fig:para_bandpass_A0}
\end{figure*} 

\section{Basics}
\label{sec:basics}
\begin{lemma}[Expectation of the Approximate Objective Function]
\label{lem:obj_exp} 
Assuming $\mb x_0\simiid\mathrm{BG}\paren{\theta}\in\R^m$, then
\begin{align}
\bb E_{\mb x_0}\brac{ \frac1m\norm{\mb Y^T\paren{\mb A_0\mb A_0^T}^{-1/2}\mb q}4^4 }
=3\theta\paren{1-\theta}\norm{\mb A^T\mb q}4^4+3\theta^2\norm{\mb A^T\mb q}2^4.
\end{align}
\end{lemma}
\begin{proof}
Let $\mb g\in\R^{2k-1}$ be a standard random Gaussian vector and $\mb P_{I}$ be the projection operator onto Bernoulli vector $I\sim \mathrm{Ber}(\theta)$. Then any column $\mb x_i\in\R^{2k-1}$ of $\mb X_0$ is equal in distribution to $\mb x_i = \mb P_{I}\mb g$ with $\mb g\simiid\mc N\paren{0,1}$.
\begin{align}
\lefteqn{\bb E_{\mb x_0}\brac{ \frac1m\norm{\mb Y^T\paren{\mb A_0\mb A_0^T}^{-1/2}\mb q}4^4 }}\nonumber\\
&=\frac1m\bb E_{I}\bb E_{\mb g}\norm{\mb q^T\mb A\mb X_0}4^4\\
&=\bb E_{I}\bb E_{\mb g}\norm{\mb q^T\mb A\mb x_i}4^4\\
&=\bb E_{I}\bb E_{\mb g}\paren{\mb q^T\mb A\mb P_{I}\mb g}^4\\
&=3\bb E_{I}\paren{\mb q^T\mb A\mb P_{I}\mb A^T\mb q}^2\\
&=3\bb E_I\paren{\sum_{i\in I}\innerprod{\mb a_i}{\mb q}^4 + \!\!\!\sum_{\set{i\neq j} \in I}\!\!\innerprod{\mb a_i}{\mb q}^2\innerprod{\mb a_j}{\mb q}^2}  \\ 
&=3\theta\paren{1-\theta}\norm{\mb A^T\mb q}4^4+3\theta^2\norm{\mb A^T\mb q}2^4
\end{align}
\end{proof}

\begin{lemma}[Root Estimation for Cubic Gradient Function]
\label{lem:cubic} 
Consider an equation of the form
\begin{equation}\label{eqn:cubic}
f\paren{x} = x \paren{ \alpha - x^2 } - \beta = 0,
\end{equation}
with $\alpha > 0$. Suppose that $\beta < \tfrac{1}{4} \alpha^{3/2}$. Then $f\paren{x} = 0$ has three solutions, $x_1, x_2, x_3$ satisfying 
\begin{align}
\max\set{ \left|x_1 - \sqrt{\alpha} \right|, \left|x_2 + \sqrt{\alpha} \right|, \left |x_3  \right | } \;\le\; \frac{2 \beta}{\alpha}.
\end{align}
\end{lemma}
\begin{proof} 
Suppose first that $\beta > 0$. Then $f\paren{0} < 0$. Moreover,  
\begin{align}
f\paren{ \tfrac{2 \beta}{\alpha} } &= 2 \beta - 8 \beta^3 / \alpha^3 - \beta \\
 &= \beta \paren{ 1 - 8 \beta^2 / \alpha^3 } \\
 &> 0.
\end{align}
Hence, $f$ has at least one root in the interval $\brac{0, \tfrac{2\beta }{\alpha }}$. Similarly, notice that $f\paren{ \sqrt{\alpha} } < 0$ and that  
\begin{align}
\lefteqn{f\paren{ \sqrt{\alpha} - \tfrac{2 \beta}{\alpha} }}\nonumber\\
  &= \alpha^{3/2} - 2 \beta - \paren{ \sqrt{\alpha} - 2 \beta / \alpha }^3 - \beta  \\
 &= \alpha^{3/2} - 3 \beta - \alpha^{3/2} + 6 \beta - 12 \beta^2 / \alpha^{3/2}  + 8 \beta^3/ \alpha^3 \\
 &= \beta\paren{ 3 - \frac{12 \beta }{ \alpha^{3/2} } + \frac{8 \beta^2}{\alpha^3}} \\ 
 & > 0.
\end{align}

Thus, there is at least one root in the interval $\brac{\sqrt{\alpha} - \tfrac{2\beta}{\alpha}, \sqrt{\alpha} }$. Finally, note that $f\paren{ - \sqrt{\alpha} } < 0$, $\frac{df}{dx}\paren{ - \sqrt{\alpha} } = - 2 \alpha$, and $\frac{d^2 f}{d x^2} \paren{ x' } = - 3 x'$ is positive for $x' \le -\sqrt{\alpha}$. Hence, convexity gives that 
\begin{align}
&f\paren{ - \sqrt{\alpha} - \tfrac{2\beta}{\alpha} } \nonumber\\
&\ge f\paren{ - \sqrt{\alpha} } + \frac{d f}{dx} \paren{ -\sqrt{\alpha} } \times \paren{ - 2 \beta / \alpha }  \\
 &= - \beta + \paren{-2 \alpha} \times \paren{ - 2 \beta / \alpha }  \\
 &= 3 \beta \\
 &> 0. 
 \end{align} 
Under this condition, there is at least one root in the interval, 
$\brac{ - \sqrt{\alpha } - 2 \beta/ \alpha, - \sqrt{\alpha }}$. These three intervals do not overlap, as long as $\tfrac{4 \beta}{\alpha} < \sqrt{\alpha}$, or $\beta < \tfrac14{\alpha^{3/2} }$. 

In the case that $\beta \le 0$, a symmetric argument applies. Thus there are exactly three solutions to equation \eqref{eqn:cubic} in the specified intervals. 
\end{proof}

\begin{lemma}
\label{lem:cross}
Let $\mb a_l$ and $\mb a_{l'}$ be two nonzero vectors with inner product $\mu_{l,l'}\doteq\innerprod{\mb a_l}{\mb a_{l'}}$. Then for any unit vector $\mb v \in \mathrm{span}\paren{ \mb a_{l}, \mb a_{l'} }$,
\begin{equation}
\abs{\innerprod{\frac{\mb a_l}{\norm{\mb a_l}2}}{\mb v}}^2+\abs{\innerprod{\frac{\mb a_{l'}}{\norm{\mb a_{l'}}2}}{\mb v}}^2\ge 1-\frac{\abs{\mu_{l,l'}}}{\norm{\mb a_l}2\norm{\mb a_{l'}}2}.
\end{equation}
\end{lemma}

\newcommand{\murel}{\mu_{\mathrm{rel}}}

\begin{proof} Let $\mb u$ and $\mb u^\perp$ be two orthogonal unit vectors, such that  
\begin{align}
&\mb a_l=\norm{\mb a_l}2\mb u,\\
&\mb a_{l'}=\frac{\mu_{l,l'}}{\norm{\mb a_l}2}\mb u+\sqrt{\norm{\mb a_{l'}}2^2-\frac{\mu^2_{l,l'}}{\norm{\mb a_l}2^2}}\mb u^{\perp}.
\end{align}
Suppose $\mb v=a\mb u+b\mb u^{\perp}$ with $a^2+b^2=1$. Let $\murel = \frac{\mu_{l,l'}}{\norm{\mb a_l}2\norm{\mb a_{l'}}2}$, then we can expand the quantity of interests as 
\begin{align}
\lefteqn{\abs{\innerprod{\frac{\mb a_l}{\norm{\mb a_l}2}}{\mb v}}^2+\abs{\innerprod{\frac{\mb a_{l'}}{\norm{\mb a_{l'}}2}}{\mb v}}^2} \nonumber \\
&=\abs{\innerprod{\mb u}{a\mb u+b\mb u^{\perp}}}^2 +\abs{\innerprod{\murel\mb u+\sqrt{1-\murel^2}\mb u^{\perp}}{a\mb u+b\mb u^{\perp}}}^2 \\
&=a^2+\paren{a\murel+b\sqrt{1-\murel^2}}^2\\
&=a^2+b^2+\paren{a^2-b^2}\murel^2+2ab\murel\sqrt{1-\murel^2} \\
&= 1 + \left[ a^2 - b^2, 2 a b \right] \left[ \murel^2, \murel \sqrt{ 1 - \murel^2 } \right]^T  
\end{align}
Since $\left[ a^2 - b^2, 2 a b \right]$ is a unit vector, then above equation is lower bounded by 
\begin{align}
1-\norm{\brac{\mu^2_{rel},\mu_{rel}\sqrt{1-\mu^2_{rel}}}}2
&=1-\abs{\mu_{rel}}\\
&=1-\frac{\abs{\mu_{l,l'}}}{\norm{\mb a_l}2\norm{\mb a_{l'}}2}
\end{align}
as claimed. 
\end{proof}

\begin{lemma}[Nonzeros in a Bernoulli Vector]
\label{lem:ber_sparsity}
Let $\mb v\simiid \mathrm{Ber}\paren{\theta}\in\R^n$, then 
\begin{equation}
\prob{\norm{\mb v}0\ge \paren{1+t}\theta n}\le2\exp\paren{-\frac{3t^2}{2t+6}\theta n}.
\end{equation}
\end{lemma}
\begin{proof}
As $\norm{\mb v}0=v_0+\cdots+v_{n-1}$, and
\begin{equation}
\abs{v_i-\theta}\le1,\quad\bb E\brac{\paren{v_i-\theta}^2}=\theta\paren{1-\theta}\le\theta
\end{equation}
with Bernstein's inequality, we obtain that
\begin{align}
\prob{\norm{\mb v}0\ge \paren{1+t}\theta n}
&\le 2\exp\paren{-\frac{t^2\theta^2n^2}{2\paren{\theta-\theta^2} n+\frac23t\theta n}}\\ 
&\le 2\exp\paren{-\frac{3t^2}{2t+6}\theta n},
\end{align}
as claimed.
\end{proof}

\begin{lemma}[Entry-wise Truncation of a Bernoulli Gaussian Vector]
\label{lem:truncate}
Suppose $\mb x_0\simiid \mathrm{BG}\paren{\theta}\in\R^m$, then
\begin{equation}
\prob{\norm{\mb x_0}{\infty}>t} \le2\theta me^{-t^2/2}.
\end{equation}

\end{lemma}
\begin{proof}
A Bernoulli-Gaussian variable $x = \omega \cdot g$ satisfies
\begin{equation}
\prob{\abs{x}\ge t} =\theta\cdot\prob{\abs{g}\ge t}\le 2\theta e^{-t^2/2},
\end{equation}
Taking a union bound over the $m$ entries of $\mb x_0$, we obtain
\begin{align}
\prob{\norm{\mb x_0}{\infty}>t} &\le m \prob{\abs{x}> t} \\
&\le 2\theta me^{-t^2/2},
\end{align}
as claimed. 
\end{proof}

\begin{lemma}[Operator Norm of a Bernoulli Gaussian Circulant Matrix]
\label{lem:X_norm}
Let $\mb C_{\mb x_0}\in\R^{m\times m}$ be the circulant matrix generated from $\mb x_0\simiid \mathrm{BG}\paren{\theta}\in\R^m$, then 
\begin{equation}
\prob{\norm{\mb C_{\mb x_0}}2\ge t}\le 2m\exp\paren{-\frac{t^2}{2\theta m+2t}}.
\end{equation}
\end{lemma}

\begin{proof}The operator norm of a circulant matrix is
\begin{equation}
\norm{\mb C_{\mb x_0}}2=\max_{l}\abs{\innerprod{\mb x_0}{\mb w_l}},
\end{equation}
where $\mb w_l$ is the $l$-th (discrete) Fourier basis vector
\begin{equation}
\mb w_l=\brac{1,~e^{l\frac{2\pi j}{m}},\cdots,~e^{l\paren{m-1}\frac{2\pi j}{m}}}^T,\quad l=0,\cdots,m-1,
\end{equation}
and $j$ is the imaginary unit.
With moment control Bernstein inequality, we obtain
\begin{align}
\prob{\abs{\innerprod{\mb x_0}{\mb w_l}}\ge t}&\le2\exp\paren{-\frac{t^2}{2\theta\norm{\mb w_l}2^2+2\norm{\mb w_l}{\infty}t}}\nonumber\\
&\le2\exp\paren{-\frac{t^2}{2\theta m+2t}}
\end{align}
together with the union bound, 
\begin{align}
\prob{\norm{\mb C_{\mb x_0}}2\ge t}&\le m\prob{\abs{\innerprod{\mb x_0}{\mb w_l}}\ge t}\\
&\le 2m\exp\paren{-\frac{t^2}{2\theta m+2t}},
\end{align}
as claimed. 
\end{proof}

\begin{lemma}[Norms of $\mb\eta$ and $\mb{\bar\eta}$]
\label{lem:eta_norms}
Suppose $\delta=\norm{\frac1{\theta m}\mb X_0\mb X_0^T-\mb I}2 \le 1/\paren{2\kappa^2}$, then vectors $\mb \eta = \mb Y^T \paren{ \mb Y \mb Y^T }^{-1/2} \mb q$ and $\mb{\bar\eta}=\mb Y^T\paren{\theta m\mb A_0\mb A_0^T}^{-1/2}\mb q$ satisfy 
\begin{align}
\norm{\mb\eta}{\infty}&\le\paren{1+\frac{4\kappa^3\delta}{\sigma_{\min}}}\paren{\frac{2k}{\theta m}}^{1/2}\!\norm{\mb x_0}{\infty},\\ 
\norm{\mb{\bar\eta}}{\infty}&\le\paren{\frac{2k}{\theta m}}^{1/2}\!\norm{\mb x_0}{\infty},\\
\norm{\mb\eta}6^6&\le\paren{1+\frac{4\kappa^3\delta}{\sigma_{\min}}}^4\frac{4k^2}{\theta^2m^2}\norm{\mb x_0}{\infty}^4,\\
\norm{\mb{\bar\eta}}{2}&\le 1+\delta/2,\\
\norm{\mb\eta-\mb{\bar\eta}}{\infty}&\le\frac{4\kappa^3\delta}{\sigma_{\min}}\paren{\frac{2k}{\theta m}}^{1/2}\!\norm{\mb x_0}{\infty},\\
\norm{\mb\eta-\mb{\bar\eta}}2&\le\paren{1+\delta/2}\frac{4\kappa^3\delta}{\sigma_{\min}}.
\end{align}

\end{lemma}

\begin{proof} Since $\delta=\norm{\frac1{\theta m}\mb X_0\mb X_0^T-\mb I}2$, then
\begin{align}
\norm{\mb X_0}2&\le\paren{\theta m}^{1/2}\sqrt{1+\delta}\\
&\le \paren{\theta m}^{1/2}\paren{1+\delta/2}.
\end{align}
As $\mb \eta = \mb Y^T \paren{ \mb Y \mb Y^T }^{-1/2} \mb q = \mb X_0^T\mb A_0^T \paren{ \mb Y \mb Y^T }^{-1/2} \mb q$, together with \Cref{lem:precond_neghalf_2}:
\begin{align}
&\norm{\mb A_0^T \paren{ \mb Y \mb Y^T }^{-1/2} \mb q}{\infty}\nonumber\\
&\le\norm{\mb A_0^T \paren{ \mb Y \mb Y^T }^{-1/2} \mb q}2 \\
&\le\norm{\mb A_0^T\paren{\paren{\mb Y\mb Y^T}^{-1/2}-\paren{\theta m \mb A_0\mb A_0^T}^{-1/2}}\mb q}2 +\norm{\mb A_0^T\paren{\theta m \mb A_0\mb A_0^T}^{-1/2}\mb q}2 \\
&\le\paren{\theta m}^{-1/2}\frac{4\kappa^3\delta}{\sigma_{\min}}\norm{\mb q}2+\paren{\theta m}^{-1/2}\norm{\mb A^T\mb q}2 \qquad \\
&\le\paren{\theta m}^{-1/2}\paren{1+\frac{4\kappa^3\delta}{\sigma_{\min}}}
\end{align}

\noindent{\bf Norms of $\mb\eta$.} Since $\norm{\mb X_0\mb e_l}2\le\sqrt{2k-1}\norm{\mb X_0\mb e_l}{\infty}$, we have
\begin{align}
\norm{\mb\eta}{\infty}&=\max_{l\in\brac{1,\cdots,m}}\innerprod{\mb X_0\mb e_l}{\mb A_0^T \paren{ \mb Y \mb Y^T }^{-1/2} \mb q}\\
&\le\max_l\norm{\mb X_0\mb e_l}2\norm{\mb A_0^T \paren{ \mb Y \mb Y^T }^{-1/2} \mb q}2\\
&\le\sqrt{2k}\norm{\mb x_0}{\infty}\cdot \paren{\theta m}^{-1/2}\paren{1+\frac{4\kappa^3\delta}{\sigma_{\min}}}.
\end{align}
At the same time, plugging in $\norm{\mb\eta}2=1$, we have
\begin{equation}
\norm{\mb\eta}6^6\le \norm{\mb\eta}2^2\norm{\mb\eta}{\infty}^4\le\paren{1+\frac{4\kappa^3\delta}{\sigma_{\min}}}^4\frac{4k^2}{\theta^2m^2}\norm{\mb x_0}{\infty}^4.
\end{equation}

\noindent{\bf Norms of $\mb{\bar\eta}$.} Here, $\mb{\bar\eta}=\mb Y^T\paren{\theta m\mb A_0\mb A_0^T}^{-1/2}\mb q=\mb X_0^T\mb A_0^T\paren{\theta m\mb A_0\mb A_0^T}^{-1/2}\mb q$ with
\begin{align}
\norm{\mb A_0^T\paren{\theta m\mb A_0\mb A_0^T}^{-1/2}\mb q}{\infty}
&\le \norm{\mb A_0^T\paren{\theta m\mb A_0\mb A_0^T}^{-1/2}\mb q}2 \\
&= \paren{\theta m}^{-1/2},
\end{align}
therefore
\begin{align}
\norm{\mb{\bar\eta}}{\infty}&\le\max_l\norm{\mb X_0\mb e_l}2\norm{\mb A_0\paren{\theta m\mb A_0\mb A_0^T}^{-1/2}\mb q}2\nonumber\\
&\le\paren{\frac{2k}{\theta m}}^{1/2}\norm{\mb x_0}{\infty},\\
\norm{\mb{\bar\eta}}2&\le\norm{\mb X_0^T}2\norm{\mb A_0\paren{\theta m\mb A_0\mb A_0^T}^{-1/2}\mb q}2\nonumber\\
&\le1+\delta/2.
\end{align}

\noindent{\bf Norms of $\mb\eta-\mb{\bar\eta}$.} With similar reasoning, we can obtain
\begin{align}
\norm{\mb\eta-\mb{\bar\eta}}{\infty}
&=\norm{\mb Y^T\paren{\mb Y\mb Y^T}^{-1/2}\mb q-\mb Y^T\paren{\theta m\mb A_0\mb A_0^T}^{-1/2}\mb q}{\infty}\nonumber\\
&\le\max_{l\in\brac{1,\cdots,m}}\norm{\mb X_0\mb e_l}2\paren{\theta m}^{-1/2}\norm{\mb A_0^T\paren{\frac1{\theta m}\mb Y\mb Y^T}^{-1/2}\!\!\!\!-\mb A_0^T\paren{\mb A_0\mb A_0^T}^{-1/2}}2\\
&\le\frac{4\kappa^3\delta}{\sigma_{\min}}\paren{\frac{2k}{\theta m}}^{1/2}\norm{\mb x_0}{\infty},
\end{align}
and
\begin{align}
\norm{\mb\eta-\mb{\bar\eta}}2
&\le\norm{\mb X_0}2\paren{\theta m}^{-1/2}\norm{\mb q}2\norm{\mb A_0^T\paren{\frac1{\theta m}\mb Y\mb Y^T}^{-1/2}\!\!\!\!-\mb A_0^T\paren{\mb A_0\mb A_0^T}^{-1/2}}2\\
&\le \paren{\theta m}^{-1/2}\frac{4\kappa^3\delta}{\sigma_{\min}}\norm{\mb X_0}2\\
&\le\paren{1+\delta/2}\frac{4\kappa^3\delta}{\sigma_{\min}},
\end{align}
completing the proof. 
\end{proof}

\section{Proof of the Main Theorem and Corollary}
\label{sec:main}
\subsection{Proof of the Main Theorem}
\begin{lemma}
If following inequalities hold
\begin{align}
\norm{\grad[\psi]\paren{\mb q}-\frac{3\paren{1-\theta}}{\theta m^2}\grad[\varphi]\paren{\mb q}}2
&\;\le\frac{3c_{\star}}{2\kappa^2}\frac{1-\theta}{\theta m^2}\norm{\mb A^T\mb q}4^6,\\
\norm{\Hess[\psi]\paren{\mb q}-\frac{3\paren{1-\theta}}{\theta m^2}\Hess[\varphi]\paren{\mb q}}2
&\;\le 3\paren{1-6c_{\star}-36c_{\star}^2-24c_{\star}^3}\frac{1-\theta}{\theta m^2}\norm{\mb A^T\mb q}4^4.
\end{align}
for all $\mb q\in\mc R_{2C_{\star}}$ with $C_{\star}\ge 10$ and $c_{\star} = 1/C_{\star}$, then any local minimum $\mb{\bar q}$ of $\psi\paren{\mb q}$ in $\mc R_{2C_{\star}}$ satisfies $\abs{\innerprod{\mb{\bar q}}{\mc P_{\bb S}\brac{\mb a_l}}}\ge 1-2c_{\star}\kappa^{-2}$ for some index $l$.
\end{lemma}
\begin{proof}
Let 
\begin{equation}
\deltagrad = \grad[\psi]\paren{\mb q}-\frac{3\paren{1-\theta}}{\theta m^2}\grad[\varphi]\paren{\mb q},
\end{equation}
and let
\begin{equation}
\dbgrad = \frac{\theta m^2}{3\paren{1-\theta}} \deltagrad .
\end{equation}
Then at any stationary point of $\psi\paren{\mb q}$, we have
\begin{align}
\mb 0&=\mb A^T\grad[\psi]\paren{\mb q}\\
&=\frac{3\paren{1-\theta}}{\theta m^2}\mb A^T\grad[\varphi]\paren{\mb q}+\mb A^T\deltagrad.
\end{align}
Hence for any index $i$, following equality always holds
\begin{align}
0&=\norm{\mb a_i}2^2\zeta^3_i+\sum_{j\neq i}\innerprod{\mb a_i}{\mb a_j}\zeta^3_j-\zeta_i\norm{\mb\zeta}4^4+\innerprod{\mb a_i}{ \dbgrad }\nonumber\\
&=\zeta^3_i-\zeta_i\underbrace{ \frac{ \norm{\mb\zeta}4^4 }{ \norm{\mb a_i}2^2 } }_{\alpha_i}+\underbrace{ \frac{ \sum_{j\neq i}\innerprod{\mb a_i}{\mb a_j}\zeta^3_j+\innerprod{\mb a_i}{\dbgrad } }{ \norm{\mb a_i}2^2 } }_{\beta'_i}
\end{align}
with $\mb\zeta=\mb A^T\mb q$. Under the assumption that 
\begin{equation}
\norm{\grad[\psi]\paren{\mb q}-\frac{3\paren{1-\theta}}{\theta m^2}\grad[\varphi]\paren{\mb q}}2\le \frac{3c_{\star}}{2\kappa^2}\frac{1-\theta}{\theta m^2}\norm{\mb\zeta}4^6, 
\end{equation}
the perturbed part can be bounded via 
\begin{equation}
\abs{\innerprod{\mb a_i}{\dbgrad }}\le\norm{\mb a_i}2\norm{\dbgrad}2\le \frac{c_{\star}}{2\kappa^2}\norm{\mb a_i}2\norm{\mb\zeta}4^6,
\end{equation}
and also 
\begin{equation}
\frac{\beta_i'}{\alpha_i^{3/2}}\le\frac{\mu\norm{\mb\zeta}3^3+\tfrac12c_\star\kappa^{-2}\norm{\mb\zeta}4^6}{\norm{\mb\zeta}4^6}\le c_\star\kappa^{-2}\le \frac14.
\end{equation}
Then by  \Cref{lem:cubic}, at every stationary point $\bar{\mb q}$, the $i$-th entry of $\mb \zeta$  resides in the set $\bigcup_{x\in \{0,\pm\sqrt{\alpha_i}\}}[x-\frac{2\beta_i'}{\alpha_i}, x+ \frac{2\beta_i'}{\alpha_i}]$ -- i.e., $\mb \zeta$ is nearly a trinary vector. 

Moreover, we can characterize the curvature of critical points in terms of the number of large entries of $\mb \zeta$. Indeed, whenever $\mb \zeta$ has at least two entries in 
$$\bigcup_{x\in \{\pm\sqrt{\alpha_i}\}} \left[x-\frac{2\beta_i'}{\alpha_i}, x+ \frac{2\beta_i'}{\alpha_i} \right],$$
using  \eqref{eqn:Hess_asympt_saddle}, there exists a direction of strict negative curvature, provided
\begin{align}
\Hess[\psi]\paren{\mb q}\prec& \; \frac{3\paren{1-\theta}}{\theta m^2}\Hess[\varphi]\paren{\mb q} \nonumber\\
&+\; 3\paren{2-11c_{\star}}\frac{1-\theta}{\theta m^2}\norm{\mb\zeta}4^4\mb I.
\end{align} 
Similarly, whenever $\mb \zeta$ has only one entry in 
$$\bigcup_{x\in \{\pm\sqrt{\alpha_i}\}}\left[x-\frac{2\beta_i'}{\alpha_i}, x+ \frac{2\beta_i'}{\alpha_i}\right],$$
using  \eqref{eqn:Hess_asympt_localmin}, we have that $\Hess[\psi](\mb q) \succ \mb 0$, provided
\begin{align}
\Hess[\psi]\paren{\mb q}&\succ \frac{3\paren{1-\theta}}{\theta m^2}\Hess[\varphi]\paren{\mb q}\nonumber\\
- & 3\paren{1-6c_{\star}-36c_{\star}^2-24c_{\star}^3}\frac{1-\theta}{\theta m^2}\norm{\mb\zeta}4^4\mb I.
\end{align} 
When $C_{\star}\ge 10$ and $c_\star\le 0.1$, we have $2-11c_{\star}>1-6c_{\star}-36c_{\star}^2-24c_{\star}^3\ge 0.016$, and so above characterization obtains. 
\end{proof}

\begin{theorem}[Main Result]
\label{thm:main}
Assume the observation $\mb y\in\R^m$ is the cyclic convolution of $\mb a_0\in\R^k$ and $\mb x_0\simiid \mathrm{BG}\paren{\theta}\in\R^m$, where the convolution matrix $\mb A_0\in\R^{k\times\paren{2k-1}}$ has minimum singular value $\sigma_{\min}>0$ and condition number $\kappa\ge 1$, and $\mb A$ has column incoherence $\mu$. If 
\begin{equation}
m\;\ge\;C\frac{\min\set{\paren{2C_{\star}\mu}^{-1},\kappa^2k^2}}{\paren{1-\theta}^2\sigma^2_{\min}}\kappa^8 k^4\log^3\paren{\kappa k}
\end{equation}
and $\theta \ge \log{k}/k$, then with probability no smaller than $1-\exp\paren{-k}-\theta^2\paren{1-\theta}^2k^{-4}-2\exp\paren{-\theta k}-48k^{-7}-48m^{-5} - 24k\exp\paren{-\tfrac1{144}\min\set{k,3\sqrt{\theta m}}}$, any local minimum $\mb{\bar q}$ of $\psi$ in $\hat{\mc R}_{2C_\star}$ satisfies $\abs{\innerprod{\mb{\bar q}}{\mc P_{\bb S}\brac{\mb a_{\tau}}}}\ge 1-c_{\star}\kappa^{-2}$ for some integer $\tau$.
\end{theorem}
\begin{proof}
From the concentration analysis for the Riemannian gradient (\Cref{lem:grad_scale}) and Hessian (\Cref{lem:hess_scale}), if
\begin{equation}
m\;\ge\;C\frac{\min\set{\paren{2C_{\star}\mu}^{-1},\kappa^2k^2}}{\paren{1-\theta}^2\sigma^2_{\min}}\kappa^8k^4\log^3\paren{\kappa k},
\end{equation}
then with probability no smaller than $1-\exp\paren{-k}-\theta^2\paren{1-\theta}^2k^{-4}-2\exp\paren{-\theta k} - 24k\exp\paren{-\tfrac1{144}\min\set{k,3\sqrt{\theta m}}}-48k^{-7}-48m^{-5}$,
\begin{align}
\norm{\grad[\psi]\paren{\mb q}-\frac{3\paren{1-\theta}}{\theta m^2}\grad[\varphi]\paren{\mb q}}2
&\;\le\frac{3c_{\star}}{2\kappa^2}\frac{1-\theta}{\theta m^2}\norm{\mb A^T\mb q}4^6,\\
\norm{\Hess[\psi]\paren{\mb q}-\frac{3\paren{1-\theta}}{\theta m^2}\Hess[\varphi]\paren{\mb q}}2 
&\;\le 3\paren{1-6c_{\star}-36c_{\star}^2-24c_{\star}^3}\frac{1-\theta}{\theta m^2}\norm{\mb A^T\mb q}4^4.
\end{align}
hold for all $\mb q\in\hat{\mc R}_{2C_\star}$ with $C_{\star}\ge 10$ and $c_{\star} = 1/C_{\star}$. Therefore, by \Cref{lem:geo_scale} any local minimum $\mb{\bar q}$ of $\psi\paren{\mb q}$ in $\mc R_{2C_{\star}}$ satisfies $\abs{\innerprod{\mb{\bar q}}{\mc P_{\bb S}\brac{\mb a_l}}}\ge 1-2c_{\star}\kappa^{-2}$ for some index $l$.
\end{proof}

\subsection{Proof of the Main Corollary}

\begin{corollary}
\label{thm:crlr}
Suppose the ground truth kernel $\mb a_0$ has induces coherence $0\le\mu\le\tfrac1{8\times48}\log^{-3/2}\paren{k}$ and sparse coefficient $\mb x_0\simiid \mathrm{BG}\paren{\theta}\in\R^m$. there exist positive constants $C\ge 2560^4$ and $C'$ such that whenever the sparsity level \begin{align}
&64k^{-1}\log{k}\le\theta\le\min\big\{\tfrac1{48^2}\mu^{-2}k^{-1}\log^{-2}{k},\\ 
&\quad \paren{\tfrac1{4}-\tfrac{640}{C^{1/4}}}\paren{3C_{\star}\mu\kappa^2}^{-2/3}k^{-1}\paren{1+36\mu^2k\log k}^{-2}\big\},\nonumber
\end{align}
and signal length 
\begin{align}
m\;\ge\;&\max\big\{C\theta^2\sigma^{-2}_{\min}\kappa^6k^3\paren{1+36\mu^2k\log{k}}^4\log\paren{\kappa k},\\
C'&\paren{1-\theta}^{-2}\sigma^{-2}_{\min}\min\set{\mu^{-1},\kappa^2k^2}\kappa^8 k^4\log^3\paren{\kappa k}\big\},\nonumber
\end{align}
then Algorithm 1 recovers $\mb{\bar a}$ such that
\begin{equation}
\norm{\mb{\bar a}\pm\mc P_{\bb S}\brac{\injector_k\shift{\extend{\mb a_0}}{\tau}}}2\le 4\sqrt{c_\star}+ck^{-1}
\end{equation}
for some integer shift $\tau\in\brac{-\paren{k-1},{k-1}}$ with probability no smaller than $1-k^{-1}-8k^{-2}-\exp\paren{-k}-\theta^2\paren{1-\theta}^2k^{-4}-2\exp\paren{-\theta k} - 24k\exp\paren{-\tfrac1{144}\min\set{k,3\sqrt{\theta m}}}-48k^{-7}-48m^{-5}$. 
\end{corollary}
\begin{proof}
From the concentration results for the Riemannian gradient, at every point $\mb q \in \hat{\mc R}_{2 C_\star}$, the objective value of $\psi\paren{\mb q}$ satisfies 
\begin{align}
& \abs{\psi\paren{\mb q}-\frac{3\paren{1-\theta}}{\theta m^2}\varphi\paren{\mb q}+\frac3{4m^2}} \nonumber\\
&\le\abs{\frac{\norm{\mb Y^T\!\paren{\mb Y\mb Y^T}^{-1/2}\!\!\!\mb q}4^4}{4m}-\frac{3\paren{1-\theta}\norm{\mb\zeta}4^4}{4\theta m^2}-\frac3{4m^2}}\\
&\le\abs{\innerprod{\mb q}{\frac{\paren{\mb Y\mb Y^T}^{\!-\!1\!/\!2}\!\mb Y\mb\eta^{\circ3}\!}{4m}-\frac{3\paren{1-\theta}}{4\theta m^2}\mb A\mb\zeta^{\circ3}\!\!-\frac3{4m^2}\mb q}}\\
&\le\norm{\frac{\paren{\mb Y\mb Y^T}^{-1/2}\!\mb Y\mb\eta^{\circ3}\!\!}{4m}-\frac{3\paren{1-\theta}}{4\theta m^2}\mb A\mb\zeta^{\circ3}-\frac3{4m^2}\mb q}2\\
&\le\frac1{4m}\!\norm{\paren{\mb Y\mb Y^T}^{\!-\!1/2}\!\mb Y\mb\eta^{\circ3}\!\!-\!\paren{\theta m\mb A_0\mb A_0^T}^{\!-\!1/2}\!\mb Y\mb\eta^{\circ3}\!}2\nonumber\\
&\quad+\frac1{4\theta^{1/2}m^{3/2}}\norm{\paren{\mb A_0\mb A_0^T}^{-1/2}\!\mb Y\!\paren{\mb\eta^{\circ3}-\mb{\bar\eta}^{\circ3}}}2\nonumber\\
&\quad+\norm{\frac1{4\theta^{1/2}m^{3/2}}\paren{\mb A_0\mb A_0^T}^{-1/2}\!\mb Y\mb{\bar\eta}^{\circ3}-\frac{3\paren{1-\theta}}{4\theta m^2}\mb A\mb\zeta^{\circ3}\!-\frac3{4m^2}\mb q}2\\
&\le\frac{3c_{\star}}{8\kappa^2}\frac{1-\theta}{\theta m^2}\min_{\mb q\in\hat{\mc R}_{2C_\star}}\norm{\mb A^T\mb q}4^6
\end{align}
with probability no smaller than $1-2\exp\paren{-\theta k} - 24k\exp\paren{-\tfrac1{144}\min\set{k,3\sqrt{\theta m}}} - 48k^{-7} - 48m^{-5}$.
The last inequality is derived with similar arguments in  \Cref{lem:grad_scale}, for simplicity, we do not present them here. Moreover, with \Cref{lem:init}, we can obtain an initialization point $\qinit$ such that 
\begin{align}
\norm{\mb A^T\qinit}4^4&\ge \paren{3C_{\star}\mu\kappa^2}^{2/3}\\
&\ge \paren{2C_{\star}\mu\kappa^2}^{2/3}+\mu/2. 
\end{align} 
Consider any descent method for $\psi$, which generates a sequence of iterates $\mb q^{(0)} = \qinit, \mb q^{(1)}, \dots, \mb q^{(k)}, \dots$ such that $\psi(\mb q^{(k)})$ is non-increasing with $k$. Then
\begin{align}
\psi\paren{\mb q^{(k)}}&\le \psi\paren{\qinit}\\
&\le \frac{3\paren{1-\theta}}{\theta m^2}\varphi\paren{\qinit}+\frac3{4m^2} + \frac{3c_{\star}}{8\kappa^2}\frac{1-\theta}{\theta m^2}\min_{\mb q\in\hat{\mc R}_{2C_\star}}\norm{\mb A^T\mb q}4^6.
\end{align}
On the other hand, the finite sample objective function value $\psi$ is close to that of $\frac{3\paren{1-\theta}}{\theta m^2}\varphi\paren{\mb q} - \frac3{4m^2}$, 

\begin{align}
\frac{3\paren{1-\theta}}{\theta m^2}\varphi\paren{\mb q^{(k)}}
&\le\psi\paren{\mb q^{(k)}} + \frac3{4m^2} + \frac{3c_{\star}}{8\kappa^2}\frac{1-\theta}{\theta m^2}\min_{\mb q\in\hat{\mc R}_{2C_\star}}\norm{\mb A^T\mb q}4^6\\
&\le\frac{3\paren{1-\theta}}{\theta m^2}\varphi\paren{\qinit}+\frac{3c_{\star}}{4\kappa^2}\frac{1-\theta}{\theta m^2}\min_{\mb q\in\hat{\mc R}_{2C_\star}}\norm{\mb A^T\mb q}4^6,
\end{align}
Therefore, we obtain that
\begin{align}
\varphi\paren{\mb q^{(k)}}&\le \varphi\paren{\qinit}+\frac{\mu}2\\
&\le \varphi\paren{\qinit}+\frac{c_{\star}}{4\kappa^2}\min_{\mb q\in\hat{\mc R}_{2C_\star}}\norm{\mb A^T\mb q}4^6,
\end{align}
which implies that $\mb q^{(k)}\in\hat{\mc R}_{2C_\star}$ always holds. At last,  \Cref{thm:main} says that any local minimum $\bar{\mb q}$ is close to $\pm \mb a_i$ for some $i$, in the sense that
\begin{equation}
\abs{\innerprod{\bar{\mb q}}{\mc P_{\bb S}\brac{\mb a_i}}}\ge 1-c_\star\kappa^{-2}.
\end{equation}

Write $\frac1{\theta m}\mb Y\mb Y^T=\mb A_0\paren{\mb I+\mb\Delta}\mb A_0^T$ with  $\norm{\mb\Delta}2\le \delta$, and let 
\begin{equation}
\bar{\mb q}=\pm\frac{\mb a_i}{\norm{\mb a_i}2}+\sqrt{2\paren{1-\abs{\innerprod{\bar{\mb q}}{\frac{\mb a_i}{\norm{\mb a_i}2}}}}}\mb\delta,
\end{equation}
with $\norm{\mb\delta}2=1$. Since 
\begin{equation}
\mb a_i = \paren{\mb A_0\mb A_0^T}^{\!-1/2}\injector_k^*\shift{\extend{\mb a_0}}{-\paren{k-i}},
\end{equation}
we have
\begin{align}
\paren{\frac{\mb Y\mb Y^T}{\theta m}}^{\!1/2}\bar{\mb q}
&=\pm \paren{\frac{\mb Y\mb Y^T}{\theta m}\!}^{\!1\!/2}\!\brac{\frac{\mb a_i}{\norm{\mb a_i}2\!}\!+\!\sqrt{2\!\paren{\!1\!-\!\abs{\innerprod{\bar{\mb q}}{\frac{\mb a_i}{\norm{\mb a_i}2\!}}}}}\mb\delta}\\
&=\pm\paren{\!\frac{\mb Y\mb Y^T}{\theta m}\!}^{\!1\!/2}\!\!\paren{\mb A_0\mb A_0^T}^{\!-\!1/2}\frac{\injector_k^*\shift{\extend{\mb a_0}}{-\paren{k-i}}}{\norm{\mb a_i}2}\nonumber\\
&\quad+\sqrt{2\paren{1-\abs{\innerprod{\!\bar{\mb q}}{\frac{\mb a_i}{\norm{\mb a_i}2}}\!}}}\paren{\frac{\mb Y\mb Y^T}{\theta m}}^{1/2}\mb\delta
\end{align}
therefore the error can be bounded as
\begin{eqnarray}
\lefteqn{\norm{\paren{\frac{\mb Y\mb Y^T}{\theta m}}^{1/2}\bar{\mb q}\pm\frac{\injector_k^*\shift{\extend{\mb a_0}}{-\paren{k-i}}}{\norm{\mb a_i}2}}2}\nonumber\\
&\le&\norm{\paren{\frac{\mb Y\mb Y^T}{\theta m}}^{1/2}\!\!\!\!\paren{\mb A_0\mb A_0^T}^{-1/2}\!\!\!\!-\mb I}2\norm{\frac{\injector_k^*\shift{\extend{\mb a_0}}{-\paren{k-i}}}{\norm{\mb a_i}2}}2\nonumber\\
&&+\sqrt{2\paren{1-\abs{\innerprod{\bar{\mb q}}{\frac{\mb a_i}{\norm{\mb a_i}2}}}}}\norm{\paren{\frac{\mb Y\mb Y^T}{\theta m}}^{1/2}}2.
\end{eqnarray}

Finally, using the fact that for any nonzero vectors $\mb u$ and $\mb v$ that $\innerprod{\mb u}{\mb v}\ge 0$,
\begin{equation}
\norm{\frac{\mb u}{\norm{\mb u}{2}} - \frac{\mb v}{\norm{\mb v}{2}} }{2} \le \frac{\sqrt{2}}{\norm{\mb v}{2}} \norm{\mb u - \mb v }{2}
\end{equation}
always holds. Therefore,
\begin{align}
\lefteqn{\norm{\mb{\bar a}\pm\mc P_{\bb S}\brac{\injector_k\shift{\extend{\mb a_0}}{i}}}2}\nonumber\\
&={ \norm{\mc P_{\bb S}\brac{\paren{\mb Y\mb Y^T}^{\!1/2}\bar{\mb q}}\pm\mc P_{\bb S}\brac{\injector_k\shift{\extend{\mb a_0}}{i}}}2 }\\
&\le{\frac{\sqrt{2}\norm{\mb a_i}2}{\norm{\injector_k^*\shift{\extend{\mb a_0}}{i-k}}2}\norm{\paren{\frac{\mb Y\mb Y^T}{\theta m}}^{\!1/2}\!\!\bar{\mb q} \pm \frac{\injector_k^*\shift{\extend{\mb a_0}}{i-k}}{\norm{\mb a_i}2}}2 } \\
&\le \kappa \sqrt{2\paren{1+\delta}\paren{1-\abs{\innerprod{\bar{\mb q}}{\frac{\mb a_i}{\norm{\mb a_i}2}}}}}  +\sqrt{2}\kappa\norm{\paren{\frac{\mb Y\mb Y^T}{\theta m}}^{\!1/2}\paren{\mb A_0\mb A_0^T}^{\!-1/2}-\mb I}2 \\
&\le { 2\kappa\sqrt{2\paren{1-\abs{\innerprod{\bar{\mb q}}{\frac{\mb a_i}{\norm{\mb a_i}2}}}}}+\sqrt{2}\kappa^3\delta/\sigma_{\min}} \\
&\qquad\paren{\text{\Cref{lem:precond_neghalf_1}}}\nonumber\\
&\le 4\sqrt{c_{\star}}+ 10\sqrt{2}\kappa^3{\sigma^{-1}_{\min}}\sqrt{k\log{m}/m}\\
&\le 4\sqrt{c_{\star}}+ck^{-1},
\end{align}
completing the proof.
\end{proof}

\section{Initialization}
\label{sec:init}
\begin{lemma}
\label{lem:init}
Suppose $\mb x_0\simiid \mathrm{BG}\paren{\theta}\in\R^m$. There exists a positive constant $C>2560^4$ such that whenever 
\begin{equation}
 m \ge C\theta^2\sigma^{-2}_{\min}\kappa^6k^3\paren{1+36\mu^2k\log{k}}^4\log\paren{\kappa k/\sigma_{\min}}
\end{equation}
and the sparsity rate 
\begin{align}
&64k^{-1}\log{k}\le\theta\le\min\big\{\tfrac1{48^2}\mu^{-2}k^{-1}\log^{-2}{k},\\ 
&\quad \paren{\tfrac1{4}-\tfrac{640}{C^{1/4}}}\paren{3C_{\star}\mu\kappa^2}^{-2/3}k^{-1}\paren{1+36\mu^2k\log^{}k}^{-2}\big\},\nonumber
\end{align}
Then the initialization $\qinit=\mc P_{\bb S}\brac{\paren{\mb Y\mb Y^T}^{-1/2}\mb y_{i}}$ satisfies 
\begin{equation}
\norm{\mb A^T\qinit}4^6\ge 3C_{\star}\mu\kappa^2, 
\end{equation}
namely $\qinit\in\hat{\mc R}_{3C_\star}$, with probability no smaller than $1-k^{-1}-8k^{-2}-2\exp\paren{-\theta k}-48k^{-7} -48m^{-5}-24k\exp\paren{-\tfrac1{144}\min\set{k,3\sqrt{\theta m}}}$. 
\end{lemma}

\begin{proof}
Since 
\begin{align}
m\ge C\frac{\theta^2}{\sigma^2_{\min}}\kappa^6k^3\paren{1+36\mu^2k\log{k}}^4\log\paren{\kappa k/\sigma_{\min}}
\end{align}
with $C\ge 2560^4$, then from Lemma \ref{lem:preconditioning}, then with probability no smaller than $1-2\exp\paren{-\theta k}-48k^{-7} -48m^{-5}-24k\exp\paren{-\tfrac1{144}\min\set{k,3\sqrt{\theta m}}}$, we have
\begin{align}
\delta&\doteq\norm{\frac1{\theta m}\mb X_0\mb X_0^T-\mb I}2\\
&\le10\sqrt{k\log m/m}\\
&\le\frac{10\sigma_{\min}}{\theta\sigma^{-1}_{\min}\kappa^3k\paren{1+36\mu^2k\log{k}}^2}\sqrt{\frac{\log\paren{\frac{C\kappa^6k^3\paren{1+36\mu^2k\log{k}}^4}{\sigma^2_{\min}}\log\paren{\frac{\kappa k}{\sigma_{\min}}}}}{C\log \paren{\kappa k/\sigma_{\min}}}}\\
&\le\frac{20\sigma_{\min}}{C^{1/4}\theta\kappa^3 k\paren{1+36\mu^2k\log{k}}^2}
\end{align}
obtains, and the last inequality holds when $C\ge 1000$ that
\begin{equation}
\log\paren{37^4C}\le \log2\sqrt{C}.
\end{equation}
Therefore
\begin{align}
&C\sigma^{-2}_{\min}\kappa^6k^3\paren{1+36\mu^2k\log{k}}^4\log\paren{\kappa k/\sigma_{\min}}\nonumber\\
&\le 37^4C\paren{\kappa k/\sigma_{\min}}^7\log^5\paren{\kappa k/\sigma_{\min}}\\
&\le 37^4C\paren{\kappa k/\sigma_{\min}}^{12},
\end{align}
or
\begin{align}
&\sqrt{\frac{\log\paren{\frac{C\kappa^6k^3\paren{1+36\mu^2k\log{k}}^4}{\sigma^2_{\min}}\log\paren{\frac{\kappa k}{\sigma_{\min}}}}}{C\log \paren{\kappa k/\sigma_{\min}}}}\nonumber\\
&\le \sqrt{\frac{\log\paren{37^4C}+12\log\paren{\kappa k/\sigma_{\min}}}{C\log \paren{\kappa k/\sigma_{\min}}}}\\
&\le \sqrt{\frac{\log{2}}{\sqrt{C}\log\paren{\kappa k/\sigma_{\min}}}+\frac{12}C}\\
&\le \frac2{C^{1/4}}\qquad\paren{k\ge 2,C\ge 16}.
\end{align}

Moreover, $\kappa^2 \delta \le 1/2$ always holds provided 
\begin{equation}
C\ge\paren{\frac{40}{\theta k\paren{1+36\mu^2k\log{k}}^2}}^4.
\end{equation}
Notice that because $\theta$ is lower bounded by $c\log k / k$, the right hand side is indeed bounded by an absolute constant.

Set $\zetainit=\mb A^T\qinit$ and $\zetahatinit=\mc P_{\bb S}\brac{\mb A^T\mb A{ \mb x_i}}$. Then using for any nonzero vectors $\mb u$ and $\mb v$,
\begin{equation}
\norm{ \frac{\mb u}{\norm{\mb u}{2}} - \frac{\mb v }{\norm{\mb v}{2}} }{2} \le \frac{2}{\norm{\mb v}{2}} \norm{\mb u - \mb v }{2},
\end{equation}
we have that 
\begin{align}
&\norm{ \zetainit - \zetahatinit }2\nonumber \\
&= \norm{\mb A^T\mc P_{\bb S}\brac{\paren{\mb Y\mb Y^T}^{-1/2}\!\!\!\!\mb A_0\mb x_i}-\mc P_{\bb S}\brac{\mb A^T\mb A\mb x_i}}2 \\
&=\norm{\frac{\mb A^T\paren{\frac1{\theta m}\mb Y\mb Y^T}^{-1/2}\mb A_0\mb x_i}{\norm{\paren{\frac1{\theta m}\mb Y\mb Y^T}^{-1/2}\mb A_0\mb x_i}2}-\frac{\mb A^T\mb A\mb x_i}{\norm{\mb A^T\mb A\mb x_i}2}}2\\
&\le{ \frac{2}{\norm{\mb A\mb x_i}2}\norm{\paren{\frac1{\theta m}\mb Y\mb Y^T}^{-1/2}\!\!\mb A_0\mb x_i-\mb A\mb x_i}2 }\\
&\le { 2 \norm{\mb A_0}2 \norm{\paren{\frac1{\theta m}\mb Y\mb Y^T}^{-1/2}\!\!-\paren{\mb A_0\mb A_0^T}^{-1/2}}2 }\\
&\le { \frac{8\kappa^3\delta}{\sigma_{\min}}},
\end{align}
{where we have used  \Cref{lem:precond_neghalf_2} in the final bound. }

Since $\norm{\cdot}4^4$ is convex, $\norm{\zetainit}4^4$ can be lower bounded via 
\begin{align}
\norm{\zetainit}4^4&\ge {\norm{\zetahatinit}4^4+4\innerprod{\zetahatinit^{\circ3}}{\zetainit-\zetahatinit}}\\
&\ge \norm{\zetahatinit}4^4-4\norm{\zetainit-\zetahatinit}2\\
&\ge \norm{\zetahatinit}4^4-{\frac{32\kappa^3\delta}{\sigma_{\min}}}.
\end{align}
Let $I=\supp\paren{\mb x_i}$, then the vector $\zetahatinit=\mc P_{\bb S}\brac{\mb A^T\mb A\mb x_i}$ is composed of $\abs{I}$ large components and small components on the off-support $I^c$ of $\mb x_i$.

\paragraph{Dense Component of $\zetahatinit$.}
Note that $\norm{\paren{\mb A^T\mb A}_{I^c,I}\mb x_i}2\le\norm{\offdiag\paren{\mb A^T\mb A}\mb x_i}2$ 
with $\norm{\offdiag\paren{\mb A^T\mb A}}{\infty}\le\mu$. We have 
\begin{align}
\expect{\offdiag\paren{\mb A^T\mb A}\mb x_i}&=\mb 0\\
\expect{\abs{\mb e_j^T\offdiag\paren{\mb A^T\mb A}\mb x_i}^2}&=\theta\norm{\mb e_j^T\offdiag\paren{\mb A^T\mb A}}2^2\nonumber\\
&\le\mu^2\theta k
\end{align}
With Bernstein's Inequality, the summation of moment-bounded independent random variables can be controlled via
\begin{equation}
\prob{\abs{\mb e_j^T\offdiag\paren{\mb A^T\mb A}\mb x_i}\ge \mu t}\le2\exp\paren{-\frac{t^2}{2\theta k+2t}}
\end{equation}
and via union bound 
\begin{equation}
\prob{\norm{\offdiag\paren{\mb A^T\mb A}\mb x_i}2^2\ge 2k \paren{\mu t}^2}\le 4k\exp\paren{-\frac{t^2}{2\theta k+2t}}.
\end{equation}
Therefore, setting $t^2 = 9\theta k\log k$, we obtain
\begin{equation}
\norm{\offdiag\paren{\mb A^T\mb A}\mb x_i}2^2\le 18 \mu^2\theta k^2\log{k},
\end{equation}
with failure probability bounded by
\begin{align}
\lefteqn{4k \exp\paren{ - \frac{9\theta k \log k }{ 2 \theta k + 2 \sqrt{9\theta k \log k} } } }\nonumber\\
&= 4 k \exp\paren{- \frac{9\log{k} }{ 2 + 6\sqrt{\paren{\theta k}^{-1} \log k } }}\\
&\le 4k^{-2}.
\end{align}
The last inequality is derived under the assumption $\paren{\theta k}^{-1}\log{k}\le { \tfrac{1}{64}}$.

\paragraph{Spiky Component of $\zetahatinit$.} On the other hand, 
\begin{align}
\expect{\norm{\diag\paren{\mb A^T\mb A}\mb x_i}{2}^2}&=\theta\norm{\diag\paren{\mb A^T\mb A}}F^2\\
&=\theta k.
\end{align}
For $\diag\paren{\mb A^T\mb A}\mb x_i$, applying the moment control Bernstein Inequality, we have
\begin{equation}
\prob{\abs{\norm{\diag\paren{\mb A^T\mb A}\mb x_i}{2}^2-\expect{\cdot}}\ge t}\le2\exp\paren{-\frac{t^2}{2\theta k+ 2t}}.
\end{equation}
By setting $t = 2\sqrt{\theta k\log{k}}$, we obtain that with probability no smaller than $1-k^{-1}$, 
\begin{equation}
\norm{\diag\paren{\mb A^T\mb A}\mb x_i}{2}^2\ge \theta k - 2\sqrt{\theta k\log{k}}.
\end{equation}

Denote the following events for the entry-wise magnitude 
\begin{equation}
\event_j = \set{ | \mb e_j^T \mathrm{offdiag}( \mb A^T \mb A ) \mb x_i | \le \mu t },
\end{equation}
and for the support size
\begin{equation}
\event_{\mathrm{supp}} = \set{ \norm{\mb x_i}{0} \le 4 \theta k }.
\end{equation}
On their intersection $\event_{\mathrm{supp}} \cap \bigcap_{j = 1}^{2k} \event_j$, we have
\begin{equation}
\norm{\mathrm{offdiag}(\mb A^T \mb A)_{I,I} \mb x_i }{2}^2 \le 4 \theta k ( \mu t )^2.
\end{equation}
The the failure probability can be bounded from the union bound as
\begin{align}
\lefteqn{\bb P\brac{ \norm{\mathrm{offdiag}(\mb A^T \mb A)_{I,I} \mb x_i }{2}^2 \ge 4 \theta k (\mu t)^2 } }\nonumber\\
&\le \bb P\left[ \, \left( \event_{\mathrm{supp}} \cap \bigcap_j \event_j \right)^c \; \right] \\
&= \bb P\left[ \, \event_{\mathrm{supp}}^c \cup \bigcup_j \event_j^c \; \right] \\
&\le \bb P\left[ \event_{\mathrm{supp}}^c \right] + \sum_j \bb P\left[ \event_j^c \right] \\
&\le \exp(-\theta k) + 4 k \exp\left( - \frac{t^2}{2 \theta k + 2 t} \right).
\end{align}

Therefore, by setting $t^2 = 9\theta k\log{k}$, we obtain
\begin{equation}
\norm{\offdiag\paren{\mb A^T\mb A}_{I,I}\mb x_i}2^2\le 36\mu^2\theta^2 k^2\log{k}
\end{equation}
with probability no smaller than $1-\exp\paren{-\theta k}-8 k^{-2}$. Therefore, with probability no smaller than $1-k^{-1}-8k^{-2}-\exp\paren{-\theta k}$, 
\begin{align}
\norm{\diag\paren{\mb A^T\mb A}\mb x_i}{2}^2 &\ge \theta k - 2\sqrt{\theta k\log{k}}\\
\norm{\offdiag\paren{\mb A^T\mb A}_{I,I}\mb x_i}2^2&\le 36\mu^2\theta^2 k^2\log{k}
\end{align}
and via Cauchy-Schwatz inequality, we obtain
\begin{align}
\lefteqn{\norm{\paren{\mb A^T\mb A}_{I,I}\mb x_i}{2}^2}\\
&=\norm{\diag\paren{\mb A^T\mb A}\mb x_i+\offdiag\paren{\mb A^T\mb A}_{I,I}\mb x_i}2^2\\
&=\norm{\diag\paren{\mb A^T\mb A}\mb x_i}2^2 + \norm{\offdiag\paren{\mb A^T\mb A}_{I,I}\mb x_i}2^2+2\innerprod{\diag\paren{\mb A^T\mb A}\mb x_i}{\offdiag\paren{\mb A^T\mb A}_{I,I}\mb x_i}\\
&\ge\norm{\diag\paren{\mb A^T\mb A}\mb x_i}2^2 -2\norm{\diag\paren{\mb A^T\mb A}\mb x_i}2\norm{\offdiag\paren{\mb A^T\mb A}_{I,I}\mb x_i}2\\
&\ge\theta k\paren{1-2\sqrt{\paren{\theta k}^{-1}\log{k}}-12\mu\sqrt{\theta k\log{k}}}\\
&\ge\theta k/2.
\end{align}

The last equation is derived by plugging in
\begin{equation}
\paren{\theta k}^{-1}\log{k}\le\tfrac1{64},\quad\mu^2\theta k\log{k}\le\tfrac1{48^2}
\end{equation}
under the assumption
\begin{equation}
64k^{-1}\log{k}\le\theta\le\tfrac1{48^2}\mu^{-2}k^{-1}\log^{-1} k.
\end{equation}

\paragraph{Lower Bound of $\norm{\cdot}4^4$.} Since with probability no smaller than $1-4k^{-2}$, $\norm{\offdiag\paren{\mb A^T\mb A}\mb x_i}2^2\le 36\mu^2\theta k^2\log k$ obtains and the relative $\norm{\cdot}2^2$ norm between the flat entries to the spiky entries in $\mb A^T\mb A\mb x_i$ can be bounded as 
\begin{align}
\frac{\norm{\paren{\mb A^T\mb A}_{I^c,I}\mb x_i}2^2}{\norm{\paren{\mb A^T\mb A}_{I,I}\mb x_i}2^2} &\le \frac{\norm{\offdiag\paren{\mb A^T\mb A}\mb x_i}2^2}{\norm{\paren{\mb A^T\mb A}_{I,I}\mb x_i}2^2} \\
&\le { 36}\mu^2 k\log k \doteq r. 
\end{align}
Since
\begin{align}
\norm{\zetahatinit}4^4&=\norm{\mc P_{\bb S}\brac{\mb A^T\mb A\mb x_i}}4^4\\
&=\frac1{\norm{\mb A^T\mb A\mb x_i}2^4}\norm{\paren{\mb A^T\mb A}_{I^c,I}\mb x_i}4^4+\frac1{\norm{\mb A^T\mb A\mb x_i}2^4}\norm{\paren{\mb A^T\mb A}_{I,I}\mb x_i}4^4\\
&\ge\frac1{\norm{\mb A^T\mb A\mb x_i}2^4}\norm{\paren{\mb A^T\mb A}_{I,I}\mb x_i}4^4\\
&=\frac{\norm{\paren{\mb A^T\mb A}_{I,I}\mb x_i}2^4\norm{\mc P_{\bb S}\brac{\paren{\mb A^T\mb A}_{I,I}\mb x_i}}4^4}{\norm{\paren{\mb A^T\mb A}_{I,I}\mb x_i+\paren{\mb A^T\mb A}_{I^c,I}\mb x_i}2^4}\\
&\ge\frac1{\paren{1+r}^2}\norm{\mc P_{\bb S}\brac{\paren{\mb A^T\mb A}_{I,I}\mb x_i}}4^4
\end{align}
and with high probability $1-\exp\paren{-\theta k}$ according to \Cref{lem:ber_sparsity}, $\mc P_{\bb S}\brac{\paren{\mb A^T\mb A}_{I,I}\mb x_i}$ satisfies 
\begin{equation}
\norm{\mc P_{\bb S}\brac{\paren{\mb A^T\mb A}_{I,I}\mb x_i}}4^4\ge\frac1{\norm{\mb x_i}0}\ge\frac{1}{2\theta\paren{2k-1}},
\end{equation}
Together, we have
\begin{align}
\norm{\zetainit}4^4&\ge\norm{\zetahatinit}4^4-{\frac{32\kappa^3\delta}{\sigma_{\min}}}\\
&\ge \frac1{\paren{1+r}^2}\norm{\mc P_{\bb S}\brac{\paren{\mb A^T\mb A}_{I,I}\mb x_i}}4^4-\frac{640C^{-1/4}}{\theta k\paren{1+36\mu^2k\log{k}}^2}\\
&\ge \paren{\frac14-\frac{640}{C^{1/4}}}\frac{1}{\theta k\paren{1+36\mu^2k\log{k}}^2}
\end{align}
holds with probability no smaller than $1-k^{-1}-8k^{-2}-2\exp\paren{-\theta k}-24k\exp\paren{-\tfrac1{144}\min\set{k,3\sqrt{\theta m}}}-48k^{-7} -48m^{-5}$. To make sure $\norm{\zetainit}4^6\ge 3C_{\star}\mu\kappa^2$ as desired, we require the sparsity to satisfy
\begin{equation}
\theta\le \paren{\tfrac1{4}-\tfrac{640}{C^{1/4}}}\paren{3C_{\star}\mu\kappa^2}^{-2/3}k^{-1}\paren{1+36\mu^2k\log^{}k}^{-2},
\end{equation}
then the initialization $\qinit\in\hat{\mc R}_{3C_{\star}}$ follows by \Cref{def:R}.
\end{proof}

\section{Preconditioning}
\label{sec:preconditioning}
\begin{lemma}
\label{lem:preconditioning}
Suppose $\mb x_{0} \simiid \mathrm{BG}\paren{\theta}\in\R^m$, then following inequality holds 
\begin{equation}
\norm{\frac1{\theta m}\mb X_0\mb X_0^T-\mb I}2\le10\sqrt{k\log{m}/m},
\end{equation}
with probability no smaller than $1-2\exp\paren{-\theta k}-24k\exp\paren{-\tfrac1{144}\min\set{k,3\sqrt{\theta m}}}-48k^{-7} -48m^{-5}$.
\end{lemma}

\begin{proof}
Since
\begin{align}
\norm{\frac1{\theta m}\mb X_0\mb X_0^T-\mb I}2\le\norm{\diag\paren{\frac1{\theta m}\mb X_0\mb X_0^T}-\mb I}2+\norm{\offdiag\paren{\frac1{\theta m}\mb X_0\mb X_0^T}}2.
\end{align}
The above term is bounded by $\delta$ with probability no smaller than $1-\eps_d-\eps_o$ 
whenever the probability that each of the terms is upper bounded by $\delta/2$ satisfies
\begin{align}
\prob{\norm{\diag\paren{\frac1{\theta m}\mb X_0\mb X_0^T}-\mb I}2 \ge\delta/2}&\le\eps_d,\\   
\prob{\norm{\offdiag\paren{\frac1{\theta m}\mb X_0\mb X_0^T}-\mb I}2 \ge\delta/2}&\le\eps_o.
\end{align}

\noindent{\bf Diagonal of $\frac1{\theta m}\mb X_0\mb X_0^T$.} Note that $\diag\paren{ \mb X_0\mb X_0^T }=\norm{\mb x_0}2^2\mb I$, so
\begin{equation}
\norm{\diag\paren{\frac1{\theta m}\mb X_0\mb X_0^T}-\mb I}2=\abs{\frac1{\theta m}\norm{\mb x_0}2^2-1}.
\end{equation}
We calculate the moment for each summand of $\norm{\mb x_0}2^2$. The summands can be seen as a $\chi_1^2$ random variable but populated with probability $\theta$, whence
\begin{align}
\bb E_{x_i\sim\mathrm{BG}(\theta)}\brac{\paren{x_i^2}^p} &= \theta\, \bb E_{X_i\sim\chi_1^2}\brac{X_i^p}\\
&= \theta \frac{\Gamma\paren{p + \tfrac{1}{2}}}{\Gamma\paren{\tfrac{1}{2}}}\\
&\le \frac{\theta p!\paren{2}^p}{2} \\
&= \frac{p!}{2} \sigma^2 R^{p-2}. 
\end{align}
Apply Bernstein's inequality for moment bounded random variables \eqref{lem:mc_bernstein_scalar} with $R = 2,\sigma^2=4\theta$, then
\begin{align}
\prob{\abs{\frac1m\norm{\mb x_0}2^2-\theta} \ge t} \leq 2\exp\paren{-\frac{mt^2}{8\theta + 4t}}.   
\end{align}
By taking $t=\tfrac12\theta\delta$, we obtain
\begin{align}
&\prob{\norm{\diag\paren{\frac1{\theta m}\mb X_0\mb X_0^T}-\mb I}2\ge\delta/2}\nonumber\\
&\le 2\exp\paren{-\frac{\theta m\delta^2}{32 + 8\delta}}\\
&\le 2\exp\paren{-\frac{100\theta k\log{m}}{32 + 80\sqrt{k\log{m}/m}}}\\
&\le 2\exp\paren{-\theta k}.
\end{align}

\noindent{\bf Off-diagonal of $\frac1{\theta m}\mb X_0\mb X_0^T$.} Note that $\offdiag\paren{\mb X_0\mb X_0^T}$ is a sub-circulant matrix generated by 
\begin{equation}
\mb r_{\mb x_0}=\brac{r_{\mb x_0}\paren{2k-2},\cdots,0,\cdots,r_{\mb x_0}\paren{2k-2}}^T
\end{equation}
with $r_{\mb x_0}\paren{\tau}=\innerprod{\mb x_0}{\shift{\mb x_0}{\tau}}$ for $\tau = 1,\cdots,2k-2$. Equivalently, we can write
\begin{equation}
\mb r_{\mb x_0} = \mb R_{\mb x_0}^T\mb x_0,
\end{equation}
with 
\begin{equation}
\mb R_{\mb x_0}=\brac{\shift{\mb x_0}{2k-2},\cdots,\mb 0,\cdots,\shift{\mb x_0}{2k-2}}\in\R^{m\times\paren{4k-3}}.
\end{equation}

Operator norm of a circulant matrix is defined as the following 
\begin{equation}
\norm{\offdiag\paren{\frac1{\theta m}\mb X_0\mb X_0^T}}2=\max_{l=0, \dots, 4k-4}\abs{\innerprod{\mb v_l}{\frac1{\theta m}\mb r_{\mb x_0}}},\end{equation}
where $\mb v_l$ is the $l$-th (discrete) Fourier basis vector
\begin{equation}
\mb v_l=\brac{1,~e^{l\frac{2\pi j}{4k-3}},\cdots,~e^{l\paren{4k-4}\frac{2\pi j}{4k-3}}}^T,
\end{equation}
and $j$ is the imaginary unit. Let $v_{l,\tau}=\mb v_l\paren{2k-2-\tau}+\mb v_l\paren{2k-2+\tau}$, then
\begin{align}
\innerprod{\mb v_l}{\mb r_{\mb x_0}}&=\sum_{\tau=1}^{2k-2}v_{l,\tau}\innerprod{\mb x_0}{\shift{\mb x_0}{\tau}}\\
&=\sum_{\tau=1}^{2k-2}v_{l,\tau}\sum_{i=0}^{m-1}\mb x_0\paren{i}\mb x_0\paren{\brac{i+\tau}_m}.
\end{align}
By decoupling (Theorem 3.4.1 of \cite{de1999decoupling}), the tail probability of the weighted autocorrelation $\innerprod{\mb v_l}{\mb r_{\mb x_0}}$ can be upper bounded via
\begin{align}
\prob{\abs{\innerprod{\mb v_l}{\mb r_{\mb x_0}}}>t}
&=\prob{\abs{\sum_{\tau=1}^{2k-2}v_{l,\tau}\innerprod{\mb x_0}{\shift{\mb x_0}{\tau}}}>t}\\
&\le6\,\prob{\abs{\sum_{\tau=1}^{2k-2}v_{l,\tau}\innerprod{\mb x_0}{\shift{\mb x'_0}{\tau}}}>\frac{t}6},
\end{align}
where $\mb x'_0 \simiid \mathrm{BG}\paren{\theta}$ is an independent copy of the random vector $\mb x_0$, we have
Plugging in $\innerprod{\mb v_l}{\mb r_{\mb x_0}}=\innerprod{\mb v_l}{\mb R_{\mb x_0}^T\mb x_0}=\innerprod{\mb R_{\mb x_0}\mb v_l}{\mb x_0}$.

\begin{equation}
\prob{\abs{\innerprod{\mb v_l}{\frac1{\theta m}\mb r_{\mb x_0}}}>t}\le 6\,\prob{\abs{\frac1{\theta m}\innerprod{\mb R_{\mb x'_0}\mb v_l}{\mb x_0}}>\frac{t}6}. \label{eqn:decoupled-bernstein}
\end{equation}
Again with Bernstein's inequality for moment bounded random variable, we have 
\begin{align}
\prob{\abs{\frac1{\theta m}\innerprod{\mb R_{\mb x'_0}\mb v_l}{\mb x_0}}\ge t}
\le2\exp\paren{-\frac{\theta m^2t^2}{2\norm{\mb R_{\mb x'_0}\mb v_l}2^2+2\norm{\mb R_{\mb x'_0}\mb v_l}{\infty}mt}}
\end{align}
\noindent{\bf Control $\norm{\mb R_{\mb x'_0}\mb v_l}2$.}
\begin{equation}
\norm{\mb R_{\mb x_0}\mb v_l}2^2\le\norm{\mb R_{\mb x_0}}2^2\norm{\mb v_l}2^2=k\norm{\mb R_{\mb x_0}}2^2
\end{equation}
With tail bound of the operator norm of a circulant matrix in \Cref{lem:X_norm}, we have
\begin{equation}
\prob{\norm{\mb R_{\mb x_0}}2\ge t}\le 4m\exp\paren{-\frac{t^2}{2\theta m+2t}}
\end{equation}
\noindent{\bf Control $\norm{\mb R_{\mb x'_0}\mb v_l}{\infty}$.}
For a discrete Fourier basis $\mb v_l$ as defined, we have
\begin{equation}
\norm{\mb v_l}2^2 = \norm{\mb v_l}0=4k-3,\quad\norm{\mb v_l}{\infty}=1
\end{equation}
Note that 
\begin{equation}
\norm{\mb R_{\mb x_0}\mb v_l}{\infty}=\max_{\tau=1, \dots, 2k-2}\abs{\innerprod{\shift{\mb x_0}{\tau}}{\mb v_l}}
\end{equation}
and moment control Bernstein inequality implies that 
\begin{equation}
\prob{\abs{\innerprod{\shift{\mb x_0}{\tau}}{\mb v_l}}\ge t}\le2\exp\paren{-\frac{t^2}{2\theta\norm{\mb v_l}2^2+2\norm{\mb v_l}{\infty}t}}.
\end{equation}
with union bound, we obtain
\begin{align}
\prob{\norm{\mb R_{\mb x_0}\mb v_l}{\infty}\ge t}
&\le \sum_{\tau=1}^{2k-2}\prob{\abs{\innerprod{\shift{\mb x_0}{\tau}}{\mb v_l}}\ge t}\\
&\le 4k\exp\paren{-\frac{t^2}{8\theta k+2t}}
\end{align}

Therefore, by plugging in
\begin{align}
&\norm{\mb R_{\mb x'_0}\mb v_l}{\infty}\le t_1=10\sqrt{\theta k\log{k}},\\
&\norm{\mb R_{\mb x'_0}\mb v_l}2\le t_2 = 5\sqrt{\theta m\log{m}},
\end{align}
we obtain the following probabilities
\begin{align}
\prob{\norm{\mb R_{\mb x'_0}\mb v_l}{\infty}\ge t_1}
&\le 4k\exp\paren{-\frac{t_1^2}{8\theta k+2t_1}}\nonumber\\
&\le 4k^{-8},\\
\prob{\norm{\mb R_{\mb x'_0}}2\ge t_2}
&\le 4m\exp\paren{-\frac{t_2^2}{2\theta m+2t_2}}\nonumber\\
&\le 4m^{-6}.
\end{align}
Denoting event
\begin{equation}
\mb E = \set{\norm{\mb R_{\mb x'_0}\mb v_l}{\infty}\le t_1,\norm{\mb R_{\mb x'_0}}2\le t_2},
\end{equation}
and combining these bounds with \eqref{eqn:decoupled-bernstein}, we obtain 
\begin{align}
\lefteqn{\prob{\norm{\offdiag\paren{\frac1{\theta m}\mb X_0\mb X_0^T}}2\ge\delta/2}}\nonumber\\
&\le 6\,\prob{\max_{l}\abs{\frac1{\theta m}\innerprod{\mb R_{\mb x'_0}\mb v_l}{\mb x_0}}\ge \frac{\delta}{12}} \\
&\le12k\,\prob{\abs{\frac1{\theta m}\innerprod{\mb R_{\mb x'_0}\mb v_l}{\mb x_0}}\ge\frac{\delta}{12}}\\
&\le12k\prob{\norm{\mb R_{\mb x'_0}\mb v_l}{\infty}> t_1}+12k\prob{\norm{\mb R_{\mb x'_0}}2>t_2}+12k\prob{\abs{\frac1{\theta m}\innerprod{\mb R_{\mb x'_0}\mb v_l}{\mb x_0}}\ge\frac{\delta}{12}\mid\mb E}\\
&\le 24k\exp\paren{-\frac{100\theta km\log m/144}{50\theta m\log m+\frac{200}{12}k\sqrt{\theta m\log{k}\log{m}}}} + 12k\paren{4k^{-8} + 4m^{-6} }\\
&\qquad\qquad\paren{t_1=10\sqrt{\theta k\log{k}},\;  t_2 = 5\sqrt{\theta m\log{m}}}\nonumber\\
&\le24k\exp\paren{-\tfrac1{144}\min\set{k,3\sqrt{\theta m}}}+ 48k^{-7} + 48m^{-5}
\end{align}
At last, by combining the control for both the diagonal and off-diagonal term, we obtain that with probability no smaller than $1-2\exp\paren{-\theta k}-24k\exp\paren{-\tfrac1{144}\min\set{k,3\sqrt{\theta m}}}-48k^{-7} -48m^{-5}$,
\begin{equation}
\norm{\frac1{\theta m}\mb X_0\mb X_0^T-\mb I}2\le10\sqrt{k\log{m}/m},
\end{equation}
holds and completes the proof. 
\end{proof}

\begin{lemma}\label{lem:precond_neghalf_1}
Suppose $\delta=\norm{\frac{1}{\theta m}\mb X_0\mb X_0^T-\mb I}2\le 1/\paren{2\kappa^2}$, then
\begin{equation}
\norm{\paren{\frac1{\theta m}\mb Y\mb Y^T}^{1/2}\paren{\mb A_0\mb A_0^T}^{-1/2}-\mb I}2\le \kappa^2\delta/\sigma_{\min}.
\end{equation}
\end{lemma}
\begin{proof}As in by \cite{Bhatia1997}, 
we denote the directional derivative of $f$ at direction $\mb \Delta$ with
\begin{equation}
Df(\mb M)\paren{\mb \Delta}=\frac{d}{dt}\bigg|_{t=0}f(\mb  M+t\mb\Delta),
\end{equation}
Denote symmetric matrix $\mb M = \mb A_0\mb A_0^T=\mb U\mb \Lambda\mb U^T$, with $\lambda_{\max}$ and $\lambda_{\min}$ being its maximum and minimum eigenvalue. Then we have
\begin{equation}
\frac1{\theta m}\mb Y\mb Y^T = \mb M+\mb \Delta,\quad\norm{\mb\Delta}2\le\lambda_{\max}\delta.
\end{equation}

Then derivative of $f$ with $Df(\mb M)$. By differential calculus, we can obtain that 
\begin{align}
&\norm{\paren{\frac1{\theta m}\mb Y\mb Y^T}^{1/2}\paren{\mb A_0\mb A_0^T}^{-1/2}-\mb I}2\nonumber\\
&=\norm{\paren{\mb A_0\mb A_0^T+\mb\Delta}^{1/2}\paren{\mb A_0\mb A_0^T}^{-1/2}-\mb I}2\\
&=\norm{\paren{\mb A_0\mb A_0^T}^{-1/2}\int_{t=0}^1Df\paren{\mb A_0\mb A_0^T+t\mb\Delta}\paren{\mb \Delta}dt}2\\
&\le\sup_{t\in[0,1]}\norm{Df\paren{\mb A_0\mb A_0^T+t\mb\Delta}}2\norm{\mb\Delta}2\norm{\paren{\mb A_0\mb A_0^T}^{\!-\!1\!/2}}2\\
&\le\sup_{t\in[0,1]}\norm{Df\paren{\mb A_0\mb A_0^T+t\mb\Delta}}2\lambda_{\max}\delta/\sigma_{\min}
\end{align}
Moreover, we denote $f(t)=t^{1/2}$ and $g(t) = t^2$, then $f = g^{-1}$. The directional derivative of $g$ has following form
\begin{equation}
Dg\paren{\mb M}\paren{\mb X}= \mb M\mb X+\mb X\mb M,
\end{equation}
and directional derivative $\mb Z= Df\paren{\mb M}\paren{\mb X}$ satisfies
\begin{equation}
\mb M\mb Z+\mb Z\mb M=\mb X.
\end{equation}

Denote $\mb M=\mb U\mb\Lambda\mb U^T$ with $\mb U$ orthogonal, without loss of generality,
\begin{equation}
\mb\Lambda\mb Z +\mb Z\mb\Lambda=\mb X.
\end{equation}

Applying Theorem VII.2.3 of \cite{Bhatia1997}, we have 
\begin{align}
\norm{Df\paren{\mb M}\paren{\mb X}}2&=\sup_{\norm{\mb X}2\le 1}\norm{\mb Z}2\\
&\le\int_{t=0}^{\infty}\norm{e^{-\mb\Lambda t}\mb Xe^{-\mb\Lambda t}}2dt\\
&\le\int_{t=0}^{\infty}e^{-2\lambda_{\min}t}\norm{\mb X}2dt
\end{align}
and
\begin{align}
\sup_{t\in[0,1]}\norm{Df\paren{\mb A_0\mb A_0^T+t\mb\Delta}}2
&\quad\le\frac{\norm{\mb X}2}{2\paren{\lambda_{\min}-\lambda_{\max}\delta}}\\
&\quad\le1/\lambda_{\min}.
\end{align}
Therefore,
\begin{equation}
\norm{\paren{\frac1{\theta m}\mb Y\mb Y^T}^{1/2}\paren{\mb A_0\mb A_0^T}^{-1/2}-\mb I}2\le \kappa^2\delta/\sigma_{\min}.
\end{equation}
\end{proof}

\begin{lemma}
\label{lem:precond_neghalf_2}
Suppose $\mb A_0$ has condition number $\kappa$ and 
\begin{equation}
\delta=\norm{\frac{1}{\theta m}\mb X_0\mb X_0^T-\mb I}2\le1/\paren{2\kappa^2}，
\end{equation}
then
\begin{equation}
\norm{\paren{\frac1{\theta m}\mb Y\mb Y^T}^{-1/2}-\paren{\mb A_0\mb A_0^T}^{-1/2}}2\le 4\kappa^2\delta/\sigma^2_{\min}.
\end{equation}
\end{lemma}
\begin{proof}
Denote symmetric matrix 
\begin{align}
\mb M = \mb A_0\mb A_0^T=\mb U\mb \Lambda\mb U^T, 
\end{align}
with $\lambda_{\max}$ and $\lambda_{\min}$ being its maximum and minimum eigenvalue. Then we have
\begin{equation}
\frac1{\theta m}\mb Y\mb Y^T = \mb M+\mb \Delta,\quad\norm{\mb\Delta}2\le\lambda_{\max}\delta.
\end{equation}
Then
\begin{align}
\norm{\paren{\frac1{\theta m}\mb Y\mb Y^T}^{-1/2}\!\!-\paren{\mb A_0\mb A_0^T}^{-1/2}}2
&\quad=\norm{\paren{\mb M+\mb\Delta}^{-1/2}-\mb M^{-1/2}}2\\
&\quad\le\norm{\mb\Delta}2 \cdot \sup_{0\le t\le1}\norm{Df\paren{\mb M+t\mb\Delta}}2.
\end{align}
Here, $f(t)=t^{-1/2}$ and $Df$ is the derivative of function $f$. In addition, we define function $g(t)=t^{-2}$, $h(t)=t^{-1}$, $w(t)=t^{2}$, and following function compositions hold
\begin{equation}
f=g^{-1},\quad g=h\circ w.
\end{equation}
For differential function $g$ and if $Dg\paren{f\paren{\mb M}}\neq 0$, we have
\begin{equation}
Df\paren{\mb M}=\brac{Dg\paren{f\paren{\mb M}}}^{-1}.
\end{equation}
The derivative of function $g$ satisfies the chain rule that
\begin{equation}
Dg\paren{\mb M}=Dh\paren{w\paren{\mb M}}\paren{Dw\paren{\mb M}}.
\end{equation}
Plug in
\begin{align}
Dh\paren{\mb M}\paren{\mb X} &= -\mb M^{-1}\mb X\mb M^{-1},\\
Dw\paren{\mb M}\paren{\mb X} &= \mb M\mb X+\mb X\mb M,
\end{align}
we obtain that
\begin{align}
Dg\paren{\mb M}\paren{\mb X}
&= Dh\paren{w\paren{\mb M}}\paren{Dw\paren{\mb M}\paren{\mb X}}\\
&= Dh\paren{w\paren{\mb M}}\brac{\mb M\mb X+\mb X\mb M}\\
&= Dh\paren{\mb M^2}\brac{\mb M\mb X+\mb X\mb M}\\
&= -\mb M^{-2}\brac{\mb M\mb X+\mb X\mb M}\mb M^{-2}\\
&= -\brac{\mb M^{-1}\mb X\mb M^{-2}+\mb M^{-2}\mb X\mb M^{-1}}.
\end{align}
Since the function $g$ is differentiable and $Dg(\mb M)\neq\mb 0$, then 
\begin{align}
Df\paren{\mb M}&=\brac{Dg\paren{f\paren{\mb M}}}^{-1}\\
&=\brac{Dg\paren{\mb M^{-1/2}}}^{-1}.
\end{align}
Hence, directional derivative $\mb Z\doteq Df\paren{\mb M}\paren{\mb X}$ satisfies
\begin{equation}
\mb M^{1/2}\mb Z\mb M+\mb M\mb Z\mb M^{1/2}=-\mb X.
\end{equation}
Denote $\mb M=\mb U\mb\Lambda\mb U^T$ with $\mb\Lambda\succ0$ and $\mb U$ orthogonal, without loss of generality
\begin{equation}
\mb\Lambda\mb Z\mb\Lambda^{1/2}+\mb\Lambda^{1/2}\mb Z\mb\Lambda=-\mb X.
\end{equation}
Above equation can be reformulated as a Sylvester equation as following
\begin{equation}
\mb\Lambda^{1/2}\mb Z-\mb Z\paren{-\mb\Lambda^{1/2}}=-\mb\Lambda^{-1/2}\mb X\mb\Lambda^{-1/2}.
\end{equation}
From Theorem VII.2.3 of \cite{Bhatia1997}, when there are no common eigenvalues of $\mb\Lambda^{1/2}$ and $-\mb\Lambda^{1/2}$, then there exists a closed form solution for matrix $\mb Z$ that
\begin{equation}
\mb Z = \int_{t=0}^{\infty}e^{-\mb\Lambda^{1/2}t}\paren{-\mb\Lambda^{-1/2}\mb X\mb\Lambda^{-1/2}}e^{-\mb\Lambda^{1/2}t}dt
\end{equation}

Therefore, the operator norm of $Df\paren{\mb M}$ can be obtained as
\begin{align}
\norm{Df(\mb M)(\mb X)}2&=\sup_{\norm{\mb X}2\le 1}\norm{\mb Z}2\\
&\;\le\int_{t=0}^{\infty}\norm{e^{-\mb\Lambda^{1/2}t}\paren{\mb\Lambda^{-1/2}\mb X\mb\Lambda^{-1/2}}e^{-\mb\Lambda^{1/2}t}}{}dt\\
&\;\le\int_{t=0}^{\infty}e^{-\lambda_{\min}t}\norm{\mb\Lambda^{-1/2}\mb X\mb\Lambda^{-1/2}}{}dt\\
&\;\le\frac{\norm{\mb X}{}}{\lambda^2_{\min}}.
\end{align}

Therefore  
\begin{align}
\norm{\paren{\mb M+\mb\Delta}^{-1/2}-\mb M^{-1/2}}2
&\le\frac{\norm{\mb\Delta}2}{\paren{\lambda_{\min}-\norm{\mb\Delta}2}^2}\\
&\le\frac{4\norm{\mb\Delta}2}{\lambda^2_{\min}}\qquad\paren{\delta\le1/\paren{2\kappa^2}}\\
&\le\frac{4\lambda_{\max}\delta}{\lambda^2_{\min}}\\
&=\frac{4\kappa^2\delta}{\sigma^2_{\min}}
\end{align}

\end{proof}

\section{Concentration for Gradient (Lemma \ref{lem:grad_scale})}
\label{sec:grad_scale}
\begin{lemma}
Suppose $\mb x_0\simiid\mathrm{BG}\paren{\theta}$. There exists a positive constant $C$ such that whenever
\begin{equation}
m\ge C\frac{\min\set{\paren{2C_{\star}\mu}^{-1}\!\!\!\!,\kappa^2k^2}}{\paren{1-\theta}^2\sigma^2_{\min}}\kappa^8 k^4\log^3\paren{\frac{\kappa k}{\paren{1-\theta}\sigma_{\min}}}
\end{equation}
and $\theta >\log{k}/k$, then with probability no smaller than $1-c_1\exp\paren{-k}-c_2k^{-4}-2\exp\paren{-\theta k}- 24k\exp\paren{-\tfrac1{144}\min\set{k,3\sqrt{\theta m}}} - 48k^{-7} - 48m^{-5}$,
\begin{equation}
\norm{\grad[\psi]\paren{\mb q}-\frac{3\paren{1-\theta}}{\theta m^2}\grad[\varphi]\paren{\mb q}}2\le c\frac{1-\theta}{\theta m^2}\frac{\norm{\mb A^T\mb q}4^6}{\kappa^2},
\end{equation}
holds for all $\mb q\in\hat{\mc R}_{2C_{\star}}$ with positive constant $c\le 3/\paren{2C_\star}$.
\end{lemma}

\begin{proof}Denote $\mb\eta=\mb Y^T\paren{\mb Y\mb Y^T}^{-1/2}\mb q$ and $\mb{\bar\eta}=\mb Y^T\paren{\theta m\mb A_0\mb A_0^T}^{-1/2}\mb q=\paren{\theta m}^{-1/2}\mb X_0^T\mb\zeta$, then
\begin{align*}
\lefteqn{\norm{\grad\brac{\psi}\paren{\mb q}-\frac{3\paren{1-\theta}}{\theta m^2}\grad\brac{\varphi}\paren{\mb q}}2}\\
&= \norm{\mb P_{\mb q^{\perp}}\brac{\frac1m\paren{\mb Y\mb Y^T}^{-1/2}\mb Y\mb\eta^{\circ3}-\frac{3\paren{1-\theta}}{\theta m^2}\mb A\mb\zeta^{\circ3}}}2\\
&\le \underbrace{\frac1m\norm{\paren{\mb Y\mb Y^T}^{-1/2}\mb Y\mb\eta^{\circ3}-\paren{\theta m}^{-1/2}\mb A\mb X_0\mb\eta^{\circ3}}2}_{\Delta^g_1}\\
&\quad+ \underbrace{\frac1{\theta^{1/2}m^{3/2}}\norm{\mb A\mb X_0\mb\eta^{\circ3}-\mb A\mb X_0\mb{\bar\eta}^{\circ3}}2}_{\Delta^g_2}\\
&\quad+ \underbrace{\norm{\mb P_{\mb q^{\perp}}\brac{\frac1{\theta^{1/2}m^{3/2}}\mb A\mb X_0\mb{\bar\eta}^{\circ3}-\frac{3\paren{1-\theta}}{\theta m^2}\mb A\mb\zeta^{\circ3}}}2}_{\Delta^g_3}.
\end{align*}
First, let us note that
\begin{align}
&C\paren{1-\theta}^{-2}\sigma_{\min}^{-2}\kappa^{10}k^6\log^3\paren{\frac{\kappa k}{\paren{1-\theta}\sigma_{\min}}}\nonumber\\
&\le C\paren{\frac{\kappa k}{\sigma_{\min}\paren{1-\theta}}}^{10}\log^3\paren{\frac{\kappa k}{\paren{1-\theta}\sigma_{\min}}}\\
&\le C\paren{\frac{\kappa k}{\paren{1-\theta}\sigma_{\min}}}^{13},
\end{align}
hence
\begin{align}
\lefteqn{\frac{\log^3\paren{C\paren{1-\theta}^{-2}\sigma_{\min}^{-2}\kappa^{10}k^6\log^3\paren{\paren{1-\theta}^{-1}\sigma_{\min}^{-1}\kappa k}}}{C\log^3\paren{\paren{1-\theta}^{-1}\sigma_{\min}^{-1}\kappa k}}}\nonumber\\
&\le\paren{\frac{\log{C}+13\log\paren{\paren{1-\theta}^{-1}\sigma_{\min}^{-1}\kappa k}}{C^{1/3}\log \paren{\paren{1-\theta}^{-1}\sigma_{\min}^{-1}\kappa k}}}^3\\
&\le\paren{\frac{\log{C}}{C^{1/3}\log\paren{\paren{1-\theta}^{-1}\sigma_{\min}^{-1}\kappa k}}+\frac{13}{C^{1/3}}}^3\\
&\le\paren{\frac1{C^{1/6}}+\frac12\frac1{C^{1/6}}}^3\qquad\paren{C\ge 10^8}\\
&\le\frac4{C^{1/2}}.
\end{align}

Given 
\begin{equation}
m\ge C\frac{\min\set{\paren{2C_{\star}\mu}^{-1},\kappa^2k^2}}{\paren{1-\theta}^2\sigma^2_{\min}}\kappa^8 k^4\log^3\paren{\frac{\kappa k}{\sigma_{\min}\paren{1-\theta}}},
\end{equation}
as the ratio $\log^3{m}/m$ decreases with increasing $m$, then
\begin{align}
\frac{\log^3{m}}{m}
&\le \frac{ \log^3\paren{\frac{C\kappa^{10}k^6}{\paren{1-\theta}^2\sigma^2_{\min}} \log^3\paren{\frac{\kappa k}{\paren{1-\theta}\sigma_{\min}}}} }
{ C\log^3\paren{\frac{\kappa k}{\paren{1-\theta}\sigma_{\min}}} }\frac{\paren{1-\theta}^2\sigma^2_{\min}}{\min\set{\paren{2C_{\star}\mu}^{-1},\kappa^2k^2}\kappa^8 k^4}\\
&\le\frac4{C^{1/2}}\frac{\paren{1-\theta}^2\sigma^2_{\min}}{\min\set{\paren{2C_{\star}\mu}^{-1},\kappa^2k^2}\kappa^8 k^4}
\end{align}
According to \Cref{lem:preconditioning}, following inequality always holds 
\begin{align}
\norm{\frac1{\theta m}\mb X_0\mb X_0^T-\mb I}2&\le\delta\\
&\le10\sqrt{k\log{m}/m}\\
&\le\frac{20\paren{1-\theta}\sigma_{\min}\max\set{\paren{2C_{\star}\mu}^{1/2},\paren{\kappa k}^{-1}}}{C^{1/4}\kappa^4 k^{3/2}\log{m}}\\
&\le\frac{20\sigma_{\min}}{C^{1/4}\kappa^3}\frac{\paren{1-\theta}\norm{\mb A^T\mb q}4^6}{\kappa^2k\log{m}},\qquad\forall\mb q\in\hat{\mc R}_{2C_{\star}} .
\end{align}
with probability no smaller than $1-\eps_0$ with $\eps_0=2\exp\paren{-\theta k} + 24k\exp\paren{-\tfrac1{144}\min\set{k,3\sqrt{\theta m}}} + 48k^{-7} + 48m^{-5}$.

Moreover, $4\kappa^3\delta/\sigma_{\min}\le1/2$ whenever
\begin{equation}
C \ge\paren{\frac{160\paren{1-\theta}}{k\log{m}}}^4,
\end{equation}
whence $\delta \le 1/\paren{8 \kappa^2}$, and \Cref{lem:precond_neghalf_2} implies that 
\begin{align}
\norm{\paren{\frac1{\theta m}\mb Y\mb Y^T}^{-1/2}\mb A_0-\paren{\mb A_0\mb A_0^T}^{-1/2}\mb A_0}2
&\le4\kappa^3\delta/\sigma_{\min}\\
&\le\frac{80\paren{1-\theta}}{C^{1/4}k\log{m}}\frac{\norm{\mb A^T\mb q}4^6}{\kappa^2},\qquad\forall\mb q\in\hat{\mc R}_{2C_{\star}}.
\end{align}
At the same time,
\begin{equation}
\norm{\mb X_0}2\le\paren{\theta m}^{1/2}\sqrt{1+\delta}\le\paren{\theta m}^{1/2}\paren{1+\delta/2}.
\end{equation}
Moreover, \Cref{lem:truncate} implies that with probability no smaller than $1-\eps_B$, we have
\begin{equation}
\norm{\mb x_0}{\infty}\le\sqrt2\log^{1/2}\paren{\frac{2\theta m}{\eps_B}}.
\end{equation}

\noindent{\bf Upper Bound for $\Delta^g_1$.}
Using \Cref{lem:eta_norms}, on the an event of probability at least $1-\eps_0-\eps_B$, 
\begin{align}
\norm{\mb\eta^{\circ3}}2&=\norm{\mb\eta}6^3\\
&\le \paren{1+\frac{4\kappa^3\delta}{\sigma_{\min}}}^2\frac{2k}{\theta m}\norm{\mb x_0}{\infty}^2\\
&\le \frac{9k}{\theta m}\log\paren{\frac{2\theta m}{\eps_B} }.
\end{align}
Therefore, we can obtain following upper bound
\begin{align}
\Delta^g_1&=\frac1m\norm{\paren{\mb Y\mb Y^T}^{-1/2}\mb Y\mb\eta^{\circ3}-\paren{\theta m}^{-1/2}\mb A\mb X_0\mb\eta^{\circ3}}2\\
&\le\frac1{\theta^{1/2}m^{3/2}}\norm{\mb X_0}2 \norm{\mb\eta^{\circ3}}2\norm{\paren{\frac1{\theta m}\mb Y\mb Y^T}^{-1/2}\mb A_0-\mb A}2 \\
&\le\frac {5}{4m}
\cdot \frac{4\kappa^3\delta}{\sigma_{\min}}\cdot\frac{9k}{\theta m}\log\paren{\frac{2\theta m}{\eps_B } } \\
&\le \frac{900\paren{1-\theta}\log\paren{2\theta m/\eps_B}}{C^{1/4}\theta m^2\log{m}}\frac{\norm{\mb A^T\mb q}4^6}{\kappa^2}\quad\forall \mb q \in \hat{\mc R}_{2C_\star}.
\end{align}

\noindent{\bf Upper Bound for $\Delta^g_2$.} Similarly, with probability no smaller than $1-\eps_0-\eps_B$, together with \Cref{lem:eta_norms}, following upper bound can be obtained
\begin{align}
&\norm{\mb\eta^{\circ3}-\mb{\bar\eta}^{\circ3}}2\nonumber\\
&=\norm{\mb\eta^{\circ3}-\diag\paren{\mb\eta^{\circ2}}\mb{\bar\eta}+\diag\paren{\mb\eta^{\circ2}}\mb{\bar\eta}-\mb{\bar\eta}^{\circ3}}2\\
&\le\norm{\mb\eta-\mb{\bar\eta}}2\norm{\diag\paren{\mb\eta^{\circ2}}}2+\norm{\mb{\bar\eta}}2\norm{\diag\paren{\mb\eta^{\circ2}-\mb{\bar\eta}^{\circ2}}}2\\
&=\norm{\mb\eta-\mb{\bar\eta}}2\norm{\mb\eta}{\infty}^2+\norm{\mb{\bar\eta}}2\norm{\mb\eta^{\circ2}-\mb{\bar\eta}^{\circ2}}{\infty}\\
&\le\norm{\mb\eta-\mb{\bar\eta}}2\norm{\mb\eta}{\infty}^2+\norm{\mb{\bar\eta}}2\norm{\mb\eta-\mb{\bar\eta}}{\infty}\norm{\mb\eta+\mb{\bar\eta}}{\infty} \\
&\le 4\paren{1+\delta/2}\frac{4\kappa^3\delta}{\sigma_{\min}}\frac{k}{\theta m}\log\paren{2\theta m/\eps_B} \brac{\paren{1+\frac{4\kappa^3\delta}{\sigma_{\min}}}^2+\paren{2+\frac{4\kappa^3\delta}{\sigma_{\min}}}}\\
&\le\frac{24k}{\theta m}\log\paren{2\theta m/\eps_B}\cdot\frac{4\kappa^3\delta}{\sigma_{\min}}.
\end{align}
Therefore, we can obtain following upper bound 
\begin{align}
\Delta^g_2&=\frac1{\theta^{1/2}m^{3/2}}\norm{\mb A\mb X_0^T\mb\eta^{\circ3}-\mb A\mb X_0^T\mb{\bar\eta}^{\circ3}}2\\
&\le\frac1{\theta^{1/2}m^{3/2}}\norm{\mb A}2\norm{\mb X_0}2\norm{\mb\eta^{\circ3}-\mb{\bar\eta}^{\circ3}}2\\
&\le\frac5{4m}\cdot\frac{24k}{\theta m}\log\paren{2\theta m/\eps_B}\cdot\frac{4\kappa^3\delta}{\sigma_{\min}}\\
&\le\frac{2400}{C^{1/4}}\frac{1-\theta}{\theta m^2}\frac{\norm{\mb A^T\mb q}4^6}{\kappa^2} \cdot\frac{\log\paren{2\theta m/\eps_B}}{\log{m}}.
\end{align}
For both $\Delta^g_1$ and $\Delta^g_2$ to be bounded by $\frac1{2C_\star}\frac{1-\theta}{\theta m^2}\frac{\norm{\mb A^T\mb q}4^6}{\kappa^2}$, we set
\begin{equation}
C\ge\paren{4800C_\star \frac{\log\paren{2\theta m/\eps_B}}{\log{m}}}^4.
\end{equation}
Notice that the right hand side is indeed bounded by a numerical constant for all $m$. \\

\noindent{\bf Tail Bound for $\Delta^g_3$.}
Note that
\begin{align}
&\paren{\mb A_0\mb A_0^T}^{-1/2}\mb Y\mb{\bar\eta}^{\circ3}\nonumber\\
&=\paren{\mb A_0\mb A_0^T}^{-1/2}\!\!\!\!\mb A_0\mb X_0\paren{\mb Y^T\paren{\theta m\mb A_0\mb A_0^T}^{-1/2}\!\!\!\!\mb q}^{\circ3}\\
&=\paren{\theta m}^{-3/2}\mb A\mb X_0\paren{\mb X_0^T\mb A^T\mb q}^{\circ3},
\end{align}
and its expectation with respect to $\mb x_0$
\begin{align}
\lefteqn{\bb E\brac{\frac1m\mb A\mb X_0\paren{\mb X_0^T\mb A^T\mb q}^{\circ3}}}\nonumber\\
&=\bb E\brac{\mb A\mb x_i\paren{\mb x_i^T\mb A^T\mb q}^{3}}\\
&=3\theta\paren{1-\theta}\mb A\mb\zeta^{\circ3}+3\theta^2\norm{\mb A^T\mb q}2^2\mb A\mb A^T\mb q\\
&=3\theta\paren{1-\theta}\mb A\mb\zeta^{\circ3}+3\theta^2\mb q,
\end{align}
hence
\begin{align}
\mb P_{\mb q^{\perp}}\brac{\bb E\brac{\frac1m\mb A\mb X_0\paren{\mb X_0^T\mb A^T\mb q}^{\circ3}}}=\mb P_{\mb q^{\perp}}\brac{3\theta\paren{1-\theta}\mb A\mb\zeta^{\circ3}}.
\end{align}
Therefore, the $\Delta^g_3$ term can be simplified as 
\begin{align}
\Delta^g_3&=\norm{\mb P_{\mb q^{\perp}}\brac{\frac1{\theta^{1/2}m^{3/2}}\mb A\mb X_0^T\mb{\bar\eta}^{\circ3}-\frac{3\paren{1-\theta}}{\theta m^2}\mb A\mb\zeta^{\circ3}}}2\\
&=\frac1{\theta^2m^2}\norm{\mb P_{\mb q^{\perp}}\brac{\frac{\mb A\mb X_0\paren{\mb X_0^T\mb\zeta}^{\circ3}\!\!}m-3\theta\paren{1-\theta}\mb A\mb\zeta^{\circ3}}}2\\
&\le\frac1{\theta^2m^2}\norm{\mb P_{\mb q^{\perp}}\brac{\frac{\mb A\mb X_0\paren{\mb X_0^T\mb\zeta}^{\circ3}\!\!}m-\bb E\brac{\cdot}}}2+\frac1{\theta^2m^2}\norm{\mb P_{\mb q^{\perp}}\brac{3\theta^2\mb q}}2\\
&\le\frac1{\theta^2m^2}\norm{\frac1m\mb X_0\paren{\mb X_0^T\mb\zeta}^{\circ3}-\bb E\brac{\cdot}}2.
\end{align}

Under the assumption that 
\begin{align}
m\ge \frac{C}{\paren{1-\theta}^2}\min\set{\mu^{-1},\kappa^2k^2}\kappa^2k^4\log^3\paren{\kappa k},
\end{align}
applying \Cref{lem:grad_sub}, we have
\begin{equation}
\norm{\frac1m\mb X_0\paren{\mb X_0^T\mb A^T\mb q}^{\circ3}-\bb E\brac{\cdot}}2\le c\theta\paren{1-\theta}\frac{\norm{\mb A^T\mb q}4^6}{\kappa^2}.
\end{equation}
with probability larger than $1-c_2\exp\paren{-k}-c_2k^{-4}$.
At last, taking $\eps_B=\theta^2k^{-4}$, we obtain that
\begin{align}
\norm{\grad\brac{\psi}\paren{\mb q}-\frac{3\paren{1-\theta}}{\theta m^2}\grad\brac{\varphi}\paren{\mb q}}2\le c\frac{1-\theta}{\theta m^2}\frac{\norm{\mb A^T\mb q}4^6}{\kappa^2},\qquad\forall \mb q\in\hat{\mc R}_{2C_\star}
\end{align}
with probability larger than $1-c_2\exp\paren{-k}-c_2k^{-4}-\eps_B-\eps_0$ as desired.
\end{proof}

\subsection{Proof of Lemma \ref{lem:grad_sub}}
\begin{lemma}
\label{lem:grad_sub}
Suppose $\mb x_0\simiid \mathrm{BG}\paren{\theta}\in\R^m$. There exist positive constant $C$ such that whenever 
\begin{equation}
m\ge \frac{C}{\paren{1-\theta}^2}\min\set{\paren{2C_\star\mu}^{-1},\kappa^2k^2}\kappa^2 k^4\log^3\paren{\kappa k}
\end{equation}
and $\theta k\ge 1$, then with probability no smaller than $1-c_1\exp\paren{-k}-c_2k^{-4}$,
\begin{equation}
\norm{\frac1m\mb X_0\paren{\mb X_0^T\mb A^T\mb q}^{\circ3}-\bb E\brac{\cdot}}2\le c\theta\paren{1-\theta}\frac{\norm{\mb A^T\mb q}4^6}{\kappa^2}
\end{equation}
holds for all $\mb q\in\hat{\mc R}_{2C_{\star}}$ with positive constant $c\le 1/\paren{2C_{\star}}$.
\end{lemma}

\begin{proof}
Let $\mb{\bar x}_i \in \R^{2k-1}$ be generated via
\begin{equation}
\mb{\bar x}_i = \begin{cases}
\mb x_i & \quad \norm{\mb x_i}{\infty}\le B \text{ and } \norm{\mb x_i}0\le4\theta k\log{m}\\
\mb 0  & \quad \text{else }
\end{cases}
\end{equation}
Let $\mb{\bar X}_0 \in \R^{\paren{2k-1} \times m}$ denote the circulant submatrix generated by $\mb{\bar x}_0$. Then $\mb{\bar X}_0=\mb X_0$ obtains whenever
\begin{enumerate}
\item $\norm{\mb x_0}{\infty}\le B$, which happens with probability no smaller than $1-2\theta me^{-B^2/2}$ according to \Cref{lem:truncate};
\item $\norm{\mb x_i}0\le 4\theta k\log{m}$ holds for any index $i$, applying \Cref{lem:ber_sparsity} and Boole's inequality we have
\begin{align}
\bb E\brac{\mb 1_{\bigcup_i\norm{\mb x_{i}}{0}>4\theta k \log{m}}}
&\le m\prob{\norm{\mb x_{i}}{0}>4\theta k \log{m}}\\
&\le 2m\exp\paren{-\tfrac34\theta k\log{m}}.
\end{align}

\end{enumerate}

Denote $\mb\zeta=\mb A^T\mb q$ and
\begin{align}
\mb g_E&=\bb E\brac{\frac1m\mb X_0\paren{\mb X_0^T\mb A^T\mb q}^{\circ3}},\\
\mb{\bar g}_E&=\bb E\brac{\frac1m\mb{\bar X}_0 \paren{\mb{\bar X}_0^T\mb A^T\mb q}^{\circ3}},
\end{align}
then, 
\begin{align}
\lefteqn{ \prob{\norm{\frac1m\mb X_0\paren{\mb X_0^T\mb\zeta}^{\circ3}-\mb g_E}2\ge c\theta\paren{1-\theta}\frac{\norm{\mb\zeta}4^6}{\kappa^2}}}\nonumber\\
&\le \prob{\norm{\frac1m\mb{\bar X}_0\paren{\mb{\bar X}_0^T\mb\zeta}^{\circ3}-\mb g_E}2\ge c\theta\paren{1-\theta}\frac{\norm{\mb\zeta}4^6}{\kappa^2}}+\;2\theta me^{-B^2/2}+2m\exp\paren{-\tfrac34\theta k\log{m}}
\end{align}
With triangle inequality, we have 
\begin{align}
\norm{\frac1m\mb{\bar X}_0 \paren{\mb{\bar X}_0^T\mb\zeta}^{\circ3}-\mb g_E}2
\le \norm{\bb E\brac{\frac1m\mb{\bar X}_0 \paren{\mb{\bar X}_0^T\mb\zeta}^{\circ3}}-\mb{\bar g}_E}2+\norm{\mb{\bar g}_E-\mb{g}_E}2.
\end{align}
Hence, provided
\begin{equation}
\norm{ \bar{\mb g}_E - \mb g_E }{2} \le \frac{c}{2} \theta (1- \theta ) \frac{ \norm{\mb\zeta}{4}^6}{\kappa^2},
\end{equation}
we have 
\begin{align}
&\prob{\norm{\frac1m\mb{\bar X}_0 \paren{\mb{\bar X}_0^T\mb\zeta}^{\circ3}-\mb{g}_E}2\ge c\theta\paren{1-\theta}\frac{\norm{\mb\zeta}4^6}{\kappa^2}}\nonumber\\
&\le\prob{\norm{\frac1m\mb{\bar X}_0 \paren{\mb{\bar X}_0^T\mb\zeta}^{\circ3} - \mb{\bar g}_E}2\ge \frac{c}2\theta\paren{1-\theta}\frac{\norm{\mb\zeta}4^6}{\kappa^2}}.
\end{align}

\noindent{\bf Truncation Level}
Next, we choose a large enough entry-wise truncation level $B$ such that the expectation of the gradient $\bb E\brac{\frac1m\mb X_0\paren{\mb X_0^T\mb\zeta}^{\circ3}}$ is close to that of its truncation $\bb E\brac{\frac1m\mb{\bar X}_0\paren{\mb{\bar X}_0^T\mb\zeta}^{\circ3}}$. 

Moreover, we introduce following events notation
\begin{equation}
\event_i\doteq\set{\norm{\mb x_i}{\infty}>B\;\cup\;\norm{\mb x_i}0>4\theta k\log{m}},
\end{equation}
then
\begin{align}
\lefteqn{\norm{\mb{\bar g}_E-\mb{g}_E}2}\nonumber\\
&= \norm{\bb E\brac{\frac1m\sum_i\mb x_i\innerprod{\mb x_i}{\mb\zeta}^3\cdot\mb 1_{\mc E_i} }}2\\
&\le \frac1m\sum_i\norm{\bb E\brac{\mb x_i\innerprod{\mb x_i}{\mb\zeta}^3\cdot\mb 1_{\mc E_i} }}2\\
&\le \frac1m\sum_i\paren{\bb E\brac{\norm{\mb x_i\paren{\mb x_i^T\mb\zeta}^{\circ3}}2^2}\cdot\bb E\brac{\mb 1_{\mc E_i} }}^{1/2}\\
&\le \paren{\bb E\brac{\norm{\mb x_i}2^8}}^{1/2}\sqrt{\bb E\brac{\mb 1_{\norm{\mb x_i}{\infty}>B}}+\bb E\brac{\mb 1_{\norm{\mb x_i}0>4\theta k\log{m}}}}\\
&\le 50k^2\sqrt{4\theta ke^{-B^2/2}+\exp\paren{-\tfrac34\theta k\log{m} } }\label{eqn:diff_bound}
\end{align}
By setting
\begin{equation}
B\ge C'\log^{1/2}\paren{ \frac{\kappa^4k^8}{\theta\paren{1-\theta}^2} },
\end{equation}
we have
\begin{equation}
\theta ke^{-B^2/2}\le \frac12\paren{\frac{c}{100}}^2\theta^2\paren{1-\theta}^2\frac{\norm{\mb\zeta}4^{12}}{\kappa^4k^4}.
\end{equation}
In addition, whenever
\begin{equation}
\theta k\ge\frac{4}{3\log{m}}\log\paren{\frac{400^2\kappa^4k^4}{c^2\theta^2\paren{1-\theta}^2\norm{\mb\zeta}4^{12}}},
\end{equation}
we have
\begin{equation}
\exp\paren{-\tfrac34\theta k\log{m}} \le\frac12\paren{\frac{c}{100}}^2\theta^2\paren{1-\theta}^2\frac{\norm{\mb\zeta}4^{12}}{\kappa^4k^4}.
\end{equation}
Therefore,
\begin{equation}
\label{eqn:fail_prob}
\sqrt{4\theta ke^{-B^2/2}+\exp\paren{-\tfrac34\theta k\log{m}}}\le \frac{c}2\theta\paren{1-\theta}\frac{\norm{\mb\zeta}4^{6}}{50\kappa^2k^2}.
\end{equation}

In addition,
\begin{equation}
\label{eqn:xi_28_moment}
\paren{\bb E\brac{\norm{\mb x_i}2^8}}^{1/2}\le\paren{7!!\cdot2^4k^4}^{1/2}< 50k^2.
\end{equation}

Plugging in Eq \eqref{eqn:xi_28_moment} and \eqref{eqn:fail_prob} back to \eqref{eqn:diff_bound}, we obtain that
\begin{equation}
\norm{\mb{\bar g}_E-\mb{g}_E}2\le\frac{c}2\theta\paren{1-\theta}\frac{\norm{\mb A^T\mb q}4^6}{\kappa^2},
\end{equation}
and hence
\begin{align}
\lefteqn{\prob{\norm{\frac1m\mb{\bar X}_0\paren{\mb{\bar X}_0^T\mb\zeta}^{\circ3}-\mb{g}_E}2\ge c\theta\paren{1-\theta}\frac{\norm{\mb\zeta}4^6}{\kappa^2}}}\nonumber\\
&\le\prob{\norm{\frac1m\mb{\bar X}_0\paren{\mb{\bar X}_0^T\mb\zeta}^{\circ3}-\mb{\bar g}_E}2\ge \frac{c}2\theta\paren{1-\theta}\frac{\norm{\mb\zeta}4^6}{\kappa^2}}.
\end{align}

\noindent{\bf Independent Submatrices.} {To deal with the complicated dependence within the random circulant matrix $\mb X_0$, we break $\mb X_0$ into submatrices $\mb X_1, \dots, \mb X_{2k-1}$, each of which is (marginally) distributed as a $\paren{2k-1} \times \frac{m}{2k-1}$ i.i.d. $\mathrm{BG}(\theta)$ random matrix. Indeed, there exists a permutation $\mb \Pi$ such that 
\begin{equation}
\mb X_0 \mb\Pi=\brac{\mb X_1,\mb X_2,\cdots, \mb X_{2k-1}},
\end{equation} 
with
\begin{equation}
\mb X_i=\brac{\mb x_{i},\mb x_{i+\paren{2k-1}},\cdots,\mb x_{i+\paren{m-2k-1}}}.
\end{equation} 
We apply similar matrix breaking approach for the truncated matrix $\mb{\bar X}$. The summands within each term $\mb{\bar X}_i\paren{\mb{\bar X}_i^T\mb\zeta}^{\circ3}$ are mutually independent and hence is amenable to classical concentration results.  
\begin{align}
\frac1m\mb{\bar X}_0\paren{\mb{\bar X}_0^T\mb\zeta}^{\circ3}&=\frac1m\sum_{l=1}^m\innerprod{\mb{\bar x}_l}{\mb\zeta}^3\mb{\bar x}_l\\
&=\sum_{i=1}^{2k-1}\!\!\frac1m\!\!\paren{\sum_{j=0}^{\frac{m}{2k-1}-1}\!\!\innerprod{\mb{\bar x}_{i+\paren{2k-1}j}}{\mb\zeta}^3\mb{\bar x}_{i+\paren{2k-1}j}}\\
&=\sum_{i=1}^{2k-1}\frac1m\mb{\bar X}_i\paren{\mb{\bar X}_i^T\mb\zeta}^{\circ3}.
\end{align}
We conservatively bound the quantity of interest, $\tfrac{1}{m} \mb{\bar X}_0 \paren{ \mb{\bar X}_0^T \mb\zeta}^{\circ 3}$, by ensuring that for each $k$, $\mb{\bar X}_k  \paren{ \mb{\bar X}_k^T \mb\zeta}^{\circ 3}$ be close to its expectation.}
\begin{align*}
\lefteqn{\prob{\norm{\frac1m\mb{\bar X}_0\paren{\mb{\bar X}_0^T\mb\zeta}^{\circ3}\!\!-\mb{\bar g}_E}2\!\!\ge \frac{c}2\theta\paren{1-\theta}\frac{\norm{\mb\zeta}4^6}{ \kappa^2} }}\\
&\le \sum_{i=1}^{2k-1}\prob{\norm{\frac1m\mb{\bar X}_i\paren{\mb{\bar X}_i^T\mb\zeta}^{\circ3}\!\!-\frac{\mb{\bar g}_E}{2k-1}}2\!\!\ge \frac{c}2\frac{\theta\paren{1-\theta}\norm{\mb\zeta}4^6}{\kappa^2\paren{2k-1}}}\\
&= \sum_{i=1}^{2k-1}\prob{\norm{\frac1m\mb{\bar X}_i\paren{\mb{\bar X}_i^T\mb\zeta}^{\circ3}\!\!-\mb{\bar g}_E}2\!\!\ge \frac{c}2\frac{\theta\paren{1-\theta}\norm{\mb\zeta}4^6}{\kappa^2\paren{2k-1}}}
\end{align*}

Applying Bernstein inequality for matrix variables as in \Cref{lem:bernstein_matrix}, with $d_1 = 2k- 1$, $d_2 = 1$, we can obtain that for independent random vectors $\mb v_1, \dots, \mb v_n$ with 
\begin{equation}
\sigma^2 = \sum_{i = 1}^n \bb E[ \| \mb v_i \|_2^2 ]
\end{equation}
and ensuring that 
\begin{equation}
\|\mb v_i \|_2 \le R \qquad {a.s.}
\end{equation} 
we obtain that 
\begin{equation}
\bb P\brac{ \norm{\sum_i \mb v_i - \bb E\brac{ \cdot} }{} > t } \le 4k \exp\paren{ \frac{-t^2/2}{\sigma^2 + 2 Rt / 3} }
\end{equation}
Here, we have used that 
\begin{align}
\norm{ \sum_{i = 1}^n \bb E[ \mb v_i \mb v_i^* ] }{} &\le \trace \sum_{i = 1}^n \bb E[ \mb v_i \mb v_i^* ] \\
&=\sum_{i = 1}^n \bb E\brac{ \norm{\mb v_i }{2}^2 }.
\end{align}
and
\begin{equation} 
\mb w_i = \mb{\bar x}_i \innerprod{ \mb{\bar x}_i }{ \mb \zeta }^3.
\end{equation} 
Notice that 
\begin{align}
\norm{\mb w_i }2 &\le \norm{ \bar{\mb x}_i }2^4 \\
&\le \paren{ 4 B^2 \theta k \log{m} }^2 \\
&= 16 B^4 \theta^2 k^2 \log{m}. 
\end{align}

Let us further note that
\begin{align}
&\sum_{\substack{j_1,\\ j_2\neq j_3\neq j_4}}  \!\!\!\!\bb E\brac{ \bar{\mb x}_i(j_1)^2  \bar{\mb x}_i( j_2 )^2 \mb \zeta_{j_2}^2  \bar{\mb x}_i( j_3 )^2 \mb \zeta_{j_3}^2  \bar{\mb x}_i( j_4 )^2 \mb \zeta_{j_4}^2 }\nonumber\\
&= 3 \sum_{ j_1\ne j_2\ne j_3} \!\!\!\! \bb E\brac{ \bar{\mb x}_i(j_1)^4 \mb\zeta_{j_1}^2 \bar{\mb x}_i( j_2 )^2 \mb \zeta_{j_2}^2  \bar{\mb x}_i( j_3 )^2 \mb\zeta_{j_3}^2 }\nonumber\\
&\quad+\;\sum_{j_1 = 1}^{2k-1} \bb E\brac{ \bar{\mb x}_i(j_1)^2} \sum_{j_1\ne j_2\ne j_3\ne j_4} \!\!\!\! \bb E\brac{ \bar{\mb x}_i( j_2 )^2 \mb \zeta_{j_2}^2  \bar{\mb x}_i( j_3 )^2 \mb \zeta_{j_3}^2  \bar{\mb x}_i( j_4 )^2 \mb \zeta_{j_4}^2 }\\
&\le 2\theta k\times \theta^3\norm{\mb\zeta}{2}^6+3\times 3\theta^3\norm{\mb\zeta}2^6
\end{align}
In similar vein, we can obtain that
\begin{align}
\lefteqn{\sum_{j_1, j_2\neq j_3} \bb E\brac{ \bar{\mb x}_i(j_1)^2  \bar{\mb x}_i( j_2 )^2 \mb \zeta_{j_2}^2  \bar{\mb x}_i( j_3 )^4 \mb \zeta_{j_3}^4 }}\nonumber\\
&= \sum_{j_1}\bb E\brac{ \bar{\mb x}_i(j_1)^2} \sum_{j_2\neq j_3\neq j_1}\!\!\bb E\brac{ \bar{\mb x}_i( j_2 )^2 \mb \zeta_{j_2}^2  \bar{\mb x}_i( j_3 )^4 \mb \zeta_{j_3}^4 }\nonumber\\
&\quad+\sum_{j_1\neq j_2} \bb E\brac{ \bar{\mb x}_i( j_1 )^4 \mb \zeta_{j_1}^2  \bar{\mb x}_i( j_2 )^4 \mb \zeta_{j_2}^4}\nonumber\\
&\quad+\sum_{j_1\neq j_2} \bb E\brac{ \bar{\mb x}_i( j_1 )^2 \mb \zeta_{j_1}^2  \bar{\mb x}_i( j_2 )^6 \mb \zeta_{j_2}^4}\\
&\le 2\theta k\times 3\theta^2\norm{\mb\zeta}2^2\norm{\mb\zeta}4^4+\paren{9+15}\theta^2\norm{\mb\zeta}2^2\norm{\mb\zeta}4^4
\end{align}

and
\begin{align}
\lefteqn{\sum_{j_1,j_2} \bb E\brac{ \bar{\mb x}_i(j_1)^2  \bar{\mb x}_i( j_2 )^6 \mb \zeta_{j_2}^6 }}\nonumber\\
&= \sum_{j_1}\bb E\brac{ \bar{\mb x}_i(j_1)^2} \sum_{j_2\neq j_1}\bb E\brac{ \bar{\mb x}_i( j_2 )^6 \mb \zeta_{j_2}^6 }+\sum_{j_1} \bb E\brac{ \bar{\mb x}_i(j_1)^8 \mb \zeta_{j_1}^6}\\
&\le 2\theta k\times 15\theta\norm{\mb\zeta}6^6+105\theta\norm{\mb\zeta}6^6
\end{align}

Now we calculate 
\begin{align}
\lefteqn{\bb E\brac{ \| \mb w_i \|_2^2 } = \bb E\brac{ \| \bar{\mb x}_i \|_2^2 \innerprod{ \bar{\mb x}_i }{\mb \zeta }^6 }} \\
&= \bb E\brac{ \sum_{j_1, \dots, j_7 } \bar{\mb x}_i(j_1)^2 \prod_{\ell = 2}^7 \bar{\mb x}_i( j_\ell ) \mb \zeta_{j_\ell} } \\
&= 15 \!\!\!\!\!\!\sum_{\substack{j_1,\\ j_2\neq j_3\neq j_4}}\!\!\!\!\!\! \bb E\brac{ \bar{\mb x}_i(j_1)^2  \bar{\mb x}_i( j_2 )^2 \mb \zeta_{j_2}^2  \bar{\mb x}_i( j_3 )^2 \mb \zeta_{j_3}^2  \bar{\mb x}_i( j_4 )^2 \mb \zeta_{j_4}^2 }\nonumber\\
&\quad+ 15\sum_{j_1, j_2\neq j_3} \bb E\brac{ \bar{\mb x}_i(j_1)^2  \bar{\mb x}_i( j_2 )^2 \mb \zeta_{j_2}^2  \bar{\mb x}_i( j_3 )^4 \mb \zeta_{j_3}^4 }\nonumber\\
&\quad+ \sum_{j_1,j_2} \bb E\brac{ \bar{\mb x}_i(j_1)^2  \bar{\mb x}_i( j_2 )^6 \mb \zeta_{j_2}^6 } \\
&\le 15 \theta^3 \norm{\mb \zeta}{2}^6\paren{2\theta k +9}\nonumber\\
&\quad+ 15 \theta^2 \norm{\mb \zeta}{4}^4\paren{ 6 +24  }\nonumber\\
&\quad+  \theta \norm{ \mb \zeta}{6}^6\paren{30\theta k + 105 }  \\
&\le 150\theta^2 k+600\theta
\end{align}
whence for $\theta > 1/k$,
\begin{equation}
\bb E\brac{ \| \mb w_i \|_2^2 } \;\le\; C \theta^2 k,
\end{equation}
and hence
\begin{equation}
\sigma^2 \;\le\; C' \theta^2 m.
\end{equation}
 Matrix Bernstein gives that 
\begin{align}
\bb P\brac{ \norm{\mb{\bar X}_i (\mb{\bar X}_i^T \mb \zeta)^{\circ 3} - \bb E\brac{ \cdot } }{2} \ge t  } \le 4k \exp\paren{ \frac{- t^2 / 2 }{ C \theta^2 m + C' B^4 \theta^2k^2\log^2{k} t } }. 
\end{align}
Setting $t = \frac{c}4 \frac{m \theta (1- \theta ) \| \mb \zeta \|_4^6}{\kappa^2\paren{2k-1}}$, we obtain that 
\begin{align}
&\bb P \brac{ \norm{\frac{1}{m} \mb{\bar X}_i (\mb{\bar X}_i^T \mb \zeta)^{\circ 3} - \bb E\brac{ \cdot } }{2} \ge \frac{c}4 \frac{\theta (1- \theta ) \| \mb \zeta \|_4^6}{\kappa^2\paren{2k-1}} }\nonumber\\
&\le 4 k \exp\paren{ -  \frac{ c'' m \paren{1 - \theta }^2 \norm{ \mb \zeta }4^{12} }{\kappa^4 k^2+  \theta \paren{1-\theta}B^4\kappa^2 k^3 \norm{\mb\zeta}4^{6}} }
\end{align}

\noindent{\bf $\eps$-Net Covering} To obtain a probability bound for all $\mb q\in\bb S^{k-1}$, we choose a set of $\mb\zeta_n=\mb A^T\mb q_n$ with $n=1,\cdots,N$. 
Suppose for any $\mb q\in\bb S^{k-1}$, there exists $\mb q_n$ such that $\norm{\mb q-\mb q_n}2\le \eps$, then
\begin{equation}
\norm{\frac1m\mb{\bar X}_i\paren{\mb{\bar X}_i^T\mb\zeta}^{\circ3}-\frac1m\mb{\bar X}_i\paren{\mb{\bar X}_i^T\mb\zeta_n}^{\circ3}}2\le L\norm{\mb q-\mb q_n}2.
\end{equation}
For entry wise bounded $\mb{\bar X}_i\in\R^{\paren{2k-1}\times\frac{m}{2k-1}}$, we have
\begin{equation}
\norm{\mb{\bar X}_i}2\le \sqrt{2\theta m}B,\quad\norm{\mb{\bar X}_i\mb e_j}2\le \sqrt{4\theta k}B,
\end{equation}
then the Lipschitz constant $L$ can be bounded as 
\begin{align}
L&\le \frac1m\norm{\mb{\bar X}_i}2\norm{\diag\paren{\mb{\bar X}^T_i\mb \zeta}^{\circ2}}2\norm{\mb{\bar X}_i^T\mb A^T}2\\
&\le 8\theta^2kB^4.
\end{align}

With triangle inequality, we have
\begin{align}
&\norm{\frac1m\mb{\bar X}_i\paren{\mb{\bar X}_i^T\mb\zeta}^{\circ3}-\bb E\brac{\frac1m\mb{\bar X}_i\paren{\mb{\bar X}_i^T\mb\zeta}^{\circ3}}}2\nonumber\\
&\le\norm{\bb E\brac{\frac1m\mb{\bar X}_i\paren{\mb{\bar X}_i^T\mb\zeta}^{\circ3}}-\bb E\brac{\frac1m\mb{\bar X}_i\paren{\mb{\bar X}_i^T\mb\zeta_n}^{\circ3}}}2\nonumber\\
&\quad+\norm{\frac1m\mb{\bar X}_i\paren{\mb{\bar X}_i^T\mb\zeta_n}^{\circ3}-\bb E\brac{\frac1m\mb{\bar X}_i\paren{\mb{\bar X}_i^T\mb\zeta_n}^{\circ3}}}2\nonumber\\
&\quad+\norm{\frac1m\mb{\bar X}_i\paren{\mb{\bar X}_i^T\mb\zeta}^{\circ3}-\frac1m\mb{\bar X}_i\paren{\mb{\bar X}_i^T\mb\zeta_n}^{\circ3}}2\\
&\le \norm{\frac1m\mb{\bar X}_i\paren{\mb{\bar X}_i^T\mb\zeta_n}^{\circ3}-\bb E\brac{\frac1m\mb{\bar X}_i\paren{\mb{\bar X}_i^T\mb\zeta_n}^{\circ3}}}2+2L\eps.
\end{align}

Hence we need to choose the $\eps$-net to cover the sphere of $\mb q$ with
\begin{equation}
\eps=\frac{c}4\frac{\theta\paren{1-\theta}}{\kappa^2 \paren{2k-1}L}
\min_{\mb q\in\bb S^{k-1}}\norm{ \mb\zeta}4^6, 
\end{equation}
plug in $L\le4\theta^2 kB^4$ and number of sample $N$ suffice
\begin{align}
N&\le\paren{\frac3{\eps}}^k\\
&\le\exp\paren{ k\ln\paren{ \frac3{\eps} } }\\
&\le\exp\brac{k\ln\paren{C\frac{\theta^2\kappa^2k^4B^4}{\theta\paren{1-\theta}}}}
\end{align}

For $n=1,\cdots,N$, denote
\begin{align}
P_i\paren{\mb q_n}\!=\!\prob{\norm{\frac{\bar{\mb X}_i\paren{\bar{\mb X}_i^T\mb\zeta_n}^{\circ3}\!\!\!\!}{m}\!-\!\bb E\!\brac{\cdot}}2\!\!\!\ge\!\frac{c\theta(1\!-\theta)\!\norm{\mb\zeta_n}4^6}{4\kappa^2\paren{2k-1}}},
\end{align}
then together with union bound over all $\mb q_n$ , we obtain that,

\begin{align}
&\prob{\sup_{\mb q\in\hat{\mc R}_{2C_\star}}\!\!\frac{\norm{\frac1m\bar{\mb X}_i\paren{\bar{\mb X}_i^T\mb\zeta}^{\circ3}-\bb E\brac{\cdot}}2}{\norm{\mb\zeta}4^6}\ge \frac{c}2\frac{\theta\paren{1-\theta}}{\kappa^2\paren{2k-1}}}\nonumber\\
&\le \sum_{\mb q_n\in\hat{\mc R}_{2C_\star}}P_i\paren{\mb q_n}\\
&\le N\max_{\mb q_n\in\hat{\mc R}_{2C_\star}}P_i\paren{\mb q_n}\\
&\le 4k\!\!\sup_{\mb q\in\hat{\mc R}_{2C_\star}} \!\!\!\!\exp\paren{- \frac{ c m \paren{1 - \theta }^2 \norm{ \mb \zeta }4^{12} }{ \kappa^4 k^2+ \theta \paren{1-\theta}B^4\kappa^2 k^3 \norm{\mb\zeta}4^{6}}} \exp\paren{k\ln\paren{\frac3{\eps}}}.
\end{align}

Hence 
\begin{align}
\lefteqn{\prob{\sup_{\mb q\in\hat{\mc R}_{2C_\star}}\frac{\norm{\frac1m\bar{\mb X}_0\paren{\bar{\mb X}_0^T\mb\zeta}^{\circ3}\!\!-\bb E\brac{\cdot}}2}{\norm{\mb\zeta}4^6}\ge \frac{c}2\frac{\theta\paren{1-\theta}}{\kappa^2}}}\nonumber\\
&\le\sum_{i}\prob{\sup_{\mb q\in\hat{\mc R}_{2C_\star}}\!\!\!\!\frac{\norm{\frac1m\bar{\mb X}_i\paren{\bar{\mb X}_i^T\mb\zeta}^{\circ3}\!\!-\bb E\brac{\cdot}}2}{\norm{\mb\zeta}4^6}\ge \frac{c\theta\paren{1-\theta}}{2\kappa^2\paren{2k-1}}}\\
&\le\paren{2k-1}\max_{i}\prob{\sup_{\mb q\in\hat{\mc R}_{2C_\star}}\!\!\!\!\frac{\norm{\frac1m\bar{\mb X}_i\paren{\bar{\mb X}_i^T\mb\zeta}^{\circ3}\!\!-\bb E\brac{\cdot}}2}{\norm{\mb\zeta}4^6}\ge \frac{c\theta\paren{1-\theta}}{2\kappa^2\paren{2k-1}}}\\
&\le 8k^2\!\!\sup_{\mb q\in\hat{\mc R}_{2C_\star}} \!\!\!\!\exp\paren{- \frac{ c m \paren{1 - \theta }^2 \norm{ \mb \zeta }4^{12} }{ \kappa^4 k^2+ \theta \paren{1-\theta}B^4\kappa^2 k^3 \norm{\mb\zeta}4^{6}}}\exp\paren{k\ln\paren{\frac3{\eps}}}, 
\end{align}

which is bounded by $\exp\paren{-k}$ as long as
\begin{align}
m&\ge C\frac{\min\set{\paren{2C_{\star}\mu}^{-2},\kappa^2k^2}}{\paren{1-\theta}^2}\kappa^2k^4\log^3\paren{\kappa k}\\
&\ge C'k\log\paren{\frac{\theta\kappa^2k^2B^4}{\paren{1-\theta}\norm{\mb\zeta}4^6}}\max\set{\frac{\kappa^4 k^2}{\paren{1-\theta}^2\norm{\mb\zeta}4^{12}},\frac{\theta B^4\kappa^2k^3}{\paren{1-\theta}\norm{\mb\zeta}4^{6}}}.
\end{align}
To sum up, we obtain that for all $\mb q\in\hat{\mc R}_{2C_{\star}}$, inequality
\begin{equation}
\norm{\frac1m\mb X_0\paren{\mb X_0^T\mb A^T\mb q}^{\circ3}-\bb E\brac{\cdot}}2\le c\theta\paren{1-\theta}\frac{\norm{\mb A^T\mb q}4^6}{\kappa^2}
\end{equation}
holds with probability no smaller than $1-c_1\exp\paren{-k}-c_2k^{-4}-c_3\exp\paren{-\theta k}$.
\end{proof}

\section{Concentration for Hessian (Lemma \ref{lem:hess_scale})}
\label{sec:hess_scale}
\begin{lemma}
Suppose $\mb x_0\simiid\mathrm{BG}\paren{\theta}$. There exists positive constant $C$ that whenever
\begin{align}
m\ge C \theta\frac{\min\set{\paren{2C_\star\mu\kappa^2}^{-4/3},k^2}}{\paren{1-\theta}^2\sigma^2_{\min}}\kappa^6 k^4\log^3\paren{\frac{\kappa k}{\paren{1-\theta}\sigma_{\min}}}
\end{align}
and $\theta \ge \log{k}/k$, then with probability no smaller than $1-c_1\exp\paren{-k}-c_2k^{-4}-48k^{-7}-48m^{-5}-24k\exp\paren{-\tfrac1{144}\min\set{k,3\sqrt{\theta m}}}$,
\begin{align}
\norm{\Hess[\psi]\paren{\mb q}-\frac{3\paren{1-\theta}}{\theta m^2}\Hess[\varphi]\paren{\mb q}}2\le c\frac{1-\theta}{\theta m^2}\norm{\mb A^T\mb q}4^4,
\end{align}
holds for all $\mb q\in\hat{\mc R}_{2C_{\star}}$ with positive constant $c\le0.048\le3\paren{1-6c_{\star}-36c_{\star}^2-24c_{\star}^3}$.
\end{lemma}

\begin{proof}
Denote $\mb\eta=\mb Y^T\paren{\mb Y\mb Y^T}^{-1/2}\mb q$ and $\mb{\bar\eta}=\mb Y^T\paren{\theta m\mb A_0\mb A_0^T}^{-1/2}\mb q = \paren{\theta m}^{-1/2}\mb X_0^T\mb\zeta$, and 
\begin{align}
&\mb W = \paren{\frac1{\theta m}\mb Y\mb Y^T}^{-1/2}-\paren{\mb A_0\mb A_0^T}^{-1/2},\\
&\widehat{\mb Y} = \paren{\mb Y\mb Y^T}^{-1/2}\mb Y.
\end{align}
Then we have
\begin{align}
&\norm{\Hess\brac{\psi}\paren{\mb q}-\frac{3\paren{1-\theta}}{\theta m^2}\Hess\brac{\varphi}\paren{\mb q}}2\nonumber\\
&=\Big\|\mb P_{\mb q^{\perp}}\brac{\frac3m\widehat{\mb Y}\diag\paren{\mb\eta^{\circ2}}\widehat{\mb Y}^T-\innerprod{\mb q}{\nabla\psi\paren{\mb q}}\mb I}\mb P_{\mb q^{\perp}}-\frac{3\paren{1-\theta}}{\theta m^2}\mb P_{\mb q^{\perp}}\brac{3\mb A \diag \paren{\mb\zeta^{\circ2}}\mb A^T-\norm{\mb\zeta}4^4\mb I}\mb P_{\mb q^{\perp}}\Big\|_2\\
&\le\Big\|\mb P_{\mb q^{\perp}}\brac{\frac3m\widehat{\mb Y}\diag\paren{\mb\eta^{\circ2}}\widehat{\mb Y}^T}\mb P_{\mb q^{\perp}}-\mb P_{\mb q^{\perp}}\brac{\frac{9\paren{1-\theta}}{\theta m^2}\mb A\diag\paren{\mb\zeta^{\circ2}} \mb A^T-\frac3{m^2}\mb I}\mb P_{\mb q^{\perp}}\Big\|_{2}\nonumber\\
&\;+\norm{\brac{\innerprod{\mb q}{\nabla\psi\paren{\mb q}}-\frac{3\paren{1-\theta}}{\theta m^2}\norm{\mb\zeta}4^4-\frac3{m^2}}\mb P_{\mb q^{\perp}}}2\\
&\le\underbrace{\frac3{\theta m^2}\norm{\mb W\mb Y\diag\paren{\mb\eta^{\circ2}}\mb Y^T\paren{\frac1{\theta m}\mb Y\mb Y^T}^{-1/2}}2}_{\Delta^H_1}\nonumber\\
&\;+\underbrace{\frac3{\theta m^2}\norm{\mb A\mb X_0\diag\paren{\mb\eta^{\circ2}}\mb Y^T\mb W}2}_{\Delta^H_2}\nonumber\\
&\;+\underbrace{\frac3{ \theta m^2}\norm{\mb A\mb X_0\diag\paren{\mb\eta^{\circ2}-\mb{\bar\eta}^{\circ2}}\mb X_0^T\mb A^T}2}_{\Delta^H_3}\nonumber\\
&\;+\underbrace{\frac3{\theta m^2}\Big\|\mb P_{\mb q^{\perp}}\brac{\mb A\mb X_0\diag\paren{\mb{\bar\eta}^{\circ2}}\mb X_0^T\mb A^T-3\paren{1-\theta}\mb A\diag\paren{\mb\zeta^{\circ2}}\mb A^T-\theta\mb I}\mb P_{\mb q^{\perp}}\Big\|_2}_{\Delta^H_4}\nonumber\\
&\;+\underbrace{\norm{\brac{\innerprod{\mb q}{\nabla\psi\paren{\mb q}}-\frac{3\paren{1-\theta}}{\theta m^2}\norm{\mb\zeta}4^4-\frac3{m^2}}\mb P_{\mb q^{\perp}}}2}_{\Delta^H_5}
\end{align}
In the rest of the proof, we prove that
\begin{align}
\Delta^H_i\le \frac{c}{9}\frac{1-\theta}{\theta m^2}\norm{\mb\zeta}4^4,\quad i=1,2,3.
\end{align}
and
\begin{align}
\Delta^H_i\le \frac{c}{3}\frac{1-\theta}{\theta m^2}\norm{\mb\zeta}4^4,\quad i=4,5.
\end{align}

First, let us note that
\begin{align}
\lefteqn{C\paren{1-\theta}^{-2}\sigma_{\min}^{-2} \kappa^6k^5\log^3\paren{\frac{\kappa k}{\paren{1-\theta}\sigma_{\min}}}}\\
&\le C\paren{\frac{\kappa k}{\paren{1-\theta}\sigma_{\min}}}^{6}\log^3\paren{\frac{\kappa k}{\paren{1-\theta}\sigma_{\min}}}\\
&\le C\paren{\frac{\kappa k}{\paren{1-\theta}\sigma_{\min}}}^{9}
\end{align}
or
\begin{align}
&\frac{\log^3\paren{C\paren{1-\theta}^{-2}\sigma_{\min}^{-2}\kappa^6k^5\log^3\paren{\frac{\kappa k}{\paren{1-\theta}\sigma_{\min}}}}}{C\log^3\paren{\frac{\kappa k}{\paren{1-\theta}\sigma_{\min}}}}\nonumber\\
&\le\paren{\frac{\log{C}+9\log\paren{\frac{\kappa k}{\paren{1-\theta}\sigma_{\min}}}}{C^{1/3}\log\paren{\frac{\kappa k}{\paren{1-\theta}\sigma_{\min}}}}}^3\\
&\le\paren{\frac{\log{C}}{C^{1/3}\log\paren{\frac{\kappa k}{\paren{1-\theta}\sigma_{\min}}}}+\frac{9}{C^{1/3}}}^3\\
&\le\paren{\frac1{C^{1/6}}+\frac12\frac1{C^{1/6}}}^3\qquad\paren{C\ge 10^8}\\
&\le\frac4{C^{1/2}}.
\end{align}

Since 
\begin{align}
m\ge C\frac{\min\set{\paren{2C_\star\mu\kappa^2}^{-4/3}\!\!,k^2}}{\paren{1-\theta}^2\sigma^2_{\min}}\kappa^6 k^4\log^3\paren{\frac{\kappa k}{\paren{1-\theta}\sigma_{\min}}},
\end{align}
as the ratio $\log^3{m}/m$ decreases with increasing $m$, then
\begin{align}
\frac{\log^3{m}}{m}
&\le \frac{ \log^3\paren{C\frac{\kappa^6k^5}{\paren{1-\theta}^2\sigma^2_{\min}}\log^3\paren{\frac{\kappa k}{\sigma_{\min}\paren{1-\theta}}}} }
{ C\log^3\paren{\frac{\kappa k}{\sigma_{\min}\paren{1-\theta}}} }\frac{\paren{1-\theta}^2\sigma^2_{\min}}{\min\set{\paren{2C_\star\mu\kappa^2}^{-2/3},k}\kappa^6 k^4}\\
&\le\frac4{C^{1/2}}\frac{\paren{1-\theta}^2\sigma^2_{\min}}{\min\set{\paren{2C_\star\mu\kappa^2}^{-2/3},k}\kappa^6 k^4}
\end{align}

According to \Cref{lem:preconditioning}, following inequality obtains
\begin{align}
\norm{\frac1{\theta m}\mb X_0\mb X_0^T-\mb I}2&\le
\delta\\
&\le10\sqrt{k\log{m}/m}\\
&\le\frac{20\paren{1-\theta}\sigma_{\min} \max\set{\paren{2C_{\star}\mu\kappa^2}^{2/3},k^{-1}}}{C^{1/4}\kappa^3 k^{3/2}\log{m}}\\
&\le{\frac{20\sigma_{\min}}{C^{1/4}\kappa^3}\cdot\frac{\paren{1-\theta}\norm{\mb A^T\mb q}4^4}{k^{3/2}\log{m}},\quad\forall\mb q\in\hat{\mc R}_{2C_{\star}}}
\end{align}
with probability no smaller than $1-\eps_0$ with $\eps_0=2\exp\paren{-\theta k} + 24k\exp\paren{-\tfrac1{144}\min\set{k,3\sqrt{\theta m}}} + 48k^{-7} + 48m^{-5}$.

We have $4\kappa^3\delta/\sigma_{\min}\le1/2$ whenever
\begin{align}
C\ge \paren{\frac{160\paren{1-\theta}}{k^{3/2}\log{m}}}^4
\end{align} 
whence $\delta \le 1/\paren{8 \kappa^2}$, and \Cref{lem:precond_neghalf_2} implies that 
\begin{align}
\norm{\paren{\frac1{\theta m}\mb Y\mb Y^T}^{-1/2}\mb A_0-\paren{\mb A_0\mb A_0^T}^{-1/2}\mb A_0}2
&\le 4\kappa^3\delta/\sigma_{\min}\\
&\le \frac{80\paren{1-\theta}\norm{\mb A^T\mb q}4^4}{C^{1/4}k^{3/2}\log{m}},\quad\forall\mb q\in\hat{\mc R}_{2C_{\star}}.
\end{align}
Moreover,
\begin{align}
\norm{\mb X_0}2&\le\paren{\theta m}^{1/2}\sqrt{1+\delta}\\
&\le\paren{\theta m}^{1/2}\paren{1+\delta/2}\\
&\le\frac{17}{16}\paren{\theta m}^{1/2}.
\end{align}
Finally, \Cref{lem:truncate} implies that with probability no smaller than $1-\eps_B$, we have
\begin{align}
\norm{\mb x_0}{\infty}\le\sqrt2\log^{1/2}\paren{\frac{2\theta m}{\eps_B}}.
\end{align}

\noindent{\bf Upper Bound for $\Delta^H_1$ and $\Delta^H_2$.}
With probability no smaller than $1-\eps_0-\eps_B$, the norms of $\mb\eta$ are upper bounded as in \Cref{lem:eta_norms},
\begin{align}
\Delta^H_1
&\le\frac3{\theta m^2}\norm{\paren{\frac1{\theta m}\mb Y\mb Y^T}^{-1/2}\mb A_0-\mb A}2\norm{\mb X_0}2^2\norm{\mb\eta}{\infty}^2\norm{\mb A_0^T\paren{\frac1{\theta m}\mb Y\mb Y^T}^{-1/2}}2\\
&\le\frac{3}{\theta m^2}\cdot\frac{4\kappa^3\delta}{\sigma_{\min}}
\cdot\paren{1+\delta/2}^2\theta m\cdot \paren{1+\frac{4\kappa^3\delta}{\sigma_{\min}}}^3\frac{4k}{\theta m}\log\paren{2\theta m/\eps_B}
\\
&\le\frac{3660}{C^{1/4}}\frac{1-\theta}{\theta m^2}\norm{\mb\zeta}4^4\cdot\frac{\log\paren{2\theta m/\eps_B}}{k^{1/2}\log{m}}.
\end{align}
A similar result holds for 
\begin{align}
\Delta^H_2
&\le\frac3{\theta m^2}\norm{\mb X_0}2^2\norm{\diag\paren{\mb\eta^{\circ2}}}2\norm{\paren{\frac1{\theta m}\mb Y\mb Y^T}^{-1/2}\mb A_0-\mb A}2\\
&\le\frac{2440}{C^{1/4}}\frac{1-\theta}{\theta m^2}\norm{\mb\zeta}4^4\cdot\frac{\log\paren{2\theta m/\eps_B}}{k^{1/2}\log{m}}.
\end{align}
To make $\Delta^H_1\le\frac{c}{9}\frac{1-\theta}{\theta m^2}\norm{\mb\zeta}4^4$ and $\Delta^H_2\le\frac{c}{9}\frac{1-\theta}{\theta m^2}\norm{\mb\zeta}4^4$, we require
\begin{align}
C\ge\paren{9\times3660 c^{-1}\frac{\log\paren{2\theta m/\eps_B}}{k^{1/2}\log{m}}}^4.
\end{align}
The right hand side is bounded by an absolute constant for all $m$. \\

\noindent{\bf Upper Bound for $\Delta^H_3$.}
With probability no smaller than $1-\eps_0-\eps_B$, the difference between $\bar{\mb \eta}^{\circ 2}$ and $\mb\eta^{\circ 2}$ is upper bounded as in \Cref{lem:eta_norms},
\begin{align}
\lefteqn{\norm{\mb\eta^{\circ2}-\mb{\bar\eta}^{\circ2}}{\infty}}\nonumber\\
&\le\norm{\mb\eta-\mb{\bar\eta}}{\infty}\norm{\mb\eta+\mb{\bar\eta}}
{\infty}\\
&\le\frac{4\kappa^3\delta}{\sigma_{\min}}\paren{2+\frac{4\kappa^3\delta}{\sigma_{\min}}}\frac{2k}{\theta m}\log\paren{2\theta m/\eps_B}\\
&\le\frac{5k}{\theta m}\log\paren{2\theta m/\eps_B}\cdot\frac{4\kappa^3\delta}{\sigma_{\min}}.
\end{align}
Therefore
\begin{align}
\Delta^H_3&=\frac3{\theta m^2}\norm{\mb A\mb X_0\diag\paren{\mb\eta^{\circ2}-\mb{\bar\eta}^{\circ2}}\mb X_0^T\mb A^T}2\\
&\le\frac{}{\theta m^2}\norm{\mb A}2^2\norm{\mb X_0}2^2 \norm{\diag\paren{\mb\eta^{\circ2}-\mb{\bar\eta}^{\circ2}}}2\\
&\le\frac{15k}{\theta m^2}\paren{1+\delta/2}^2\log\paren{2\theta m/\eps_B}\cdot\frac{4\kappa^3\delta}{\sigma_{\min}}\\
&\le\frac{1400\paren{1-\theta}\log\paren{2\theta m/\eps_B}}{C^{1/4}\theta k^{1/2}m^2\log{m}}\norm{\mb\zeta}4^4.
\end{align}
Again, $\Delta^H_3$ is bounded by $\frac{c}9\frac{1-\theta}{\theta m^2}\norm{\mb\zeta}4^4$ whenever
\begin{align}
C\ge\paren{9\times 1400c^{-1}\frac{\log\paren{2\theta m/\eps_B}}{k^{1/2}\log{m}}}^4
\end{align}

\noindent{\bf Upper Bound for $\Delta^H_4$.} 
Recall that 
\begin{align}
\mb{\bar\eta}=\mb Y^T\paren{\theta m\mb A_0\mb A_0^T}^{-1/2}\mb q, 
\end{align}
then
\begin{align}
&\bb E\brac{\mb X_0\diag\paren{\mb{\bar\eta}^{\circ2}}\mb X_0^T}\nonumber\\
&= \bb E\brac{\frac1{\theta m}\mb X_0\diag\paren{\mb X_0^T\mb A^T\mb q}^{\circ2}\mb X_0^T}\\
&= 3\paren{1-\theta}\diag\paren{\mb A^T\mb q}^{\circ2}+2\theta\mb A^T\mb q\mb q^T\mb A+\theta\norm{\mb A^T\mb q}2^2\mb I,
\end{align}
once including the projection $\mb P_{\mb q^{\perp}}$, we have
\begin{align}
&\mb P_{\mb q^{\perp}}\bb E\brac{\mb A\mb X_0\diag\paren{\mb{\bar\eta}^{\circ2}}\mb X_0^T\mb A^T}\mb P_{\mb q^{\perp}}\\
&=\mb P_{\mb q^{\perp}}\brac{3\paren{1-\theta}\mb A\diag\paren{\mb\zeta^{\circ2}}\mb A^T+\theta\mb I}\mb P_{\mb q^{\perp}}.\nonumber
\end{align}
Therefore
\begin{align}
\Delta^H_4&=\frac3{\theta m^2}\Big\|\mb P_{\mb q^{\perp}}\brac{\mb A\mb X_0\diag\paren{\mb{\bar\eta}^{\circ2}}\mb X_0^T\mb A^T}\mb P_{\mb q^{\perp}}-\mb P_{\mb q^{\perp}}\brac{3\paren{1-\theta}\mb A\diag\paren{\mb\zeta^{\circ2}}\mb A^T+\theta\mb I}\mb P_{\mb q^{\perp}}\Big\|_2\\
&\le\frac3{\theta^2m^2}\norm{\frac1m\mb X_0\diag\paren{\mb X_0^T\mb\zeta}^{\circ2}\mb X_0^T-\bb E\brac{\cdot}}2
\end{align}

Under the assumption for sample size that $m\ge C\paren{1-\theta}^{-2}\kappa^4\min\set{\paren{2C_\star\mu}^{-2/3},k}k^3\log^5\paren{\kappa k}$, applying \Cref{lem:hess_sub}, we have
\begin{align}
\norm{\frac1m\mb X_0\diag\paren{\mb X_0^T\mb\zeta}^{\circ2}\mb X_0^T-\bb E\brac{\cdot}}2\le \frac{c}{9}\theta\paren{1-\theta}\norm{\mb\zeta}4^4.
\end{align}
simultaneously at every $\mb q \in \hat{\mc R}_{2C_\star}$ with probability no smaller than $1-c_1\exp\paren{-k}-c_2k^{-4}$. \\

\noindent{\bf Upper Bound for $\Delta^H_5$.}
Note that this term is essentially the difference between 
\begin{align}
\Delta^H_5&=\norm{\brac{\innerprod{\mb q}{\nabla\psi\paren{\mb q}}\!-\!\frac{3\paren{1-\theta}}{\theta m^2}\norm{\mb\zeta}4^4\!-\!\frac3{m^2}}\mb P_{\mb q^{\perp}}}2\\
&\le\abs{\innerprod{\mb q}{\nabla\psi\paren{\mb q}}-\frac{3(1\!-\theta)}{\theta m^2}\norm{\mb\zeta}4^4-\frac3{m^2}}\\
&\le\frac1{\theta^2m^2}\abs{\frac1m\norm{\mb X_0^T\mb\zeta}4^4-3\theta\paren{1-\theta}\norm{\mb\zeta}4^4-3\theta^2}+\abs{\innerprod{\mb q}{\nabla\psi\paren{\mb q}}-\frac1{\theta^2m^2}\norm{\mb X_0^T\mb\zeta}4^4}\\
&\le\frac1{\theta^2m^2}\!\norm{\frac{\mb A\mb X_0\!\paren{\mb X_0^T\mb \zeta}^{\circ3}\!\!\!\!}m-\!3\theta(1\!-\theta)\mb A^T\mb\zeta^{\circ3}\!\!-\!3\theta^2\mb q}2+\frac1m\abs{\norm{\mb\eta}4^4-\norm{\mb{\bar\eta}}4^4}
\end{align}
Recall that
\begin{align}
\bb E\brac{\frac1m\mb A\mb X_0\paren{\mb X_0^T\mb A^T\mb q}^{\circ3}}
&=\bb E\brac{\mb A\mb x_i\paren{\mb x_i^T\mb A^T\mb q}^{3}}\\
&=3\theta\paren{1-\theta}\mb A\mb\zeta^{\circ3}+3\theta^2\mb q,
\end{align}
With similar argument as in \Cref{lem:grad_scale}, we can show that this term can be bounded by $\frac{c}{6}\frac{1-\theta}{\theta m^2}\norm{\mb\eta}4^4$ whenever
\begin{align}
m\ge C'\frac{\min\set{\paren{\mu\kappa^2}^{-4/3}\!\!,k^2}}{\paren{1-\theta}^2 \sigma^2_{\min}}\kappa^6k^4\log^3\!\paren{\frac{\kappa k}{(1\!-\theta)\sigma_{\min}}}.
\end{align}
Moreover, with probability $1-\eps_0-\eps_B$
\begin{align}
\lefteqn{\frac1m\abs{\norm{\mb\eta}4^4-\norm{\mb{\bar\eta}}4^4}}\nonumber\\
&\le\frac1m\abs{\innerprod{\mb\eta-\mb{\bar\eta}\mb}{4\mb\eta^{\circ3}}}\\
&\le\frac4m\norm{\mb\eta-\mb{\bar\eta}}2\norm{\mb\eta}6^3\\
&\le\frac{16\kappa^3\delta}{\sigma_{\min}m}\paren{1+\delta/2}\paren{1+\frac{4\kappa^3\delta}{\sigma_{\min}}}^2\frac{4 k}{\theta m}\log\paren{2\theta m/\eps_B}\\
&\le\frac{153k}{\theta m^2}\log\paren{2\theta m/\eps_B}\cdot\frac{\kappa^3\delta}{\sigma_{\min}}\\
&\le\frac{3060}{C^{1/4}}\frac{\paren{1-\theta}}{\theta m^2}\norm{\mb\zeta}4^4\cdot\frac{\log\paren{2\theta m/\eps_B}}{k^{1/2}\log{m}},
\end{align}
which is bounded by $\frac{c}6\frac{1-\theta}{\theta m^2}\norm{\mb\zeta}4^4$ whenever
\begin{align}
C\ge\paren{6\times 3060 c^{-1}\frac{\paren{1-\theta}\log\paren{2\theta m/\eps_B}}{k^{1/2}\log{m}}}^4.
\end{align}
The right hand side is bounded by an absolute constant for all $m$.

Adding up failure probabilities, we have that with probability larger than $1-c_2\exp\paren{-k}-c_2k^{-4}-\eps_0$, 
\begin{align}
\norm{\Hess\brac{\psi}\paren{\mb q}-\frac{3\paren{1-\theta}}{\theta m^2}\Hess\brac{\varphi}\paren{\mb q}}2\le c\frac{1-\theta}{\theta m^2}\norm{\mb A^T\mb q}4^4
\end{align}
holds as desired for all $\mb q\in\hat{\mc R}_{2C_\star}$, where $\eps_0=2\exp\paren{-\theta k} + 24k\exp\paren{-\tfrac1{144}\min\set{k,3\sqrt{\theta m}}} + 48k^{-7} + 48m^{-5}$. 
\end{proof}

\subsection{Proof of Lemma \ref{lem:hess_sub}}
\begin{lemma}
\label{lem:hess_sub}Suppose $\mb x_0\simiid\mathrm{BG}\paren{\theta}$. There exist constants $C>0$ that whenever
\begin{align}
m\ge C\frac{\min\set{\paren{2C_\star\mu\kappa^2}^{-4/3},k^2}}{\paren{1-\theta}^2}k^4\log^3\paren{\kappa k},
\end{align}
and $\theta k>1$, then with probability no smaller than $1-c_1\exp\paren{-k}-c_2k^{-4}$,
\begin{align}
\norm{\frac1m\mb X_0\diag\paren{\mb X_0^T\mb A^T\mb q}^{\circ2}\mb X_0^T-\bb E\brac{\cdot}}2\le c\theta\paren{1-\theta}\norm{\mb A^T\mb q}4^4,
\end{align}
holds for all $\mb q\in\hat{\mc R}_{2C_{\star}}$ with positive constant $c\le0.005\le\paren{1-6c_{\star}-36c_{\star}^2-24c_{\star}^3}/3$.
\end{lemma}

\begin{proof}The proof strategy for the finite sample concentration of the Hessian is similar to that of the gradient as presented in \Cref{lem:grad_sub}. For simplicity, we will only demonstrate some key steps here, please refer to \Cref{lem:grad_sub} for detailed arguments.

Again, from \Cref{lem:truncate}, the coefficient satisfies $\norm{\mb x_0}{\infty}\le B$ with probability no smaller than $1-2\theta me^{-B^2/2}$. We write $\bar{\mb x}_0(i) = \mb x_0(i) \indicator{\abs{\mb x_0(i) } \le B}$, and let $\bar{\mb X}_0$ denote the circulant matrix generated by the truncated vector $\bar{\mb x}_0$. Denote
\begin{align}
\mb H_E &= \bb E\brac{\frac1m\mb X_0\diag\paren{\mb X_0^T\mb A^T\mb q}^{\circ2}\mb X_0^T},\\
\mb{\bar H}_E &= \bb E\brac{\frac1m\mb{\bar X}_0\diag\paren{\mb{\bar X}_0^T\mb A^T\mb q}^{\circ2}\mb{\bar X}_0^T},
\end{align}
then
\begin{align}
&\prob{\norm{\frac{\mb X_0\diag\paren{\mb X_0^T\mb\zeta}^{\circ2}\!\!\mb X_0^T}m\!-\!\mb H_E
}2\!\!\ge c\theta(1\!-\theta)\!\norm{\mb\zeta}4^4}\nonumber\\
&\le\prob{\norm{\frac{\mb{\bar X}_0\diag\paren{\mb{\bar X}_0^T\mb\zeta}^{\circ2}\!\!\mb{\bar X}_0^T}m\!-\!\mb H_E}2\!\!\ge c\theta(1\!-\theta)\!\norm{\mb\zeta}4^4}+2\theta me^{-B^2/2}+2m\exp\paren{-\frac34\theta k\log{m}}
\end{align}
while via triangle inequality, 
\begin{align}
\lefteqn{\norm{\frac1m\mb{\bar X}_0\diag\paren{\mb{\bar X}_0^T\mb A^T\mb q}^{\circ2}\mb{\bar X}_0^T-\mb H_E}2}\nonumber\\
&\le\norm{\frac1m\mb{\bar X}_0\diag\paren{\mb{\bar X}_0^T\mb A^T\mb q}^{\circ2}\mb{\bar X}_0^T - \mb{\bar H}_E }2+\norm{\mb{\bar H}_E -\mb H_E}2.
\end{align}

\noindent{\bf Truncation Level.}
Next, we choose a large enough entry-wise truncation level $B$ such that the expectation of the Hessian $\bb E\brac{\mb X_0\diag\paren{\mb X_0^T\mb A^T\mb q}^{\circ2}\mb X_0^T}$ is close to that of its truncation $\bb E\brac{\mb{\bar X}_0\diag\paren{\mb{\bar X}_0^T\mb A^T\mb q}^{\circ2}\mb{\bar X}_0^T}$.
Moreover, we introduce following events notation
\begin{align}
\event_i\doteq\set{\norm{\mb x_i}{\infty}>B\;\cup\;\norm{\mb x_i}0>4\theta k\log{m}},
\end{align}
then
\begin{align}
\lefteqn{\norm{\mb{\bar H}_E -\mb H_E}2}\nonumber\\
&=\norm{\bb E\brac{\frac1m\sum_i\innerprod{\mb x_i}{\mb\zeta}^2\mb x_i\mb x_i^T\cdot\mb 1_{\mb E_i}}}F\\
&\le\frac1m\sum_i\norm{\bb E\brac{\innerprod{\mb x_i}{\mb\zeta}^2\mb x_i\mb x_i^T\cdot\mb 1_{\mb E_i}}}F\\
&\le\frac1m\sum_i\paren{\bb E\brac{\norm{\innerprod{\mb x_i}{\mb\zeta}^2\mb x_i\mb x_i^T}F^2}\cdot\bb E\brac{\mb 1_{\mb E_i}}}^{1/2}\\
&\le\paren{\bb E\brac{\norm{\mb x_i}2^8}}^{1/2}\sqrt{\bb E\brac{\mb 1_{\norm{\mb x_i}{\infty}>B}}+\bb E\brac{\mb 1_{\norm{\mb x_i}0>4\theta k\log{m}}}}\\
&\le50k^2\sqrt{4\theta ke^{-B^2/2}+\exp\paren{-\tfrac34\theta k\log{m}}}
\end{align}

By setting
\begin{align}
B\ge C'\log^{1/2}\paren{ \frac{k^7}{\theta\paren{1-\theta}^2} }
\end{align}
we have
\begin{align}
\theta ke^{-B^2/2}\le c'\theta^2\paren{1-\theta}^2\frac{\norm{\mb\zeta}4^8}{k^4}
\end{align}
In addition, whenever 
\begin{align}
 \theta k\ge\frac4{3\log{m}}\log\paren{\frac{400^2k^4}{c^2\theta^2\paren{1-\theta}^2\norm{\mb\zeta}4^8}}, 
\end{align}
we have
\begin{align}
\exp\paren{-\tfrac34\theta k\log{m}} \le\frac12\paren{\frac{c\theta(1\!-\theta)}{100}}^2\frac{\norm{\mb\zeta}4^8}{k^4}.
\end{align}
Hence,
\begin{align}
\sqrt{4\theta ke^{-B^2\!/\!2}\!+\exp\!\paren{-\tfrac34\theta k\log{m}}}\!\le\! \frac{c\theta \paren{1-\theta}}{100 k^2}\! \norm{\mb\zeta}4^4.
\end{align}
Therefore, we can obtain that
\begin{align}
\norm{\mb{\bar H}_E -\mb H_E }2\le\frac{c}2\theta\paren{1-\theta}\norm{\mb\zeta}4^4
\end{align}
always holds, hence
\begin{align}
\lefteqn{\prob{\norm{\frac{\mb{\bar X}_0\diag\paren{\mb{\bar X}_0^T\mb\zeta}^{\circ2}\!\!\mb{\bar X}_0^T}m\!-\!\mb{H}_E }2\!\!\ge c\theta(1\!-\theta)\norm{\mb\zeta}4^4}}\nonumber\\
&\le\!\prob{\norm{\frac{\mb{\bar X}_0\diag\paren{\mb{\bar X}_0^T\mb\zeta}^{\circ2}\!\!\mb{\bar X}_0^T}m \!-\! \mb{\bar H}_E }2\!\!\ge \frac{c}2\theta(1\!-\theta)\norm{\mb\zeta}4^4}.
\end{align}

\noindent{\bf Independent Sub-matrices.}  As we did in \Cref{lem:grad_sub}, we remove the dependence in $\mb X_0$ by sampling every $2k-1$ column such that 
\begin{align}
\mb X_0 \mb\Pi=\brac{\mb X_1,\mb X_2,\cdots, \mb X_{2k-1}}, 
\end{align}
where
\begin{align}
\mb X_i=\brac{\mb x_{i},\mb x_{i+\paren{2k-1}},\cdots,\mb x_{i+\paren{m-2k-1}}},
\end{align}
and $\mb\Pi$ is a certain permutation of the columns of $\mb X_0$.

Applying Bernstein inequality for matrix variables as in \Cref{lem:bernstein_matrix}, with $\mb M_i=\innerprod{\mb{\bar x}_i}{\mb A^T\mb q}^2\mb{\bar x}_i\mb{\bar x}_i^T\in\R^{\paren{2k-1}\times\paren{2k-1}}$. Since
\begin{align}
\norm{ \mb M_i }{2}&=\norm{\innerprod{\mb{\bar x}_i}{\mb A^T\mb q}^2\mb{\bar x}_i\mb{\bar x}_i^T}2\\
&\le\norm{\mb{\bar x}_i}2^4\\
&\le 4B^4k^2
\end{align}
and
\begin{align}
\norm{ \bb E\brac{  \mb M_i \mb M_i^* } }{}&=\norm{ \bb E\brac{ \mb M_i^* \mb M_i } }{}\\
&= \norm{ \bb E\brac{  \innerprod{\mb{\bar x}_i}{\mb A^T\mb q}^4\mb{\bar x}_i\mb{\bar x}_i^T \mb{\bar x}_i\mb{\bar x}_i^T} }{} \\
&= \norm{ \bb E\brac{ \innerprod{\mb{\bar x}_i}{\mb \zeta}^4\norm{\mb{\bar x}_i}2^2\mb{\bar x}_i\mb{\bar x}_i^T} }{} \\
&\le \bb E\brac{ \innerprod{\mb{\bar x}_i}{\mb \zeta}^4\norm{\mb{\bar x}_i}2^4 },
\end{align}
we obtain the following upper bound:
\begin{align}
&\bb E\brac{ \innerprod{\mb{\bar x}_i}{\mb \zeta}^4\norm{\mb{\bar x}_i}2^4 }\nonumber\\
&=\bb E\brac{ \sum_{j_1,j_2}^{2k-1}\mb{\bar x}_i\paren{j_1}^2\mb{\bar x}_i\paren{j_2}^2 \sum_{j_3,\cdots,j_6} \prod_{\ell=3}^6\mb{\bar x}_i\paren{j_\ell}\mb\zeta_{j_\ell}}\\
&=3\bb E\brac{ \sum_{j_1,j_2}^{2k-1}\mb{\bar x}_i\paren{j_1}^2\mb{\bar x}_i\paren{j_2}^2\!\!\sum_{j_3\neq j_4}\!\!\mb{\bar x}_i\paren{j_3}^2\mb\zeta^2_{j_3}\mb{\bar x}_i\paren{j_4}^2\mb\zeta^2_{j_4}}\nonumber\\
&\quad+ \bb E\brac{ \sum_{j_1,j_2}^{2k-1}\mb{\bar x}_i\paren{j_1}^2\mb{\bar x}_i\paren{j_2}^2\cdot \sum_{j_3}\mb{\bar x}_i\paren{j_3}^4\mb\zeta^4_{j_3}}\\
&=3\bb E\brac{ \sum_{\substack {j_1\neq j_2\\ \neq j_3\neq j_4}}\mb{\bar x}_i\paren{j_1}^2 \mb{\bar x}_i\paren{j_2}^2\mb{\bar x}_i\paren{j_3}^2\mb\zeta^2_{j_3}\mb{\bar x}_i\paren{j_4}^2\mb\zeta^2_{j_4}}\nonumber\\
&\quad+ 3\bb E\brac{ \sum_{j_1\neq j_2\neq j_3}\mb{\bar x}_i\paren{j_1}^4 \mb{\bar x}_i\paren{j_2}^2\mb\zeta^2_{j_2}\mb{\bar x}_i\paren{j_3}^2\mb\zeta^2_{j_3}}\nonumber\\
&\quad+ 6\bb E\brac{ \sum_{j_1\neq j_2} \mb{\bar x}_i\paren{j_1}^6 \mb\zeta^2_{j_1} \mb{\bar x}_i\paren{j_2}^2 \mb\zeta^2_{j_2}}\nonumber\\
&\quad+ 6\bb E\brac{ \sum_{j_1\neq j_2\neq j_3}\mb{\bar x}_i\paren{j_1}^2 \mb{\bar x}_i\paren{j_2}^4 \mb\zeta^2_{j_2}\mb{\bar x}_i\paren{j_3}^2 \mb\zeta^2_{j_3}}\nonumber\\
&\quad+ 6\bb E\brac{ \sum_{j_1\neq j_2} \mb{\bar x}_i\paren{j_1}^4 \mb\zeta^2_{j_1}\mb{\bar x}_i\paren{j_2}^4 \mb\zeta^2_{j_2}}\nonumber\\
&\quad+ 2\bb E\brac{ \sum_{j_1\neq j_2}\mb{\bar x}_i\paren{j_1}^2\mb{\bar x}_i\paren{j_2}^6\mb\zeta^4_{j_2}}\nonumber\\
&\quad+\bb E\brac{ \sum_{j_1\neq j_2\neq j_3}\mb{\bar x}_i\paren{j_1}^2\mb{\bar x}_i\paren{j_2}^2 \mb{\bar x}_i\paren{j_3}^4\mb\zeta^4_{j_3}}\nonumber\\
&\quad+\bb E\brac{ \sum_{j_1\neq j_2}\mb{\bar x}_i\paren{j_1}^4\mb{\bar x}_i\paren{j_2}^4\mb\zeta^4_{j_2}}\nonumber\\
&\quad + \bb E\brac{ \sum_{j}\mb{\bar x}_i\paren{j}^8\mb\zeta^4_{j}}\\
&\le \paren{105\theta + 18\theta^2 k + 60\theta^2 k + 12\theta^3 k^2 }\norm{\mb\zeta}4^4\nonumber\\
&\quad+ 3\paren{21\theta^2 + 30\theta^2 + 4\theta^4k^2 + 12\theta^2k}\norm{\mb\zeta}2^4\\
&\le C\theta^3 k^2
\end{align}
Assuming $\theta m\ge 1$, hence
\begin{align}
\sigma^2 = C \theta^3k m.
\end{align}

Setting $t = \frac{c}2 \frac{ \theta\paren{1- \theta} m\norm{ \mb \zeta }4^4}{2k-1}$ in Matrix Bernstein gives
\begin{align}
\bb P\brac{ \norm{\mb{\bar X}_i \paren{\mb{\bar X}_i^T \mb \zeta}^{\circ 3} - \bb E\brac{ \cdot } }{2} > t  }\le 8k \exp\paren{ \frac{- t^2 / 2 }{ C {\theta^3 k} m + C' B^4 \theta^2 k^2 t } },
\end{align}
we therefore obtain 
\begin{align}
&\bb P \brac{ \norm{\frac{\mb{\bar X}_i\diag\paren{\mb{\bar X}_i^T\mb\zeta}^{\circ2}\!\!\mb{\bar X}_i^T}m \!-\! \bb E\brac{ \cdot } }{2}\!\!\! > c \frac{\theta \paren{1- \theta }\| \mb \zeta \|_4^6}{2k-1} }\nonumber \\
& \;\le\; 8 k \exp\paren{ - \frac{ c m \paren{1 - \theta }^2 \norm{ \mb \zeta }4^8 }{ \theta k^3 +\theta \paren{1-\theta}B^4 k^3\norm{\mb\zeta}4^4} }.
\end{align}

\noindent{\bf $\eps$-Net Covering} To obtain a probability bound for all $\mb q\in\bb S^{k-1}$, we choose a set of $\mb\zeta_n=\mb A^T\mb q_n$ with $n=1,\cdots,N$. 
Since for any $\mb q,\mb q'\in\bb S^{k-1}$ and $\mb \zeta'=\mb A^T\mb q'$, we have
\begin{align}
&\norm{\frac{\mb{\bar X}_i\diag\paren{\mb{\bar X}_i^T\mb \zeta}^{\circ2}\!\mb{\bar X}_i^T}m-\frac{\mb{\bar X}_i\diag\paren{\mb{\bar X}_i^T\mb\zeta'}^{\circ2}\!\mb{\bar X}_i^T}m}2\nonumber\\
&= \frac1m\norm{\mb{\bar X}_i\diag\brac{\paren{\mb{\bar X}_i^T\mb \zeta}^{\circ2}\!\!-\paren{\mb{\bar X}_i^T\mb\zeta'}^{\circ2}}\mb{\bar X}_i^T}2\\
&\le \frac{\norm{\mb{\bar X}_i}2^2}{m}\norm{\diag\brac{\paren{\mb{\bar X}_i^T\mb \zeta}^{\circ2}\!\!-\paren{\mb{\bar X}_i^T\mb\zeta'}^{\circ2}}}2\\
&\le \frac{\norm{\mb{\bar X}_i}2^2}m\norm{\mb{\bar X}_i^T\mb \zeta+\mb{\bar X}_i^T\mb \zeta'}{\infty}\norm{\mb{\bar X}_i^T\mb \zeta-\mb{\bar X}_i^T\mb\zeta'}{\infty}\\
&\le L\norm{\mb q-\mb q'}2
\end{align}
Then the Lipschitz constant $L$ is upper bounded by
\begin{align}
L&\le\frac{\norm{\mb{\bar X}_i}2^2}m\norm{\mb{\bar X}_i^T\mb A^T}2\!\!\paren{\norm{\mb{\bar X}_i^T\mb \zeta}{\infty}\!\!+\norm{\mb{\bar X}_i^T\mb \zeta'}{\infty}}\\
&\le\frac2m\norm{\mb{\bar X}_i}2^4\\
&\le 8\theta^2mB^4.
\end{align}
With triangle inequality, we have
\begin{align}
&\norm{\frac{\mb{\bar X}_i\diag\paren{\mb{\bar X}_i^T\mb\zeta}^{\circ2}\mb{\bar X}_i^T}m-\bb E\brac{\cdot}}2\nonumber\\
&\le\norm{\frac{\mb{\bar X}_i\diag\paren{\mb{\bar X}_i^T\mb\zeta}^{\circ2}\!\!\mb{\bar X}_i^T}m-\frac{\mb{\bar X}_i\diag\paren{\mb{\bar X}_i^T\mb\zeta_n}^{\circ2}\!\!\mb{\bar X}_i^T}m}2\nonumber\\
&\;+\norm{\bb E\!\brac{\!\frac{\mb{\bar X}_i\diag\paren{\mb{\bar X}_i^T\mb\zeta}^{\circ2}\!\!\mb{\bar X}_i^T\!}m\!}\!\!-\!\bb E\!\brac{\!\frac{\mb{\bar X}_i\diag\paren{\mb{\bar X}_i^T\mb\zeta_n}^{\circ2}\!\!\mb{\bar X}_i^T\!}m\!}}2\nonumber\\
&\;+\norm{\frac{\mb{\bar X}_i\diag\paren{\mb{\bar X}_i^T\mb\zeta_n}^{\circ2}\!\!\mb{\bar X}_i^T\!}m\!-\!\bb E\!\brac{\frac{\mb{\bar X}_i\diag\paren{\mb{\bar X}_i^T\mb\zeta_n}^{\circ2}\!\!\mb{\bar X}_i^T\!}m}}2\\
&\le\norm{\frac{\mb{\bar X}_i\paren{\mb{\bar X}_i^T\mb\zeta_n}^{\circ2}\mb{\bar X}_i^T}m-\bb E\!\brac{\frac{\mb{\bar X}_i\paren{\mb{\bar X}_i^T\mb\zeta_n}^{\circ2}\mb{\bar X}_i^T}m}}2 + 2L\eps 
\end{align}
Next, we are going to choose the $\eps$-net to cover the sphere of $\mb q$ with
\begin{align}
\eps=\frac{c}4\frac{\theta\paren{1-\theta}}{\paren{2k-1}L}\min_{\mb q\in\bb S^{k-1}}\norm{ \mb\zeta}4^4,
\end{align}
hence the number of samples $N$ is bounded by
\begin{align}
N&=\paren{\frac3{\eps}}^k\\
&\le\exp\paren{-k\ln\eps}\\
&\le C\exp\brac{k\log\paren{\frac{\theta B^4k^2m}{1-\theta}}}.
\end{align}

For $n=1,\cdots,N$, denote
\begin{align}
P_i\paren{\mb q_n}=\prob{\norm{\frac{\mb{\bar X}_i\diag\paren{\mb{\bar X}_i^T\mb\zeta_n}^{\circ2}\!\!\mb{\bar X}_i^T\!}m\!-\!\bb E\!\brac{\cdot}}2\!\!\ge \!\frac{c\theta(1\!-\theta)\!\norm{\mb\zeta_n}4^4}{4(2k-1)}}
\end{align}
together with union bound over all $\mb q_n$, we obtain 

\begin{align}
&\prob{\!\sup_{\mb q\in\hat{\mc R}_{2C_\star}}\!\!\!\!\frac{\norm{\frac{\mb{\bar X}_i\diag\paren{\mb{\bar X}_i^T\mb\zeta}^{\circ2}\!\!\mb{\bar X}_i^T}m-\bb E\brac{\cdot}}2}{\norm{\mb\zeta}4^4}\!\ge \frac{c\theta\paren{1-\theta}}{2(2k-1)}}\nonumber\\
&\le \sum_{\mb q_n\in\hat{\mc R}_{2C_\star}}P_i\paren{\mb q_n}\\
&\le N\max_{\mb q_n\in\hat{\mc R}_{2C_\star}}P_i\paren{\mb q_n}\\
&\le 8k\!\!\!\sup_{\mb q\in\hat{\mc R}_{2C_\star}}\!\!\!\! \exp\paren{-  \frac{ c m \paren{1 - \theta }^2 \norm{ \mb \zeta }4^{8} }{ \theta k^3+ \theta \paren{1-\theta}B^4 k^3 \norm{\mb\zeta}4^4} }\exp\paren{ k\ln\paren{\frac3{\eps}}}.
\end{align}

Hence 

\begin{align}
\lefteqn{\prob{\!\sup_{\mb q\in\hat{\mc R}_{2C_\star}}\!\!\frac{\norm{\frac{\mb{\bar X}_0\diag\paren{\mb{\bar X}_0^T\mb\zeta}^{\circ2}\!\!\mb{\bar X}_0^T\!}m\!-\bb E\brac{\cdot}}2}{\norm{\mb\zeta}4^4}\!\ge \frac{c}2\theta(1\!-\theta)}}\nonumber\\
&\le\sum_{i}\prob{\!\!\sup_{\mb q\in\hat{\mc R}_{2C_\star}}\!\!\!\!\frac{\norm{\!\frac{\mb{\bar X}_i\diag\paren{\mb{\bar X}_i^T\mb\zeta}^{\circ2}\!\!\mb{\bar X}_i^T\!}m\!-\!\bb E\brac{\cdot}}2}{\norm{\mb\zeta}4^4}\!\ge \frac{c\theta(1\!-\theta)}{2(2k\!-\!1)}}\\
&\le\paren{2k-1}\max_{i}\prob{\sup_{\mb q\in\hat{\mc R}_{2C_\star}}\!\!\frac{ \norm{\frac{\mb{\bar X}_i\diag\paren{\mb{\bar X}_i^T\mb\zeta}^{\circ2}\!\!\mb{\bar X}_i^T}m-\bb E\brac{\cdot}}2}{\norm{\mb\zeta}4^4}\ge \frac{c\theta\paren{1-\theta}}{2(2k-1)}}\\
&\le 16k^2\!\!\sup_{\mb q\in\hat{\mc R}_{2C_\star}}\!\!\!\!\exp\paren{-  \frac{ c' m \paren{1 - \theta }^2 \norm{ \mb \zeta }4^8 }{  \theta k^3+\theta (1\!-\theta)B^4 k^3\norm{\mb\zeta}4^4}  }\exp\paren{k\ln\paren{\frac3{\eps}} }
\end{align}

Therefore, by taking
\begin{align}
m&\ge \frac{C\theta}{\paren{1-\theta}^2}\min\set{\paren{2C_\star\mu\kappa^2}^{-4/3}\!\!,k^2}k^4\log^3k\\
&\ge C'\theta k\log\paren{\frac{\theta kmB^4}{\paren{1-\theta}\norm{\mb\zeta}4^4}} \frac{ k^3+ \paren{1-\theta}B^4 k^3\norm{\mb\zeta}4^4}{  \paren{1 - \theta }^2 \norm{ \mb \zeta }4^8 }
\end{align}
and adding up failure probability, we obtain
\begin{align}
&\norm{\frac1m\mb X_0\diag\paren{\mb X_0^T\mb A^T\mb q}^{\circ2}\mb X_0^T-\bb E\brac{\cdot}}2\nonumber\\
&\qquad\le c\theta\paren{1-\theta}\norm{\mb A^T\mb q}4^4
\end{align}
with probability no smaller than $1-c_1\exp\paren{-k}-c_2\theta\paren{1-\theta}^2k^{-4}-c_3\exp\paren{-\theta k}$.
\end{proof}

\section{Tools}
\begin{lemma}[Moments of the Gaussian Random Variables] \label{lem:guassian_moment}
If $X \sim \mc N\left(0, \sigma^2\right)$, then it holds for all integer $p \geq 1$ that
\begin{align}
\expect{\abs{X}^p} &= \sigma^p \paren{p -1}!! \brac{ \sqrt{\frac{2}{\pi}} \indicator{p\; \text{odd}}+\indicator{p\; \text{even}} } \\
&\leq \sigma^p \paren{p -1}!!. 
\end{align}
\end{lemma}

\begin{lemma}[Moments of the $\chi^2$ Random Variables] \label{lem:chi_sq_moment}
If $X \sim \mc \chi^2\paren{n}$, then it holds for all integer $p \geq 1$,
\begin{align}
\expect{X^p} &= 2^p \frac{\Gamma\paren{p + n/2}}{\Gamma\paren{n/2}} \\
&=  \prod_{k=1}^p (n+2k-2)\le p!(2n)^p/2 
\end{align}
\end{lemma}

\begin{lemma}[Moments of the $\chi$ Random Variables] \label{lem:chi_moment}
If $X \sim \mc \chi\paren{n}$, then it holds for all integer $p \geq 1$,
\begin{align}
\expect{X^p} = 2^{p/2} \frac{\Gamma\paren{p/2 + n/2}}{\Gamma\paren{n/2}} \le p!! n^{p/2}. 
\end{align}
\end{lemma}

\begin{lemma}[Moment-Control Bernstein's Inequality for Scalar RVs, Theorem 2.10 of~\cite{foucart2013mathematical}] \label{lem:mc_bernstein_scalar}
Let $X_1, \dots, X_p$ be i.i.d. real-valued random variables. Suppose that there exist some positive number $R$ and $\sigma^2$ such that
\begin{align*}
\expect{\abs{X_k}^m} \leq \frac{m!}{2} \sigma^2 R^{m-2}, \; \; \text{for all integers}\; m \ge 2. 
\end{align*} 
Let $S \doteq \frac{1}{p}\sum_{k=1}^p X_k$, then for all $t > 0$, it holds  that 
\begin{align}
\prob{\abs{S - \expect{S}} \ge t} \leq 2\exp\left(-\frac{pt^2}{2\sigma^2 + 2Rt}\right).   
\end{align}
\end{lemma}
\begin{corollary}[Moment-Control Bernstein's Inequality for Vector RVs, Corollary A.10 of \cite{SQW15-pp}] 
\label{cor:vector-bernstein} 
Let $\mb x_1, \dots, \mb x_p \in \R^d$ be i.i.d. random vectors. Suppose there exist some positive number $R$ and $\sigma^2$ such that
\begin{equation*}
\bb E\left[ \norm{\mb x_k }{}^m \right] \;\le\; \frac{m!}{2} \sigma^2R^{m-2}, \quad \text{for all integers $m \ge 2$}. 
\end{equation*}
Let $\mb s = \frac{1}{p}\sum_{k=1}^p \mb x_k$, then for any $t > 0$, it holds that
\begin{align}
\bb P\brac{\norm{\mb s - \bb E\brac{\mb s}}{} \geq t} \; \leq \; 2(d+1)\exp\paren{-\frac{pt^2}{2\sigma^2+2Rt}}.
\end{align}
\end{corollary}

\begin{lemma}[Moment-Control Bernstein's Inequality for Matrix RVs, Theorem 6.2 of~\cite{tropp2012user}] 
\label{lem:mc_bernstein_matrix}
Let $\mb X_1, \dots, \mb X_p\in \R^{d \times d}$ be i.i.d. random, symmetric matrices. Suppose there exist some positive number $R$ and $\sigma^2$ such that
\begin{align}
\expect{\mb X_k^m} &\preceq \frac{m!}{2} \sigma^2 R^{m-2} \mb I ,\\ 
-\expect{\mb X_k^m} &\preceq \frac{m!}{2} \sigma^2 R^{m-2} \mb I. 
\label{eqn:bern-moment-conditions}
\end{align}
for all integers $m \ge 2$.
Let $\mb S \doteq \frac{1}{p} \sum_{k = 1}^p \mb X_k$, then for all $t > 0$, it holds that 
\begin{align}
\prob{\norm{\mb S - \expect{\mb S}}{} \ge t} \le 2d\exp\paren{-\frac{pt^2}{2\sigma^2 + 2Rt}}.
\end{align}
\end{lemma}

\begin{lemma}[Bernstein's Inequality for Uncentered Matrix RVs]
\label{lem:bernstein_matrix}
The matrix Bernstein inequality states that for independent random matrices $\mb M_1, \dots, \mb M_n \in \R^{d_1 \times d_2}$, if 
\begin{equation} \label{eqn:mb-var}
\sigma^2 = \max\set{ \norm{ \sum_{i = 1}^n \bb E[ \mb M_i \mb M_i^* ] }{}, \norm{ \sum_{i = 1}^n \bb E[ \mb M_i^* \mb M_i ] }{} },
\end{equation}
and 
\begin{equation} \label{eqn:mb-as}
\norm{ \mb M_i }{2} \le R \qquad \text{a.s.},
\end{equation}
then
\begin{equation}
\bb P\brac{ \norm{ \sum_i \mb M_i - \bb E\brac{ \cdot} }{} > t } \le  (d_1 + d_2 ) \exp\paren{ \frac{ - t^2 / 2 }{ \sigma^2 + 2 R t / 3 } }.
\end{equation} 
\end{lemma}

\begin{proof}
For zero mean random matrices 
\begin{align}
\mb M_1-\bb E \mb M_1, \dots, \mb M_n- \bb E \mb M_n\in \R^{d_1 \times d_2}, 
\end{align}
we have that
\begin{equation}
\| \mb M_i - \bb E \mb M_i \|_2 \le 2 R,
\end{equation}
and
\begin{align}
\mb 0 &\preceq \sum_{i = 1}^n \bb E[ ( \mb M_i - \bb E \mb M_i ) ( \mb M_i - \bb E \mb M_i )^* ] \preceq \sum_{i = 1}^n \bb E[ \mb M_i \mb M_i^* ] ,\\
\mb 0 &\preceq \sum_{i = 1}^n \bb E[ ( \mb M_i - \bb E \mb M_i )^* ( \mb M_i - \bb E \mb M_i )]\preceq \sum_{i = 1}^n \bb E[ \mb M_i^* \mb M_i ].
\end{align} 

Plugging corresponding quantities back to Theorem 1.6 of \cite{tropp2012user}, we obtain that 
\begin{equation}
\bb P\brac{ \norm{ \sum_i \mb M_i - \bb E\brac{ \cdot}}{} > t }\le  (d_1 + d_2 ) \exp\paren{ \frac{ - t^2 / 2 }{ \sigma^2 + 2R t/3 } }.
\end{equation} 
\end{proof}

\end{document}